\documentclass[11pt,reqno]{amsart}

\usepackage{amsmath, amssymb, amsthm, amsfonts, cite, enumitem}

\usepackage[margin=25mm]{geometry}
\usepackage{multirow}
\usepackage{multicol}
\usepackage[dvipsnames,svgnames,table]{xcolor}
\usepackage{graphicx}
\usepackage{epstopdf}
\usepackage{ulem}
\usepackage{hyperref}
\usepackage{pstricks}

\usepackage{mathtools}
\newtheorem{theorem}{Theorem}[section]
\newtheorem{proposition}[theorem]{Proposition}
\newtheorem{corollary}[theorem]{Corollary}
\newtheorem{lemma}[theorem]{Lemma}
\newtheorem{remark}{Remark}
\newtheorem{definition}{Definition}
\numberwithin{equation}{section}

\newcommand{\vo}{V_{\omega,1}}
\newcommand{\vt}{V_{\omega,2}}

\newcommand{\pno}{p_{n,1}}
\newcommand{\pnt}{p_{n,2}}
\newcommand{\ppnt}{p_{n+1,2}}
\newcommand{\ppru}{\mathcal{P}}
\newcommand{\qppru}{\mathcal{Q}}
\newcommand{\pru}{{\mathcal{P}'}}
\newcommand{\qpru}{{\mathcal{Q'}}}
\newcommand{\tT}{{T'}}
\newcommand{\tA}{{A'}}
\newcommand{\tB}{{B'}}

\newcommand{\transn}{T_{\omega,n}}

\newcommand{\prodtransn}{{\bf T}_{\omega,n}}
\newcommand{\prodtransnm}{{\bf T}_{\omega,n-1}}
\newcommand{\ttransn}{T'_{\omega,n}}

\newcommand{\deltap}{\delta^+}
\newcommand{\deltam}{\delta^-}
\newcommand{\D}{D_\omega}

\newcommand{\Z}{\mathbb{Z}}
\newcommand{\R}{\mathbb{R}}
\newcommand{\C}{\mathbb{C}}
\newcommand{\n}{\mathbb{N}}

\newcommand{\p}{\mathbb{P}}
\newcommand{\esp}{\mathbb{E}}

\newcommand{\bra}{\langle}
\newcommand{\ket}{\rangle}

\newcommand{\norm}[1]{\left\lVert #1 \right\rVert}

\newcommand{\e}{\mathrm{e}}
\newcommand{\h}{\mathcal{H}}
\newcommand{\vp}{\phi}% eigenfunction
\newcommand{\X}{{\bf X}}
\newcommand{\K}{\mathcal{K}}

\def \simless {\mathbin{\lower 3pt\hbox{$\rlap{\raise 5pt
              \hbox{$\char'074$}}\mathchar"7218$}}}

\author[O. Bourget, G. R. Moreno Flores and A. Taarabt]{Olivier Bourget$^1$, Gregorio R. Moreno Flores$^{2,*}$ and Amal Taarabt$^3$}
\date{}


\thanks{ AMS 2010 {\it subject classifications}. 82B44, 47B80}

\thanks{$^*$ Corresponding author}

\thanks{$^1$ $^2$ $^3$ Facultad de Matem\'aticas, Pontificia Universidad Cat\'olica de Chile.}

\thanks{$^1$  Partially supported by Fondecyt grant 1161732}

\thanks{$^2$  Partially supported by Fondecyt grant 1171257, N\'ucleo Milenio `Modelos Estoc\'asticos de Sistemas Complejos y Desordenados' and MATH Amsud `Random Structures and Processes in Statistical Mechanics'}

\thanks{$^3$  Partially supported by Fondecyt grant 11190084}

\address{Facultad de Matem\'aticas\\
Pontificia Universidad Cat\'olica de Chile\\
Vicu\~na Mackenna 4860, Macul\\
Santiago, Chile}

\email{bourget@mat.uc.cl, grmoreno@mat.uc.cl, amtaarabt@mat.uc.cl}

\title[Dirac Operators in a Decaying Potential]{One-dimensional Discrete Dirac Operators in a Decaying Random Potential I: Spectrum and Dynamics}

\begin{document}

\begin{abstract}
	We study the spectrum and dynamics of  a one-dimensional discrete Dirac operator in a random potential obtained by damping an i.i.d. environment with an envelope of type $n^{-\alpha}$ for $\alpha>0$.
	We recover all the spectral regimes previously obtained for the analogue Anderson model in a random decaying potential, namely: absolutely continuous spectrum in the super-critical region $\alpha>\frac12$; a transition from pure point to singular continuous spectrum in the critical region $\alpha=\frac12$; and pure point spectrum in the sub-critical region $\alpha<\frac12$.
	From the dynamical point of view, delocalization in the super-critical region follows from the RAGE theorem. 
	In the critical region, we exhibit a simple argument based on lower bounds on eigenfunctions showing that no dynamical localization can occur even in the presence of point spectrum. 
	Finally, we show dynamical localization in the sub-critical region by means of the fractional moments method and provide control on the eigenfunctions. 
	%Based again on lower bounds on eigenfunctions, we show that heavier stretched exponential moments diverge.
	
\end{abstract}

\maketitle

\tableofcontents

%%%%%%%%%%%%%%%%%%%%%%%%%%%%%%%%%%%%%%%%%%%%%%%%%%%%%%%%%%%%
%%%%%%%%%%%%%%%%%%%%%%%%%%%%%%%%%%%%%%%%%%%%%%%%%%%%%%%%%%%%
%%%%%%%%%%%%%%%%%%%%%%%%%%%%%%%%%%%%%%%%%%%%%%%%%%%%%%%%%%%%
%%%%%%%%%%%%%%%%%%%%%%%%%%%%%%%%%%%%%%%%%%%%%%%%%%%%%%%%%%%%
%%%%%%%%%%%%%%%%%%%%%%%%%%%%%%%%%%%%%%%%%%%%%%%%%%%%%%%%%%%%
%%%%%%%%%%%%%%%%%%%%%%%%%%%%%%%%%%%%%%%%%%%%%%%%%%%%%%%%%%%%
%%%%%%%%%%%%%%%%%%%%%%%%%%%%%%%%%%%%%%%%%%%%%%%%%%%%%%%%%%%%
%%%%%%%%%%%%%%%%%%%%%%%%%%%%%%%%%%%%%%%%%%%%%%%%%%%%%%%%%%%%
%%%%%%%%%%%%%%%%%%%%%%%%%%%%%%%%%%%%%%%%%%%%%%%%%%%%%%%%%%%%
%%%%%%%%%%%%%%%%%%%%%%%%%%%%%%%%%%%%%%%%%%%%%%%%%%%%%%%%%%%%
%%%%%%%%%%%%%%%%%%%%%%%%%%%%%%%%%%%%%%%%%%%%%%%%%%%%%%%%%%%%
%%%%%%%%%%%%%%%%%%%%%%%%%%%%%%%%%%%%%%%%%%%%%%%%%%%%%%%%%%%%
%%%%%%%%%%%%%%%%%%%%%%%%%%%%%%%%%%%%%%%%%%%%%%%%%%%%%%%%%%%%
%%%%%%%%%%%%%%%%%%%%%%%%%%%%%%%%%%%%%%%%%%%%%%%%%%%%%%%%%%%%
%%%%%%%%%%%%%%%%%%%%%%%%%%%%%%%%%%%%%%%%%%%%%%%%%%%%%%%%%%%%
%%%%%%%%%%%%%%%%%%%%%%%%%%%%%%%%%%%%%%%%%%%%%%%%%%%%%%%%%%%%
%%%%%%%%%%%%%%%%%%%%%%%%%%%%%%%%%%%%%%%%%%%%%%%%%%%%%%%%%%%%
%%%%%%%%%%%%%%%%%%%%%%%%%%%%%%%%%%%%%%%%%%%%%%%%%%%%%%%%%%%%
%%%%%%%%%%%%%%%%%%%%%%%%%%%%%%%%%%%%%%%%%%%%%%%%%%%%%%%%%%%%
%%%%%%%%%%%%%%%%%%%%%%%%%%%%%%%%%%%%%%%%%%%%%%%%%%%%%%%%%%%%
%%%%%%%%%%%%%%%%%%%%%%%%%%%%%%%%%%%%%%%%%%%%%%%%%%%%%%%%%%%%
%%%%%%%%%%%%%%%%%%%%%%%%%%%%%%%%%%%%%%%%%%%%%%%%%%%%%%%%%%%%
%%%%%%%%%%%%%%%%%%%%%%%%%%%%%%%%%%%%%%%%%%%%%%%%%%%%%%%%%%%%
%%%%%%%%%%%%%%%%%%%%%%%%%%%%%%%%%%%%%%%%%%%%%%%%%%%%%%%%%%%%

\section{Introduction}

The emergence of two-dimensional materials and the subsequent avalanche of related studies led to significant theoretical and experimental advances in condensed matter. 
The experimental discovery of 
graphene, a two-dimensional material composed of carbon atoms arranged in a honeycomb structure, was accomplished in 2004 \cite{NGM}. 
Due to its unusual and remarkable properties 
such as Klein tunnelling and finite minimal conductivity \cite{KNG},
graphene has attracted great attention in the recent years. It has emerged as a fascinating 
system for fundamental studies in condensed matter physics, as well as a promising candidate material 
for future applications in nanoelectronics and molecular devices. 
The simplest model for the dynamics of charge carriers in such a structure is the discrete Laplacian on a honeycomb lattice but
at low excitations energies this dynamics is actually described by a
massless two-dimensional Dirac operator \cite{CGPNG}. In particular, the Dirac cone structure gives graphene massless fermions, leading to half integer \cite{GMN,STZ}, fractional \cite{BGKSS,ADDLS} and
fractal \cite{DWM,HYY} quantum Hall Effects, in addition to ultrahigh carriers mobility \cite{BJS} and many other novel phenomena and properties.

In related contexts, Dirac operators have found various applications in electronic transport \cite{SAHR}, photonic structures \cite{RH1,RH2} and utracold matter in optical lattices \cite{BEG}.
Dirac operators are also used to study relativistic and non-relativistic electron localization phenomena as well as in investigations of electrical conduction
in disordered systems \cite{BDMRS,RB,dO-Bernoulli}.
Further electronic and transport studies of Dirac operators are hence  relevant for understanding the charge transport mechanism of a variety of physical systems.

Most of the spectral and dynamical aspects of random Dirac operators parallel well known results for the Anderson model obtained for instance in the works \cite{AM, Bu1, Bu2, BDFG, Car, CKM, DG, GK1, GT, GMP, JZ, KS}. Nonetheless, they require non-trivial adaptations of the proofs due to the matrix form and first order structure of the model and, in some situations, led to new behaviours. For the discrete model in an independent and identically distributed potential with a regular enough distribution, dynamical localization was obtained in \cite{dO-DL} by means of the fractional moment method of \cite{AM} in the three classical regimes: large disorder, near the band edge, and for all energies in the one-dimensional setting. The work \cite{dO-Bernoulli} considers the one-dimensional model with a Bernoulli potential. Dynamical localization is obtained for all energies in the massive case by means of a multiscale analysis which primary input is the positivity of the Lyapunov exponent (see \cite{CKM} for the Anderson model in this situation). In the massless case, the authors observe special configurations of the atoms of the potential which lead to zero Lyapunov exponents and transport for certain energies, a phenomenon which is not encountered in the Anderson model (see also \cite{dO-letter} and, for related phenomena in different contexts, see\cite{dBG, HS}). 
The work \cite{dO-LB} establishes dynamical lower bounds in one dimension in the spirit of \cite{GKT}.
The very recent work \cite{BCZ} establishes band edge localization for a continuous random Dirac-like operator under an open gap assumption.

%Some works on dynamical localization for discrete
%Anderson Dirac models have been done in
%\cite{dO-sparse-I,dO-DL,dO-LB,dO-Bernoulli,
%dO-letter,dO-sparse-II}.
%
%
%Even though localization is well established for one-dimensional Anderson models \cite{Car,GMP,KS}, random Dirac operator still are a challenge where transition localization-delocalization is expected for small disorder. 
%Some results have been known on the Bethe lattice and tree graphs \cite{Kl,ASW,FHS}.

Even though localization is well established for discrete random Dirac operators, the existence of continuous spectrum is an open question. In the Anderson model, absolutely continuous spectrum was shown to exist on tree graphs \cite{Kl, ASW, FHS} but there are still no available results in this direction on the lattice. 
A delocalization-localization transition has been proved for the related random Landau Hamiltonians where non-trivial transport occurs near Landau levels \cite{GKT}. 
To understand how the absolutely continuous spectrum can survive the addition of disorder in the Anderson model, it has been proposed to modulate the random potential by a decaying envelope, a point of view that has been followed since at least the work \cite{Kr} where extended states were obtained. Subsequent works in this direction include \cite{B1, B2, FGKM}. In one-dimension, the model was shown to display a rich phase diagram with different kinds of spectrum arising for different values of the parameters \cite{DSS, D, KLS}. From the dynamical perspective, dynamical localization was shown in \cite{Si82} for slowly decaying potentials while transport was observed for critical rate of decay in \cite{GKT}. 
%Sparse potentials were also considered in \cite{KLS}.

%In order to understand this transition
%phenomenon and how continuous spectrum could persist in spite of the randomness, one consider a decaying envelope on the ergodic random potential \cite{FGKM,KKO,Kr,Si82}.
%To define the generic model, let 
%\begin{equation*}
% H_{\omega,\lambda,\alpha}=H_0 +\lambda\gamma_\alpha  V_\omega \quad \mathrm{on}\quad \mathrm{L}^2(\R^d)\quad \mathrm{or}\quad \ell^2(\Z^d),
%\end{equation*}
%where $H_0$ is a background operator,
%$\lambda>0$ is the disorder parameter, $\gamma_\alpha$ is the envelope function
%$\gamma_\alpha(x)=\langle x\rangle^{-\alpha}$ for $\alpha\ge0$
%where $\langle x\rangle=\sqrt{1+|x|^2}$,
%and $V_\omega$ is a random potential. 
%This Anderson model is well studied in dimension one \cite{Si82,KRS} and a \textit{critical} value of $\alpha=\frac12$ gives rise to a transition from a point spectrum to continuous and singular spectrum.
%Localization is less understood for random Dirac operators (RDO).

In this work, we propose to follow this perspective by studying the one-dimensional Dirac operator in a decaying random potential. In a related spirit, sparse potentials were considered in \cite{dO-sparse-I, dO-sparse-II} but the model considered here has been untouched so far. 
Our results include
\begin{enumerate}
	\item[1.-] the nature of the spectrum depending on the decay rate of the potential,
	
	\item[2.-] transport for critically decaying potentials, and
	
	\item[3.-] dynamical localization for slowly decaying potentials.
\end{enumerate}
From the technical point of view, we follow the martingale approach of \cite{KLS} to study the spectrum of the operator. This technique relies on a decomposition of the Pr\"ufer transform including martingales terms which can be estimated by probabilistic arguments. To obtain such a decomposition, we introduce a novel Pr\"ufer transform for the Dirac operator leading to an explicit recursion on the complex plane which is in turn suitable for a martingale analysis. This transform is closer in spirit to the one introduced in \cite{KRS} for the Anderson model and differs from the one used in \cite{dO-sparse-I} in the context of the Dirac operator with sparse potentials. It has the advantage that the disorder variables are nicely factorized into linear and quadratic terms only.

Our proof of delocalization for critically decaying potentials is based on lower bounds on eigenfunctions and seems to be novel. It is much less quantitative than the bounds obtained in \cite{GKT} for the Anderson model but has the advantage to be very simple.

We prove dynamical localization for slowly decaying potentials following the fractional moment method of \cite{AM} by relating the fractional moments of the Green's function to estimates the norm of transfer matrices. The corresponding result for the Anderson model in a decaying random media was obtained in \cite{Si82} by means of the Kunz-Souillard method \cite{KS,Bu1,Bu2,DG}. It is likely that a suitable adaptation of these techniques for Dirac operators could be applied in our context. We choose this different perspective as it relies directly on the analysis of the Pr\"ufer transform that we developed to study the spectrum of the operator and allows us to consider random variables with unbounded densities under some mild regularity assumptions, unlike the Kunz-Souillard method which assumes bounded densities. In a related context, our approach was also successful in providing a proof of dynamical localization for the continuum Anderson model with a slowly decaying random potential \cite{BMT02}.
In addition, the Pr\"ufer transform analysis provides lower bounds on eigenfunctions which we use to show that certain stretched exponential moments blow up, a fact that in some sense quantifies the strength of localization.

%%%%%%%%%%%%%%%%%%%%%%%%%%%%%%%%%%%%%%%%%%%%%%%%%%%%%%%%%%%% structure of the article
The present article is organized as follows: in Section \ref{sec:Prufer}, we present the transfer matrix analysis which will be the central ingredient in our proofs. In particular, we define the Pr\"ufer transform in Section \ref{sec:the-transform}. 
Section \ref{sec:asymptotics-TM} contains the asymptotics of the transfer matrices obtained via martingale methods. 
We show absolutely continuous spectrum for the super critical regime in Section \ref{sec:ac}. 
The spectral transition and transport in the critical region are proved in Section \ref{sec:critical}. 
For the sub-critical regime, we show spectral and dynamical localization in Section
\ref{sec:DL-pp} and
\ref{sec:DL} respectively. 
Finally, the appendix contains several technical estimates and some parts of the proofs which were deferred to lighten the presentation.

%%%%%%%%%%%%%%%%%%%%%%%%%%%%%%%%%%%%%%%%%%%%%%%%%%
%%%%%%%%%%%%%%%% Notation
\bigskip

{\bf Notation.}
We set $\n^*=\{1,2,\cdots\}$ and let $\h$ be the Hilbert space $l^2(\n^*,\C^2)$ with its natural canonical basis
$\{\delta^\pm_n:\, n\in\n^*\}$. 
A vector $\Phi=(\Phi_n)_n\in\h$ is given by two sequences $\vp^\pm=(\vp_n^\pm)\in\ell^2(\n^*,\C)$ such that 
\begin{eqnarray*}
	\Phi_n
	=
	\begin{pmatrix}
		\vp^+_n \\ \vp^-_n
	\end{pmatrix} \quad\quad \text{for}\ n\in\n^*.
\end{eqnarray*}
We will occasionally denote this relation by $\Phi=\vp^+ \otimes \vp^-$.
If $B$ is a set, we write $\chi_B$ for its characteristic function.
Constants such as $ C(a,b,\dots)$ will be finite and positive and will depend only on the parameters or quantities $a, b,\dots$; they will be independent of the other parameters or quantities involved in the equation. Note that the value of $C(a,b,\dots)$ may change from line to line.
\newline
Given an open interval $I\subset\R$, we consider $\mathcal{C}_{c,+}^\infty$ is the class of infinitely differentiable non-negative real valued functions  with compact support contained in $I$.
We set $P_I(H)=\chi_I(H)$ the spectral projection of an operator $H$ on the interval $I$ and $\sigma(H)$ for its spectrum.
The pure point, absolutely continuous, and singular continuous components will be denoted by $\sigma_{pp}(H), \sigma_{ac}(H)$ and $\sigma_{sc}(H)$ respectively. Finally, we consider the position operator $|\X|$ on $\h$ defined by $|\X|\delta^{\pm}_{n} = n$.
%%%%%%%%%%%%%%%% %%%%%%%%%%%%%%%%%%%%%%%%%%%%%%%%%%%%%%%%%%%%
%%%%%%%%%%%%%%%%%%%%%%%%%%%%%%%%%%%%%%%%%%%%%%%%%%%%%%%%%%%%

%%%%%%%%%%%%%%%%%%%%%%%%%%%%%%%%%%%%%%%%%%%%%%%%%%%%%%%%%%%%
%%%%%%%%%%%%%%%%%%%%%%%%%%%%%%%%%%%%%%%%%%%%%%%%%%%%%%%%%%%%
%%%%%%%%%%%%%%%%%%%%%%%%%%%%%%%%%%%%%%%%%%%%%%%%%%%%%%%%%%%%
%%%%%%%%%%%%%%%%%%%%%%%%%%%%%%%%%%%%%%%%%%%%%%%%%%%%%%%%%%%%
%%%%%%%%%%%%%%%%%%%%%%%%%%%%%%%%%%%%%%%%%%%%%%%%%%%%%%%%%%%%
%%%%%%%%%%%%%%%%%%%%%%%%%%%%%%%%%%%%%%%%%%%%%%%%%%%%%%%%%%%%
%%%%%%%%%%%%%%%%%%%%%%%%%%%%%%%%%%%%%%%%%%%%%%%%%%%%%%%%%%%%

\section{Model and Results}

%%%%%%%%%%%%%%%%%%%%%%%%%%%%%%%%%%%%%%%%%%%%%%%%%%%%%%%%%%%%
%%%%%%%%%%%%%%%%%%%%%%%%%%%%%%%%%%%%%%%%%%%%%%%%%%%%%%%%%%%%
%%%%%%%%%%%%%%%%%%%%%%%%%%%%%%%%%%%%%%%%%%%%%%%%%%%%%%%%%%%%
%%%%%%%%%%%%%%%%%%%%%%%%%%%%%%%%%%%%%%%%%%%%%%%%%%%%%%%%%%%%

\subsection{Dirac operator with decaying random potential}

%%%%%%%%%%%%%%%%%%%%%%%%%%%%%%%%%%%%%%%%%%%%%%%%%%%%%%%%%%%%
%%%%%%%%%%%%%%%%%%%%%%%%%%%%%%%%%%%%%%%%%%%%%%%%%%%%%%%%%%%%
%%%%%%%%%%%%%%%%%%%%%%%%%%%%%%%%%%%%%%%%%%%%%%%%%%%%%%%%%%%%
%%%%%%%%%%%%%%%%%%%%%%%%%%%%%%%%%%%%%%%%%%%%%%%%%%%%%%%%%%%%
We consider the \textit{free Dirac operator} $D$ defined by
\begin{equation}\label{free Dirac}
	D
	=
	\begin{pmatrix}
		m & d \\
		d^* & -m
	\end{pmatrix} \quad \mathrm{on}\quad \h,
\end{equation}
with mass $m\ge0$ and where $d$ and $d^*$ are the finite difference operators acting on  $\ell^2(\n^*,\C)$ as $(du)_n=u_n-u_{n+1}$ and  $(d^*u)_n=u_n-u_{n-1}$  for $u=(u_n)_n\in\ell^2(\n^*,\C)$ with the convention $u_{0}=0$. 
Writing $\Phi=\vp^+ \otimes \vp^-$ for $\Phi\in\h$, we have
\begin{equation*}
	D \Phi
	=
	\begin{pmatrix}
		m \vp^+ + d\vp^-
		\\
		d^*\vp^+ - m \vp^-
	\end{pmatrix}.
\end{equation*}
To define the perturbed operator, we introduce a family of integrable random variables $\{ V_{\omega,i}(n):\, n\in \n^*, i=1,2\}$
%$(V_{\omega,i}(n))_{n\in\n^*}$, $i=1,\, 2$, 
defined on a probability space $(\Omega,\mathcal{F},\p)$. We denote the expected value with respect to $\p$ by $\esp$. We then define the random multiplication operator $V_{\omega}$ acting on the canonical vectors as
\begin{equation}\label{V_omega}
	V_{\omega}\delta^+_n = \vo(n) \delta^+_n,
	\quad
	V_{\omega}\delta^-_n = \vt(n) \delta^-_n.
\end{equation}
Let $\lambda>0$, $\alpha>0$ and let $(a_n)_n$ be a positive sequence such that $\displaystyle \lim_{n\to\infty}n^{\alpha}a_n = 1$. In most of the following, we will assume that
\begin{enumerate}
	\item[ \textbf{(A1)}] The random variables $\{V_{\omega,i}(n);\, n\in\n^*,\, i=1,\, 2\}$ are independent.
	
	\item[\textbf{(A2)}] $\esp[V_{\omega,i}(n)]=0$.
	
	\item[\textbf{(A3a)}] $\esp[V_{\omega,i}(n)^2]^{1/2}= \lambda a_n$.
	
	\item[\textbf{(A3b)}]	$\esp[V_{\omega,i}(n)^4] \leq C a_n^2$, for some finite $C>0$.
	
	\item[\textbf{(A4)}] There exist a $\p$-almost surely finite constant $C(\omega)>0$ and $\varepsilon>0$ such that 
		\begin{equation*}
			|V_{\omega,i}(n)| \leq C(\omega) n^{-\frac{2\alpha}{3}-\varepsilon},
\end{equation*}			
	for all $n\in\n^*$, $\omega\in\Omega$ and $i=1,2$.
\end{enumerate}
More general hypothesis will be considered and clarified in due time.
Notice that these assumptions can be achieved for instance by considering
\begin{eqnarray}\label{eq:basic-model}
	V_{\omega,1}(n) = \lambda a_n \omega_{1,n},
	\quad 
	V_{\omega,2}(n) = \lambda a_n \omega_{2,n},
\end{eqnarray}
for a family of independent integrable random variables $\{\omega_{n,i}:, n\in\n^*,i=1,2\}$ defined on $(\Omega,\mathcal{F},\p)$ such that $\esp[\omega_{n,i}]=0$ and $\esp[\omega_{n,i}^2]=1$ with some suitable conditions on their moments.
%%%%%%%%%%%%%%%%%%%%%%%%%%%%%%%%%%%%%%%%%%%%%%%%%%%%%%%%%%%%
%%%%%%%%%%%%%%%%%%%%%%%%%%%%%%%%%%%%%%%%%%%%%%%%%%%%%%%%%%%%

FInally, the \textit{Dirac operator in a decaying random potential} is given by
\begin{equation}\label{Dirac}
	\D=D_{\omega,\lambda,\alpha} = D+V_{\omega}, \quad\mathrm{on}\quad \h.
\end{equation}
Notice that $\D$ is a non-ergodic family of bounded self-adjoint operators on $\h$ for every $\omega\in\Omega$. In particular, the existence of deterministic spectral components is not straightforward. Nonetheless, as $V_{\omega}$ is compact, the essential spectra of $\D$ and $D$ coincide.
%%%%%%%%%%%%%%%%%%%%%%%%%%%%%%%%%%%%%%%%%%%%%%%%%%%%%%%%%%%%
%%%%%%%%%%%%%%%%%%%%%%%%%%%%%%%%%%%%%%%%%%%%%%%%%%%%%%%%%%%%

%%%%%%%%%%%%%%%%%%%%%%%%%%%%%%%%%%%%%%%%%%%%%%%%%%%%%%%%%%%%
%%%%%%%%%%%%%%%%%%%%%%%%%%%%%%%%%%%%%%%%%%%%%%%%%%%%%%%%%%%%
%%%%%%%%%%%%%%%%%%%%%%%%%%%%%%%%%%%%%%%%%%%%%%%%%%%%%%%%%%%%
%%%%%%%%%%%%%%%%%%%%%%%%%%%%%%%%%%%%%%%%%%%%%%%%%%%%%%%%%%%%
%%%%%%%%%%%%%%%%%%%%%%%%%%%%%%%%%%%%%%%%%%%%%%%%%%%%%%%%%%%%
%%%%%%%%%%%%%%%%%%%%%%%%%%%%%%%%%%%%%%%%%%%%%%%%%%%%%%%%%%%%
%%%%%%%%%%%%%%%%%%%%%%%%%%%%%%%%%%%%%%%%%%%%%%%%%%%%%%%%%%%%
%%%%%%%%%%%%%%%%%%%%%%%%%%%%%%%%%%%%%%%%%%%%%%%%%%%%%%%%%%%%

%\subsection{Results}

%%%%%%%%%%%%%%%%%%%%%%%%%%%%%%%%%%%%%%%%%%%%%%%%%%%%%%%%%%%%
%%%%%%%%%%%%%%%%%%%%%%%%%%%%%%%%%%%%%%%%%%%%%%%%%%%%%%%%%%%%
%%%%%%%%%%%%%%%%%%%%%%%%%%%%%%%%%%%%%%%%%%%%%%%%%%%%%%%%%%%%
%%%%%%%%%%%%%%%%%%%%%%%%%%%%%%%%%%%%%%%%%%%%%%%%%%%%%%%%%%%%

\subsection{Spectral regimes}

%%%%%%%%%%%%%%%%%%%%%%%%%%%%%%%%%%%%%%%%%%%%%%%%%%%%%%%%%%%%
%%%%%%%%%%%%%%%%%%%%%%%%%%%%%%%%%%%%%%%%%%%%%%%%%%%%%%%%%%%%

%WARNING: THE FOURIER STUFF IS FOR DELTA ON THE WHOLE LINE

The spectral structure of the Dirac operator $D$ can be inferred from the simple relation
\begin{equation}\label{dirac^2}
 D^2=\begin{pmatrix}
        \Delta+m^2&0\\0&\Delta+m^2
       \end{pmatrix},
\end{equation}
where $\Delta$ is the discrete Laplacian defined on $\ell^2(\n^*,\C)$ by $(\Delta u)_n=2u(n)-u(n+1)-u(n-1)$, with the convention $u(0)=u(-1)=0$. It is well-known that the spectrum of $\Delta$ fills the interval $[0,4]$ and consists of purely absolutely continuous spectrum in its interior. In particular, $\sigma(D^2)=[m^2,m^2+4]$ where this equality reminds the relation between momentum and energy
in relativistic quantum mechanics. 
%
%Under a Fourier transform from $\ell^2(\Z)$ to $\mathrm{L}^2([0,2\pi])$,
%we can see that $\Delta$ is unitarily equivalent to the multiplication 
%operator by the function $h_0(\nu)=2(1-\cos \nu)$ so that
%$\Delta$ has purely absolutely continuous spectrum $\sigma(\Delta)=[0,4]$.
%In particular, $\sigma(D^2)=[m^2,m^2+4]$ where this equality reminds the relation between momentum and energy
%in relativistic quantum mechanics. 
The spectrum of the free Dirac operator \eqref{free Dirac} is hence given by
\begin{equation*}
\Sigma:= \sigma(D)=[-\sqrt{m^2+4},-m]\cup[m,\sqrt{m^2+4}].
\end{equation*}
It is known that $\mathring\Sigma:=(-\sqrt{m^2+4},-m)\cup(m,\sqrt{m^2+4})$ consists of purely absolutely continuous spectrum \cite[Proposition 2.8]{dO-sparse-II}. 

Since $\D$ is a compact pertubation of $D$, the essential spectrum of $\D$ coincides with $\Sigma$. In particular, $\D$ can only have discrete spectrum at energies outside $\Sigma$. As the family $(\D)_{\omega\in\Omega}$ is not ergodic, there is no a priori guaranty that the sepctral components are deterministic.

%%%%%%%%%%%%%%%%%%%%%%%%%%%%%%%%%%%%%%%%%%%%%%%%%%%%%%%%%%%%
\begin{figure}[h!]
\begin{center}
\includegraphics[scale=0.6]{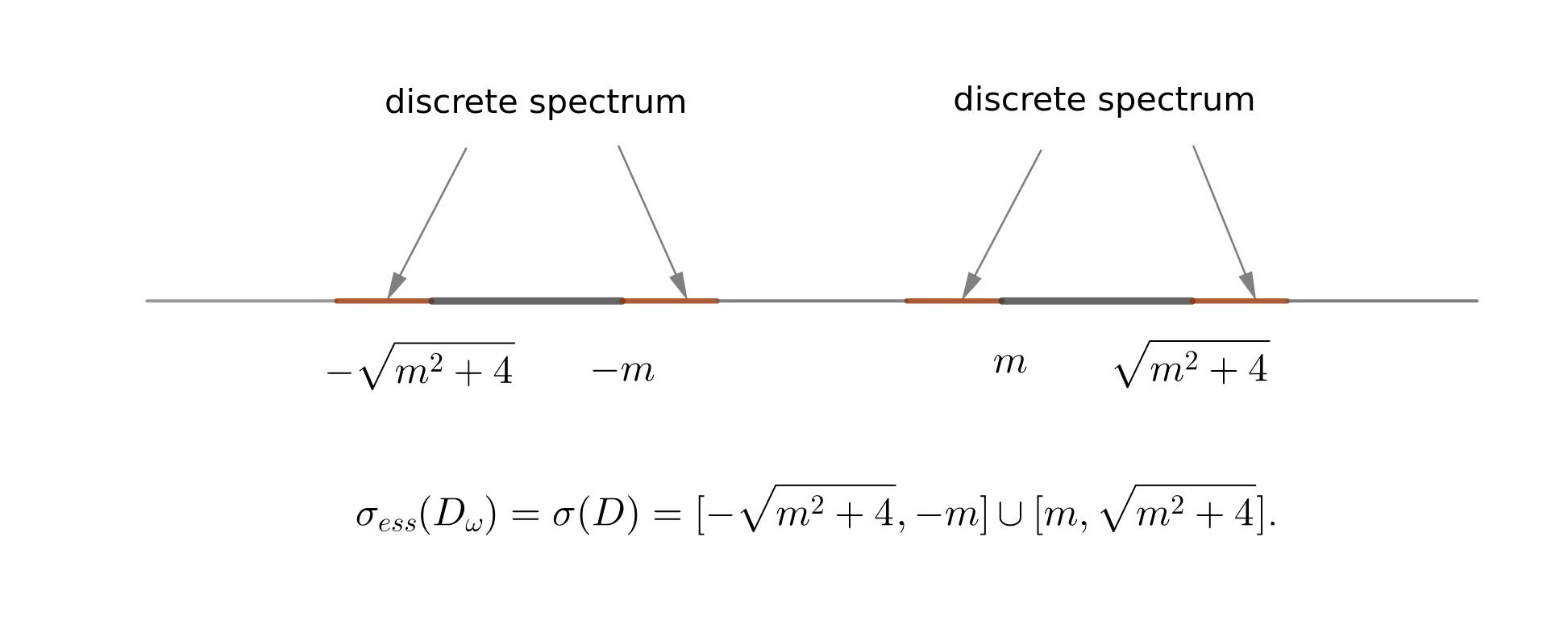}
%\vspace*{0.2cm}
\caption{Spectral structure of the operator $\D$.} \label{figd}
\end{center}
\end{figure}
%%%%%%%%%%%%%%%%%%%%%%%%%%%%%%%%%%%%%%%%%%%%%%%%%%%%%%%%%%%%

We present our result on the nature of the spectrum of $\D$ for different values of the parameters.
Let $\lambda_m:\Sigma \to [0,\infty)$ be the function defined by
\begin{equation*}
	\lambda_m(E)^2
	=
	\frac12
	\frac{(E^2-m^2)(m^2+4-E^2)}{m^2+E^2}.
\end{equation*}
Let $\lambda^*(m)$ be its maximal value in $[m,\sqrt{m^2+4}]$ and for $|\lambda|<\lambda^*(m)$, let $E^*_-(\lambda,m)<E^*_+(\lambda,m)$ be the two roots of the equation $\lambda_m(E)=|\lambda|$ in $[m,\sqrt{m^2+4}]$.

%%%%%%%%%%%%%%%%%%%%%%%%%%%%%%%%%%%%%%%%%%%%%%%%%%%%%%%%%%%%
%%%%%%%%%%%%%%%%%%%%%%%%%%%%%%%%%%%%%%%%%%%%%%%%%%%%%%%%%%%%
\begin{theorem}\label{thm:spectrum}
	Assume \textbf{(A1)}-\textbf{(A4)}.
	Then, the essential spectrum of $\D$ is $\p$-a.s. equal to $\Sigma$. Furthermore, 

\begin{enumerate}
	\item[(1)] \textbf{Super-critical case.} If $\alpha>\frac12$ then, for all $\lambda>0$, the spectrum of $\D$ is almost surely purely absolutely continuous in $\mathring\Sigma$.
	% namely $\sigma(\D)\cap\mathring\Sigma=\sigma_{ac}(\D) $.
					\vspace{1ex}

    \item[(2)] \textbf{Critical case.} If $\alpha=\frac12$ then for all $\lambda > 0$, the a.c. spectrum of $\D$ is almost surely empty. Moreover,
			\begin{enumerate}[label=\alph*.-]
			\item[a.] If $\lambda\geq \lambda^*(m)$, then the spectrum of $\D$ is almost surely pure point in $\mathring\Sigma$.
			%, namely $\sigma(\D)\cap\mathring\Sigma=\sigma_{pp}(\D) $.
			
			\vspace{1ex}
			
			\item[b.] If $\lambda< \lambda^*(m)$ then, almost surely, the spectrum of $\D$ is purely singular continuous in $\{E\in\mathring\Sigma:\, |E|\in(E^*_-(\lambda,m),E^*_+(\lambda,m))\}$ and pure point in the complement of this set.
		\end{enumerate}

	\vspace{1ex}
		
	\item[(3)] \textbf{Sub-critical case.} If $\alpha<\frac12$ then for all $\lambda> 0$, the spectrum of $\D$ is almost surely pure point in $\mathring\Sigma$.
	%, namely $\sigma(\D)\cap\mathring\Sigma= \sigma_{pp}(\D)$
\end{enumerate}
\end{theorem}
%%%%%%%%%%%%%%%%%%%%%%%%%%%%%%%%%%%%%%%%%%%%%%%%%%%%%%%%%%%%
%%%%%%%%%%%%%%%%%%%%%%%%%%%%%%%%%%%%%%%%%%%%%%%%%%%%%%%%%%%%
\begin{remark}
	The assumptions \textbf{(A1)}-\textbf{(A4)} are made in such a way that the three parts of the above theorem can be jointly proved. Nonetheless, they can be weakened in the following ways:
	\begin{enumerate}
		
		\item Part 1 can be proved replacing \textbf{(A3a)} and \textbf{(A3b)} by
			\begin{eqnarray*}\label{eq:minimal-condition-ac}
				\sum_{n,i} \left(
					\esp[V_{\omega,i}(n)^2]+\esp[V_{\omega,i}(n)^4]
				\right)
				< \infty.
			\end{eqnarray*}
			The hypothesis \textbf{(A4)} is not required in this part.
			Furthermore, this is the only place where \textbf{(A3b)} is used.
			
		\item Assumption \textbf{(A4)} can be verified under some moments conditions. For instance, if 		
			\begin{eqnarray*}
				\esp[|V_{\omega,i}(n)|^p] \leq C n^{-(\frac{2\alpha}{3}+\varepsilon)p},
			\end{eqnarray*}
		for some $C>0$, $\varepsilon>0$ and $p>1$ such that $p>\frac{1}{\varepsilon}$, an application of Borel-Cantelli yields that for all $\varepsilon'<\varepsilon-\frac{1}{p}$, we have
			\begin{eqnarray*}
			\qquad
				|V_{\omega,i}(n)| \leq n^{-\frac{2\alpha}{3}-\varepsilon'},
			\end{eqnarray*}
			for all $n\geq \tau$, $i=1,2$, for some $\p$-almost surely finite $\tau=\tau(\omega)$, so that \textbf{(A4)} holds.
			
			\item If the potential has the form \eqref{eq:basic-model} for independent random variables $\{\omega_{n,i}:\, n\in\n^*,\, i=1,2\}$, then \textbf{(A2)}, \textbf{(A3a)} and \textbf{(A3b)} are satisfied if
			\begin{eqnarray*}
				\esp[\omega_{n,i}]=0,
				\quad
				\esp[\omega_{n,i}^2]=1
				\quad
				\text{and}
				\quad
				\esp[\omega_{n,i}^4] \leq C a_n^{-2},
			\end{eqnarray*}
			respectively. By Borel-Cantelli lemma, the condition \textbf{(A4)} is satisfied if the $\omega_{n,i}$'s have finite moments of order $p>\frac{3}{\alpha}$ such that
				\begin{eqnarray*}
					\esp[|\omega_{n,i}|^p] \leq C n^{\gamma p},
				\end{eqnarray*}
				for some $C\in(0,\infty)$ and $\gamma < \frac{\alpha}{3}-\frac{1}{p}$.
			
		\item Part 3 can be proved as a consequence of our dynamical localization result below which requires some moments assumptions that imply \textbf{(A4)}. However, the proof of dynamical localization requires some regularity of the law of the random variables, for instance \textbf{(A5)}. See Theorem \ref{thm:DL} and its consequences in Proposition \ref{thm:consequences-DL}.

		\item We note that our hypothesis are slightly more general than the ones stated in \cite{KLS} for the Anderson model in a decaying random potential. In particular, we allow unbounded random variables. In \cite{KLS}, it is assumed that \textbf{(A4)} holds with a deterministic constant. Assuming the weaker random condition requires minor adjustments. Hypothesis \textbf{(A4)} is made in order to truncate some expansions to order $2$.
	\end{enumerate}
\end{remark}
%%%%%%%%%%%%%%%%%%%%%%%%%%%%%%%%%%%%%%%%%%%%%%%%%%%%%%%%%%%%
%%%%%%%%%%%%%%%%%%%%%%%%%%%%%%%%%%%%%%%%%%%%%%%%%%%%%%%%%%%%
\begin{remark}
	The proofs of a.c. spectrum in Part 1 and absence of a.c. spectrum in Part 2 are based on general criteria of Last and Simon \cite{LS} developed for the Anderson model on $\ell^2(\mathbb{N})$. By means of the transform $\deltam_n \mapsto \delta_{2n}$, $\deltap_n\mapsto\delta_{2n+1}$, the operator $\D$ can be seen as a finite-range operator on $\ell^2(\mathbb{N})$ to which the criteria of Last and Simon applies with minor adaptations. These criteria were already applied for Dirac operators in a sparse potential in \cite{dO-sparse-II}.
	
%	Part 3 is a consequence of the dynamical localization result below and could be proved for the model defined on the whole line with minor modifications. It is nonetheless possible to prove spectral localization directly (see \cite{KLS}).
\end{remark}
%%%%%%%%%%%%%%%%%%%%%%%%%%%%%%%%%%%%%%%%%%%%%%%%%%%%%%%%%%%%
%%%%%%%%%%%%%%%%%%%%%%%%%%%%%%%%%%%%%%%%%%%%%%%%%%%%%%%%%%%%
%\begin{remark}
%	It was shown in \cite{KKO} that the spectrum of the Anderson model in a decaying potential is pure point outside the essential spectrum. Their argument can be applied in our context with minor adaptations.
%\end{remark}
%%%%%%%%%%%%%%%%%%%%%%%%%%%%%%%%%%%%%%%%%%%%%%%%%%%%%%%%%%%%
%%%%%%%%%%%%%%%%%%%%%%%%%%%%%%%%%%%%%%%%%%%%%%%%%%%%%%%%%%%%

%%%%%%%%%%%%%%%%%%%%%%%%%%%%%%%%%%%%%%%%%%%%%%%%%%%%%%%%%%%%
%%%%%%%%%%%%%%%%%%%%%%%%%%%%%%%%%%%%%%%%%%%%%%%%%%%%%%%%%%%%
%%%%%%%%%%%%%%%%%%%%%%%%%%%%%%%%%%%%%%%%%%%%%%%%%%%%%%%%%%%%
%%%%%%%%%%%%%%%%%%%%%%%%%%%%%%%%%%%%%%%%%%%%%%%%%%%%%%%%%%%%

\subsection{Dynamical regimes}

%%%%%%%%%%%%%%%%%%%%%%%%%%%%%%%%%%%%%%%%%%%%%%%%%%%%%%%%%%%%
%%%%%%%%%%%%%%%%%%%%%%%%%%%%%%%%%%%%%%%%%%%%%%%%%%%%%%%%%%%%
Delocalization in the regions of continuous spectrum is a consequence of the RAGE theorem \cite{CFKS}.
The next theorem establishes the absence of dynamical localization in the critical regime even in the region of pure point spectrum.
%%%%%%%%%%%%%%%%%%%%%%%%%%%%%%%%%%%%%%%%%%%%%%%%%%%%%%%%%%%%
%%%%%%%%%%%%%%%%%%%%%%%%%%%%%%%%%%%%%%%%%%%%%%%%%%%%%%%%%%%%
\begin{theorem}\label{thm:transport}
 	%$\lambda>0$ and assume \textbf{(A1)}-\textbf{(A3a)}.
	Let $\alpha=\frac12$, $\lambda>0$ and $I$ be a compact interval such that $I\subset\mathring \sigma_{pp}(\D)$. Then, there exists $p_0=p_0(I)>0$ such that
	for $\p$-almost every $\omega$, 
	\begin{eqnarray}\label{delocalization}
		\limsup_{t\to \infty} \left\| |\X|^{p/2} \e^{-it \D} \psi\right\|^2 = \infty,
	\end{eqnarray}
	for all $p>p_0$ and $\psi \in {\text Ran}P_I(\D)$. %where $P_I(\D)$ denotes the spectral projection of $\D$ on the interval $I$.
\end{theorem}
%%%%%%%%%%%%%%%%%%%%%%%%%%%%%%%%%%%%%%%%%%%%%%%%%%%%%%%%%%%%
%%%%%%%%%%%%%%%%%%%%%%%%%%%%%%%%%%%%%%%%%%%%%%%%%%%%%%%%%%%%
The work \cite{GKT} gives precise quantitative lower bounds on the moments for the one-dimensional Anderson model with a variety of potentials, including the critically decaying random case. This information is missing in the above theorem which proof is nonetheless elementary and robust. General lower bounds in the spirit of \cite{GKT} for one-dimensional Dirac operators were obtained in \cite{dO-LB}. It is likely that our analysis could be used as an input for their method to obtain bounds on the transport exponents.
%%%%%%%%%%%%%%%%%%%%%%%%%%%%%%%%%%%%%%%%%%%%%%%%%%%%%%%%%%%%
%%%%%%%%%%%%%%%%%%%%%%%%%%%%%%%%%%%%%%%%%%%%%%%%%%%%%%%%%%%%

To characterize the dynamical localization, we define the eigenfunction correlator 
\begin{equation}\label{correlator}
	Q_\omega(u,\sigma;n,\sigma';I)
	=
	\sup_{\substack{f\in \mathcal{C}_{c,+}^\infty\\ \norm{f}_\infty\le1}}
	 \left| \langle \delta^{\sigma}_u , P_I(\D)f(\D) \delta^{\sigma'}_n \rangle \right|,
\end{equation}
for $u,n \in \n^*$ and $\sigma,\sigma'\in\{+,-\}$.

%$P_I(D_{\omega,\lambda})$ denotes the spectral projection of $D_{\omega,\lambda}$ on the interval $I$ and 
%$C_0(I)$ denotes the space of bounded measurable functions compactly supported in $I$ and 
%%%%%%%%%%%%%%%%%%%%%%%%%%%%%%%%%%%%%%%%%%%%%%%%%%%%%%%%%%%%
%%%%%%%%%%%%%%%%%%%%%%%%%%%%%%%%%%%%%%%%%%%%%%%%%%%%%%%%%%%%
\begin{definition}
 We say that $\D$ exhibits dynamical localization in an interval $I\subset\R$ if we have
 \begin{equation}\label{eq:DL}
  \sum_{n,\sigma'}\esp\left[Q_\omega(u,\sigma;n,\sigma';I)^2\right]<\infty,
 \end{equation}
for all $u\in\n^*$ and $\sigma\in\{+,-\}$.
\end{definition}
%%%%%%%%%%%%%%%%%%%%%%%%%%%%%%%%%%%%%%%%%%%%%%%%%%%%%%%%%%%%
%%%%%%%%%%%%%%%%%%%%%%%%%%%%%%%%%%%%%%%%%%%%%%%%%%%%%%%%%%%%
Next, we state the additional hypothesis needed for the proof of dynamical localization:
\begin{enumerate}
	\item[\textbf{(A5)}] There exist $C>0$ and $\varepsilon>0$ such that
				\begin{equation*}
					\esp[|V_{\omega,i}(n)|^3] \leq C n^{-(2\alpha + \varepsilon)}.
				\end{equation*}
	
	\item[\textbf{(A6)}] There exist $p>1$, $\gamma\geq 0$ and $C>0$ such that 
			\begin{enumerate}
				\item[\textbf{(A6a)}] For each $n\in\n^*$ and $i=1,2$, the random variable $V_{\omega,i}(n)$ admits a 		density $\rho_{n,i}$ so that
					\begin{equation*}
						\int_{\R} \rho_{n,i}(y)^p dy \leq C \lambda^{-\gamma} a_n^{-\gamma}.
					\end{equation*}
				\item[\textbf{(A6b)}] For all $n\in\n^*$ and $i=1,2$,
					\begin{equation*}
						\int_{\R} |y|^s \
						\rho_{n,i}(y)^p dy \leq C \lambda^{-\gamma} a_n^{-\gamma},
						\quad\
						\text{for all}\ \, 0\leq s<\frac{p-1}{p},
					\end{equation*}
					and
					\begin{equation*}
						\int_{\R} (1+|y|)^{-s} \rho_{n,i}(y)^p dy \geq C \lambda^{\gamma} a_n^{\gamma},
						\quad\ 
						\text{for all}\ \, 0\leq s<1.
					\end{equation*}
			\end{enumerate}
\end{enumerate}
%%%%%%%%%%%%%%%%%%%%%%%%%%%%%%%%%%%%%%%%%%%%%%%%%%%%%%%%%%%%
%%%%%%%%%%%%%%%%%%%%%%%%%%%%%%%%%%%%%%%%%%%%%%%%%%%%%%%%%%%%
By density, we refer to a non-negative function $\rho_{n,i} \in L^1(\R)$ such that
\begin{equation*}
	\p[V_{\omega,i}(n)\in A] = \int_A \rho_{n,i}(y)\, dy,
\end{equation*}
for all Borel set $A\subset \R$.
%\begin{eqnarray}
%	\esp[f(V_{\omega,i}(n))] = \int_{\R} f(y)\rho_{n,i}(y)\, dy,
%\end{eqnarray}
%for all measurable function $f$. 
Notice that we do not assume $\rho_{n,i}$ to be bounded. Note that \textbf{(A5)} implies \textbf{(A4)} by a standard Borel-Cantelli argument.
%%%%%%%%%%%%%%%%%%%%%%%%%%%%%%%%%%%%%%%%%%%%%%%%%%%%%%%%%%%%
%%%%%%%%%%%%%%%%%%%%%%%%%%%%%%%%%%%%%%%%%%%%%%%%%%%%%%%%%%%%

The next theorem contains our result on dynamical localization in the \textit{sub-critical} regime.
%%%%%%%%%%%%%%%%%%%%%%%%%%%%%%%%%%%%%%%%%%%%%%%%%%%%%%%%%%%%
%%%%%%%%%%%%%%%%%%%%%%%%%%%%%%%%%%%%%%%%%%%%%%%%%%%%%%%%%%%%
\begin{theorem}\label{thm:DL}
	Let $0<\alpha<\frac12$ and $\lambda> 0$. Assume \textbf{(A1)}-\textbf{(A3a)}, \textbf{(A5)} and \textbf{(A6)}. Then, for each $u\in\n^*$ and each compact energy interval $I\subset \mathring\Sigma$, there exists constants $C=C(u,I)>0$, $c(u,I)>0$ and $\kappa=\kappa(I)>0$ such that
	\begin{equation}\label{eq:FM-bound}
		\esp[Q_\omega(u,\sigma;n,\sigma';I)^2] \leq C (\lambda a_n)^{-\kappa} \e^{-cn^{1-2\alpha}},
	\end{equation}
	for all $n\in \n^*$ and $\sigma,\sigma'\in\{+, -\}$.
	 In particular, $\D$ almost surely exhibits dynamical localization in the interval $I$.
\end{theorem}
%%%%%%%%%%%%%%%%%%%%%%%%%%%%%%%%%%%%%%%%%%%%%%%%%%%%%%%%%%%%
%%%%%%%%%%%%%%%%%%%%%%%%%%%%%%%%%%%%%%%%%%%%%%%%%%%%%%%%%%%%
\begin{remark}
	We follow the fractional moments method \cite{AM,AW}.
	The analogue of Theorem \ref{thm:DL} for the Anderson model in a sub-critical decaying potential was obtained in \cite{Si82} by means of the Kunz-Souillard method \cite{KS} (see also \cite{DG, Bu1, Bu2}). It is likely that this method can be adapted to our setting. However, its application in \cite{Si82} requires to assume that the random variables admit a bounded density (which may grow with $n$). This hypothesis is not needed here. Instead, we assume \textbf{(A5)} which implies some regularity on their laws. 
	Our proof can be easily adapted to the Anderson model with, in fact, some simplifications.
\end{remark}
%%%%%%%%%%%%%%%%%%%%%%%%%%%%%%%%%%%%%%%%%%%%%%%%%%%%%%%%%%%%
%%%%%%%%%%%%%%%%%%%%%%%%%%%%%%%%%%%%%%%%%%%%%%%%%%%%%%%%%%%%
\begin{remark}
	With minor modifications, we may merely assume that $a_n \geq cn^{-\alpha}$ for some $\alpha\in(0,\frac12)$.
\end{remark}
%%%%%%%%%%%%%%%%%%%%%%%%%%%%%%%%%%%%%%%%%%%%%%%%%%%%%%%%%%%%
%%%%%%%%%%%%%%%%%%%%%%%%%%%%%%%%%%%%%%%%%%%%%%%%%%%%%%%%%%%%
Although the lack of ergodicity of the model induces the dependence of \eqref{eq:FM-bound}  on the base site $u$, it can be shown that the bound \eqref{eq:DL} still implies pure point spectrum and finiteness of the moments. In particular, it implies Part 3 in Theorem \ref{thm:spectrum} paying the price of hypothesis \textbf{(A5)} and \textbf{(A6)}.
%%%%%%%%%%%%%%%%%%%%%%%%%%%%%%%%%%%%%%%%%%%%%%%%%%%%%%%%%%%%
%%%%%%%%%%%%%%%%%%%%%%%%%%%%%%%%%%%%%%%%%%%%%%%%%%%%%%%%%%%%
\begin{proposition}\label{thm:consequences-DL}
	Assume that dynamical localization for $\D$ holds in the sense of \eqref{eq:DL} in an energy interval $I\subset\R$. Then the following holds
	\begin{enumerate}
		\item[(1)] The spectrum of $\D$ is almost surely pure point in $I$.
		
		\item[(2)] For all $p>0$,
			\begin{equation*}
				\esp\left(\sup_{t\in\R}
					\norm{
						|\X|^{\frac{p}{2}}\e^{-it\D}P_I(\D)\psi_0}^2
				\right)<\infty,
			\end{equation*}
			for all $\psi_0 \in\h$ with bounded support. 
	\end{enumerate}
\end{proposition}
%%%%%%%%%%%%%%%%%%%%%%%%%%%%%%%%%%%%%%%%%%%%%%%%%%%%%%%%%%%%
%%%%%%%%%%%%%%%%%%%%%%%%%%%%%%%%%%%%%%%%%%%%%%%%%%%%%%%%%%%%

Our analysis provides a control on the eigenfunctions of $\D$. Let $0<\alpha<\frac12$ and denote by $\Phi_{\omega,E}=(\Phi_{\omega,E,n})_n$ the eigenfunction of $\D$ corresponding to the eigenvalue $E$. In Proposition \ref{thm:asymptotics-eigenfunctions-pp}, we follow \cite[Theorem 8.6]{KLS} to show that
\begin{equation*}
	\lim_{n\to \infty} \frac{1}{n^{1-2\alpha}} \log \| \Phi_{\omega,E,n}\| = -\beta(E,\lambda), \quad \p-\text{a.s.},
\end{equation*}
for almost every fixed $E\in\mathring\Sigma$. In particular, this shows that for almost every $E\in\mathring\Sigma$, $\p$-almost surely, there exists a constant $C_{\omega,E}$ such that
\begin{equation*}
	\| \Phi_{\omega,E,n}\| \leq C_{\omega,E}\ \e^{-\beta(E,\lambda)n^{1-2\alpha}}.
\end{equation*}
It is known that certain types of decay of eigenfunctions are closely related to dynamical localization \cite{DeRJLS1,DeRJLS2,GT}. Such criteria usually require a control on the localization centres of the eigenfunctions, uniformly in energy intervals. This information is missing in the above bound. We provide this uniform control in the next proposition.
%%%%%%%%%%%%%%%%%%%%%%%%%%%%%%%%%%%%%%%%%%%%%%%%%%%%%%%%%%%%
%%%%%%%%%%%%%%%%%%%%%%%%%%%%%%%%%%%%%%%%%%%%%%%%%%%%%%%%%%%%
\begin{proposition} \label{thm:SULE}
	Let $0<\alpha<\frac12$ and $\lambda> 0$. Under the hypothesis of Theorem \ref{thm:DL},
	for all compact energy interval $I\subset \mathring\Sigma$, there exists two deterministic constants $c_1=c_1(I),\, c_2=c_2(I)$ and almost surely finite random quantities $c_{\omega}=c_{\omega}(I),\, C_{\omega}=C_{\omega}(I)$ such that
	\begin{eqnarray}\label{eq:SULE}
		c_{\omega} \, e^{-c_1 n^{1-2\alpha}} 
		\leq 
		\| \Phi_{\omega,E,n}\| 
		\leq 
		C_{\omega} e^{-c_2 n^{1-2\alpha}}, \quad \p-\text{a.s.},
	\end{eqnarray}
	for all $E \in I\cap\sigma(\D)$ and all $n\in\n^*$.
	%Furtheremore, we can take any $c_1 < \beta(\lambda,E)$.
\end{proposition}
%%%%%%%%%%%%%%%%%%%%%%%%%%%%%%%%%%%%%%%%%%%%%%%%%%%%%%%%%%%%
%%%%%%%%%%%%%%%%%%%%%%%%%%%%%%%%%%%%%%%%%%%%%%%%%%%%%%%%%%%%
 This asymptotics, although less precise about the exact rate of decay, is uniform in energy intervals. The upper bound \eqref{eq:SULE} can be seen as a stretched form of the condition SULE  where the localization centres are all equal to $0$ (see \cite{DeRJLS1}, equation (2)). 
%%%%%%%%%%%%%%%%%%%%%%%%%%%%%%%%%%%%%%%%%%%%%%%%%%%%%%%%%%%%
%%%%%%%%%%%%%%%%%%%%%%%%%%%%%%%%%%%%%%%%%%%%%%%%%%%%%%%%%%%%

The lower bound above allows us to characterize the `strength' of the localization according to the next theorem.
%%%%%%%%%%%%%%%%%%%%%%%%%%%%%%%%%%%%%%%%%%%%%%%%%%%%%%%%%%%%
%%%%%%%%%%%%%%%%%%%%%%%%%%%%%%%%%%%%%%%%%%%%%%%%%%%%%%%%%%%%
\begin{theorem}\label{thm:lower-bound-DL}
	Let $0<\alpha<\frac12$ and $\lambda> 0$.
	Let $I\subset \mathring\Sigma$. Under the hypothesis of Theorem \ref{thm:DL}, we have
		\begin{equation*}
				\esp\left(\sup_{t\in\R}
					\norm{\e^{|\X|^\kappa}
						 \e^{-it\D}P_I(\D)\psi_0}^2
				\right)<\infty,
			\end{equation*}
	for all $\kappa < 1-2\alpha$ and all $\psi_0 \in\h$ with bounded support, while
	\begin{eqnarray*}
		\limsup_{t\to \infty} \left\| \e^{|\X|^{\kappa}} \e^{-it \D} \psi\right\|^2 = \infty,
		\quad
		\p-a.s.,
	\end{eqnarray*}
	for all $\kappa>1-2\alpha$ and $\psi \in {\text Ran}P_I(\D)$.
\end{theorem}
%%%%%%%%%%%%%%%%%%%%%%%%%%%%%%%%%%%%%%%%%%%%%%%%%%%%%%%%%%%%
%%%%%%%%%%%%%%%%%%%%%%%%%%%%%%%%%%%%%%%%%%%%%%%%%%%%%%%%%%%%
\begin{remark}
	The lower bound in Proposition \ref{thm:SULE} and the second statement inTheorem \ref{thm:lower-bound-DL} will be proved assuming only \textbf{(A1)}-\textbf{(A3a)}.
\end{remark}

\section{Transfer matrices and Pr\"ufer transform}\label{sec:Prufer}

%%%%%%%%%%%%%%%%%%%%%%%%%%%%%%%%%%%%%%%%%%%%%%%%%%%%%%%%%%%%
%%%%%%%%%%%%%%%%%%%%%%%%%%%%%%%%%%%%%%%%%%%%%%%%%%%%%%%%%%%%
Let $\Phi=\vp^+\otimes \vp^-$
%$\displaystyle X
%= 
%\begin{pmatrix}
%	x^+ \\ x^-
%\end{pmatrix}
%$
be a solution of the eigenfunction equation $\D\Phi = E\Phi$. 
Then, the coordinates of $\Phi$ solve the system of equations
%We obtain the system of coordinate equations
\begin{eqnarray}
 (m+\vo(n)) \vp^+_n+\vp^-_n-\vp^-_{n+1}
 &=&
 E\vp^+_n\notag
 \\
 \label{syst1}
 \vp^+_n-\vp^+_{n-1}-(m-\vt(n))\vp^-_n
 &=&
 E\vp^-_n.
\end{eqnarray}
Let $p_1(E)=m-E$ and $p_2(E)=m+E$, and $\pno=p_1(E-\vo(n))$ and $\pnt=p_2(E-\vt(n))$. 
Then, the system \eqref{syst1} above can be written in the more compact form
\begin{eqnarray}
\nonumber
 \pno \vp^+_n+\vp^-_n-\vp^-_{n+1}
 &=&
 0
 \\
 \label{eq:eigensystem}
 \vp^+_n-\vp^+_{n-1}-\pnt \vp^-_n
 &=&
 0.
\end{eqnarray}
In the following, we will consider two different indexations of sequences in $(\C^2)^{\n}$ given by
\begin{eqnarray}\label{index}
	\Phi_n 
	=
	\begin{pmatrix}
		\vp_n^+
		\\
		\vp^-_n
	\end{pmatrix},
	\quad
	\Phi_n'
	=
	\begin{pmatrix}
		\vp^-_{n+1}
		\\
		\vp^+_n
	\end{pmatrix}.
\end{eqnarray}
We call these the first and second coordinate system respectively. The reasons to consider these two representations will become apparent in our proof of dynamical localization. 
\newline
In the following, we will assume that $p_1(E)<0$ and $p_2(E)>0$, which means that $E\in(m,\sqrt{m^2+4})$. By symmetry, the behaviour of the system in the other band of the spectrum is completely analogous. Occasionally, we will write the final result of the computations for the complementary case $p_1(E)>0$ and $p_2(E)<0$ to highlight the similarities.

%%%%%%%%%%%%%%%%%%%%%%%%%%%%%%%%%%%%%%%%%%%%%%%%%%%%%%%%%%%%
%%%%%%%%%%%%%%%%%%%%%%%%%%%%%%%%%%%%%%%%%%%%%%%%%%%%%%%%%%%%
%%%%%%%%%%%%%%%%%%%%%%%%%%%%%%%%%%%%%%%%%%%%%%%%%%%%%%%%%%%%
%%%%%%%%%%%%%%%%%%%%%%%%%%%%%%%%%%%%%%%%%%%%%%%%%%%%%%%%%%%%
\subsection{Transfer matrices}

%%%%%%%%%%%%%%%%%%%%%%%%%%%%%%%%%%%%%%%%%%%%%%%%%%%%%%%%%%%%
%%%%%%%%%%%%%%%%%%%%%%%%%%%%%%%%%%%%%%%%%%%%%%%%%%%%%%%%%%%%

%%%%%%%%%%%%%%%%%%%%%%%%%%%%%%%%%%%%%%%%%%%%%%%%%%%%%%%%%%%%
%%%%%%%%%%%%%%%%%%%%%%%%%%%%%%%%%%%%%%%%%%%%%%%%%%%%%%%%%%%%

\subsubsection{First coordinate system}

%%%%%%%%%%%%%%%%%%%%%%%%%%%%%%%%%%%%%%%%%%%%%%%%%%%%%%%%%%%%
%%%%%%%%%%%%%%%%%%%%%%%%%%%%%%%%%%%%%%%%%%%%%%%%%%%%%%%%%%%%
Shifting indexes in the second equation in \eqref{eq:eigensystem}, we obtain
\begin{eqnarray*}
 \pno \vp^+_n+\vp^-_n-\vp^-_{n+1}
 &=&
 0
 \\
 \vp^+_{n+1}-\vp^+_{n}-\ppnt\ \vp^-_{n+1}
 &=&
 0
\end{eqnarray*}
which can be written in matrix form as
\begin{eqnarray*}
	\begin{pmatrix}
		0 & 1 \\
		1 & -\ppnt
	\end{pmatrix}
	\Phi_{n+1}
	=
	\begin{pmatrix}
		\pno & 1 \\
		1 & 0
	\end{pmatrix}
	\Phi_n,
\end{eqnarray*}
yielding $\Phi_{n+1} = \transn \Phi_n$, where
\begin{eqnarray}\label{transfer}
	\transn
	=
	\begin{pmatrix}
		\pno \ \ppnt + 1 & \ppnt \\
		\pno & 1
	\end{pmatrix},
\end{eqnarray}
is the transfer matrix. 
Setting the disorder to be zero, we obtain the transfer matrix of the free system. In fact, if $D\Phi=E\Phi$, then $\Phi_{n+1}=T\Phi_n$ with
\begin{eqnarray*}
	T
	=
	T(E)
	=
	\begin{pmatrix}
		p_1p_2+1 & p_2 \\
		p_1 & 1
	\end{pmatrix},
\end{eqnarray*}
where we wrote $p_1=p_1(E)$ and $p_2=p_2(E)$ for simplicity.
Noticing that $\pno = p_1 + \vo(n)$ and $\ppnt = p_2-\vt(n+1)$, we can see that the transfer matrix admits the decomposition
\begin{eqnarray}
	\label{eq:decomposition-transfer-matrices}
	\transn
	&=&
	T + \vo(n) A_1 + \vt(n+1) A_2 +\vo(n)\vt(n+1)A_3,
\end{eqnarray}
with
\begin{eqnarray}\label{A123}
	A_1
	=
	\begin{pmatrix}
		 p_2 & 0\\
		1& 0
	\end{pmatrix},
	\quad
	A_2
	=
	\begin{pmatrix}
		-p_1  & -1\\
		0& 0
	\end{pmatrix},
	\quad
	A_3
	&=&
	\begin{pmatrix}
		-1 & 0\\
		0 & 0
	\end{pmatrix}.
\end{eqnarray}

%%%%%%%%%%%%%%%%%%%%%%%%%%%%%%%%%%%%%%%%%%%%%%%%%%%%%%%%%%%%
%%%%%%%%%%%%%%%%%%%%%%%%%%%%%%%%%%%%%%%%%%%%%%%%%%%%%%%%%%%%

\subsubsection{Second coordinate system}

%%%%%%%%%%%%%%%%%%%%%%%%%%%%%%%%%%%%%%%%%%%%%%%%%%%%%%%%%%%%
%%%%%%%%%%%%%%%%%%%%%%%%%%%%%%%%%%%%%%%%%%%%%%%%%%%%%%%%%%%%
The same computations applied directly to the system \eqref{eq:eigensystem} yield the recursion $\Phi'_{n+1} = \ttransn \Phi'_n$ with
\begin{eqnarray}\label{transfer'}
	\ttransn
	=
	\begin{pmatrix}
		\pno\pnt+1 & \pno \\
		\pnt & 1
	\end{pmatrix}.
\end{eqnarray}
%%%%%%%%%%%%%%%%%%%%%%%%%%%%%%%%%%%%%%%%%%%%%%%%%%%%%%%%%%%%
%%%%%%%%%%%%%%%%%%%%%%%%%%%%%%%%%%%%%%%%%%%%%%%%%%%%%%%%%%%%
Setting the disorder to be $0$, we obtain the transfer matrix for the free system 
\begin{eqnarray*}
	\tT=\tT(E)
	=
	\begin{pmatrix}
		p_1p_2 + 1 & p_1 \\
		p_2 & 1
	\end{pmatrix}.
\end{eqnarray*}
%%%%%%%%%%%%%%%%%%%%%%%%%%%%%%%%%%%%%%%%%%%%%%%%%%%%%%%%%%%%
%%%%%%%%%%%%%%%%%%%%%%%%%%%%%%%%%%%%%%%%%%%%%%%%%%%%%%%%%%%%
Once again, we can see that the transfer matrix $\ttransn$ admits the decomposition
\begin{eqnarray}\label{decomp'}
	\ttransn
	&=&
	\tT + \vo(n) A'_1 + \vt(n+1) A'_2 +\vo(n)\vt(n+1)A'_3,
\end{eqnarray}
with
\begin{eqnarray}\label{A123'}
	\tA_1
	=
	\begin{pmatrix}
		p_2 & 1\\
		0 & 0
	\end{pmatrix},
	\quad
	\tA_2
	=
	\begin{pmatrix}
		-p_1 & 0\\
		-1 & 0
	\end{pmatrix},
	\quad
	\tA_3
	=
	\begin{pmatrix}
		-1 & 0 \\
		0 & 0
	\end{pmatrix}.
\end{eqnarray}
%%%%%%%%%%%%%%%%%%%%%%%%%%%%%%%%%%%%%%%%%%%%%%%%%%%%%%%%%%%%
%%%%%%%%%%%%%%%%%%%%%%%%%%%%%%%%%%%%%%%%%%%%%%%%%%%%%%%%%%%%

%%%%%%%%%%%%%%%%%%%%%%%%%%%%%%%%%%%%%%%%%%%%%%%%%%%%%%%%%%%%
%%%%%%%%%%%%%%%%%%%%%%%%%%%%%%%%%%%%%%%%%%%%%%%%%%%%%%%%%%%%
%%%%%%%%%%%%%%%%%%%%%%%%%%%%%%%%%%%%%%%%%%%%%%%%%%%%%%%%%%%%
%%%%%%%%%%%%%%%%%%%%%%%%%%%%%%%%%%%%%%%%%%%%%%%%%%%%%%%%%%%%

\subsection{Natural basis}

%%%%%%%%%%%%%%%%%%%%%%%%%%%%%%%%%%%%%%%%%%%%%%%%%%%%%%%%%%%%
%%%%%%%%%%%%%%%%%%%%%%%%%%%%%%%%%%%%%%%%%%%%%%%%%%%%%%%%%%%%
Our next task consists in finding a suitable coordinate system where the transfer matrices can be written as a perturbation of the identity. These coordinate systems will be obtained by diagonalization of the free transfer matrices $T$ and $\tT$. We detail the computations for the first coordinate system and only write the final results for the second one since the arguments are identical.
%%%%%%%%%%%%%%%%%%%%%%%%%%%%%%%%%%%%%%%%%%%%%%%%%%%%%%%%%%%%
%%%%%%%%%%%%%%%%%%%%%%%%%%%%%%%%%%%%%%%%%%%%%%%%%%%%%%%%%%%%

\subsubsection{First coordinate system}

%%%%%%%%%%%%%%%%%%%%%%%%%%%%%%%%%%%%%%%%%%%%%%%%%%%%%%%%%%%%
%%%%%%%%%%%%%%%%%%%%%%%%%%%%%%%%%%%%%%%%%%%%%%%%%%%%%%%%%%%%
Recall that
\begin{eqnarray*}
	T=T(E)
	=
	\begin{pmatrix}
		p_1p_2 + 1 & p_2 \\
		p_1 & 1
	\end{pmatrix}.
\end{eqnarray*}
The starting point of our analysis is the observation that
\begin{eqnarray*}
	T
	=
	\frac{1}{p_1p_2}
	M^2
	\quad 
	\text{with}
	\quad
	M
	=
	\begin{pmatrix}
		p_1p_2 & p_2 \\
		p_1 & 0
	\end{pmatrix}
\end{eqnarray*}
Let us diagonalize $M$. The solutions of the characteric equation $\det (M-xI)=0$ are given by
\begin{eqnarray*}
	x 
	&=& 
	\frac{p_1p_2 \pm \sqrt{p_1^2p_2^2+4p_1p_2}}{2}
	=
	\frac{p_1p_2 \pm \sqrt{p_1p_2(p_1p_2+4)}}{2}.
\end{eqnarray*}
Note that $p_1p_2+4 = m^2+4-E^2 \geq 0$ while $p_1p_2<0$. Hence
\begin{eqnarray*}
	x
	=
	\sqrt{-p_1p_2}
	\left(
		-\frac{\sqrt{-p_1p_2}}{2}
		\pm
		\frac{i}{2}
		\sqrt{p_1p_2 + 4}
	\right)
	=:
	\sqrt{-p_1p_2} \ \e^{\pm ik},
\end{eqnarray*}
where $\cos k = -\frac{\sqrt{-p_1p_2}}{2}$.
This shows that the eigenvalues of $T$ are equal to $-\e^{\pm 2ik}$.
%%%%%%%%%%%%%%%%%%%%%%%%%%%%%%%%%%%%%%%%%%%%%%%%%%%%%%%%%%%%
%%%%%%%%%%%%%%%%%%%%%%%%%%%%%%%%%%%%%%%%%%%%%%%%%%%%%%%%%%%%
\begin{remark}\label{rk:range-k}
	In particular one has $\cos k <0$. This still leaves us the freedom to choose the sign of $\sin k$. Below, we will need to choose it in such a way that $\sin 2k >0$. Hence, we assume $\sin k < 0$. In other words, we take $k\in(-\pi,-\tfrac{\pi}{2})$. Notice that 
	$$\displaystyle \sin(2k)=\frac12 \sqrt{(E^2-m^2)(m^2+4-E^2)},$$
	 vanishes exactly on the four edges of the spectrum. Moreover, we have $E^2=m^2 + 4\cos^2 k$.
\end{remark}
%%%%%%%%%%%%%%%%%%%%%%%%%%%%%%%%%%%%%%%%%%%%%%%%%%%%%%%%%%%%
%%%%%%%%%%%%%%%%%%%%%%%%%%%%%%%%%%%%%%%%%%%%%%%%%%%%%%%%%%%%
The coordinates of the eigenvectors of $M$ (and hence of $T$) must satisfy the equation
\begin{eqnarray*}
	p_1p_2a + p_2b &=&\sqrt{-p_1p_2} e^{\pm ik} a
	\\
	p_1 a &=& \sqrt{-p_1p_2}e^{\pm ik} b.
\end{eqnarray*}
Recalling our assumption $p_1<0$ and $p_2>0$, we can choose them as
\begin{eqnarray*}
	v_{\pm}
	=
	\begin{pmatrix}
		-\sqrt{p_2}\ \e^{\pm ik} \\ \sqrt{-p_1}
	\end{pmatrix}.
\end{eqnarray*}
We can now generate our natural basis.
%\begin{eqnarray}
%	\ppru_1
%	=
%	\begin{pmatrix}
%		-\sqrt{p_2}\cos k & -\sqrt{p_2} \sin k \\
%		\sqrt{-p_1} & 0
%	\end{pmatrix}
%\end{eqnarray}
Let $v(1)=v_+$ and $v(n+1)=Tv(n)$. We define
\begin{eqnarray*}
	\ppru_n 
	=
	\begin{pmatrix}
		\Re v(n) & \Im v(n)
\end{pmatrix}.	 
\end{eqnarray*}
Since $v(1)$ is an eigenvector of $T$ with eigenvalue $-\e^{2ik}$, we have
\begin{eqnarray*}
	v(n) = T^{n-1} v(1) = (-1)^{n-1} \e^{2i(n-1)k}v(1)
	=
	(-1)^{n-1}
	\begin{pmatrix}
		-\sqrt{p_2}\ \e^{i(2n-1)k} \\ \sqrt{-p_1}\ \e^{i(2n-2)k}
	\end{pmatrix}.
\end{eqnarray*}
Hence, our matrix of change of basis is given by
\begin{eqnarray}\label{eq:change-of-basis-first-system}
	\ppru_n
	=
	(-1)^{n-1}
	\begin{pmatrix}
		-\sqrt{p_2} \cos((2n-1)k) & -\sqrt{p_2} \sin((2n-1)k)\\
		\sqrt{-p_1} \cos((2n-2)k) & \sqrt{-p_1} \sin((2n-2)k)
	\end{pmatrix}.
\end{eqnarray}
and satisfies $\ppru_{n+1} = T \ppru_n$.
%%%%%%%%%%%%%%%%%%%%%%%%%%%%%%%%%%%%%%%%%%%%%%%%%%%%%%%%%%%%
%%%%%%%%%%%%%%%%%%%%%%%%%%%%%%%%%%%%%%%%%%%%%%%%%%%%%%%%%%%%

	If $p_1>0$ and $p_2<0$, the eigenvectors of $T$ can be taken as
	\begin{eqnarray*}
	v_{\pm}
	=
	\begin{pmatrix}
		\sqrt{-p_2} \ \e^{\pm ik} \\
		\sqrt{p_1}
	\end{pmatrix},
	\end{eqnarray*}
	yielding the change of basis
	\begin{eqnarray*}
		\qppru_n
		=
		(-1)^{n-1}
		\begin{pmatrix}
			\sqrt{-p_2}\cos((2n-1)k) & \sqrt{-p_2}\cos((2n-1)k)
			\\
			\sqrt{p_1}\cos((2n-2)k) & \sqrt{p_1}\cos((2n-2)k)
		\end{pmatrix},
	\end{eqnarray*}
	which again satisfies $\qppru_{n+1}=T \qppru_n$.

%%%%%%%%%%%%%%%%%%%%%%%%%%%%%%%%%%%%%%%%%%%%%%%%%%%%%%%%%%%%
%%%%%%%%%%%%%%%%%%%%%%%%%%%%%%%%%%%%%%%%%%%%%%%%%%%%%%%%%%%%

%%%%%%%%%%%%%%%%%%%%%%%%%%%%%%%%%%%%%%%%%%%%%%%%%%%%%%%%%%%%
%%%%%%%%%%%%%%%%%%%%%%%%%%%%%%%%%%%%%%%%%%%%%%%%%%%%%%%%%%%%

\subsubsection{Second coordinate system}

The same analysis can be performed with the transfer matrix
\begin{eqnarray*}
	\tT=\tT(E)
	=
	\begin{pmatrix}
		p_1p_2 + 1 & p_1 \\
		p_2 & 1
	\end{pmatrix},
\end{eqnarray*}
yielding the change of basis
\begin{eqnarray*}
	\pru_{\!\!  n}
	=
	(-1)^{n-1}
	\begin{pmatrix}
		\sqrt{-p_1} \cos((2n-1)k) & \sqrt{-p_1} \sin((2n-1)k)\\
		\sqrt{p_2} \cos((2n-2)k) & \sqrt{p_2} \sin((2n-2)k)
	\end{pmatrix},
\end{eqnarray*}
for $p_1<0$ and $p_2>0$, and
\begin{eqnarray*}
	\qpru_n
	=
	(-1)^{n-1}
	\begin{pmatrix}
		-\sqrt{p_1} \cos((2n-1)k) & -\sqrt{p_1} \sin((2n-1)k)\\
		\sqrt{-p2} \cos((2n-2)k) & \sqrt{-p2} \sin((2n-2)k)
	\end{pmatrix},
\end{eqnarray*}
for $p_1>0$ and $p_2<0$.
%%%%%%%%%%%%%%%%%%%%%%%%%%%%%%%%%%%%%%%%%%%%%%%%%%%%%%%%%%%%
%%%%%%%%%%%%%%%%%%%%%%%%%%%%%%%%%%%%%%%%%%%%%%%%%%%%%%%%%%%%
%%%%%%%%%%%%%%%%%%%%%%%%%%%%%%%%%%%%%%%%%%%%%%%%%%%%%%%%%%%%
%%%%%%%%%%%%%%%%%%%%%%%%%%%%%%%%%%%%%%%%%%%%%%%%%%%%%%%%%%%%

\subsection{Transfer matrices in the natural basis}

%%%%%%%%%%%%%%%%%%%%%%%%%%%%%%%%%%%%%%%%%%%%%%%%%%%%%%%%%%%%
%%%%%%%%%%%%%%%%%%%%%%%%%%%%%%%%%%%%%%%%%%%%%%%%%%%%%%%%%%%%
Recall the decomposition \eqref{eq:decomposition-transfer-matrices}:
\begin{eqnarray*}
	T_{\omega,n}
	=
	T + \vo(n) A_1 + \vt(n+1) A_2 +\vo(n)\vt(n+1)A_3.
\end{eqnarray*}
The objective is to write $T_{\omega,n}$ in each of the natural basis introduced above. Let us illustrate this in the basis $\ppru_n$. We introduce new coordinates $\Psi_n$ in such a way that $\Phi_n = \ppru_n \Psi_n$. Hence, the recursion for $\Psi_n$ becomes
\begin{eqnarray*}
	\Psi_{n+1} = \ppru_{n+1}^{-1} \Phi_{n+1} = \ppru_{n+1}^{-1} T_{\omega, n} \Phi_n = \ppru_{n+1}^{-1} T_{\omega, n} \ppru_n \Psi_n.
\end{eqnarray*}
Summarizing, we write
\begin{eqnarray*}
	\Psi_{n+1} = M_{\omega, n} \Psi_n \quad \text{with} \quad M_{\omega, n}= \ppru_{n+1}^{-1} T_{\omega, n} \ppru_n.
\end{eqnarray*}
Note that $\ppru_{n+1}^{-1} T\ \ppru_n=\ppru_{n+1}^{-1}\ppru_{n+1}=I$. Hence,
\begin{eqnarray*}
	M_{\omega,n}
	=
	I + \vo(n) B_{n,1} + \vt(n+1) B_{n,2} +\vo(n)\vt(n+1)B_{n,3},
\end{eqnarray*}
with
\begin{eqnarray}\label{B123}
	B_{n,1} = \ppru_{n+1}^{-1} A_1 \ppru_n,
	\quad
	B_{n,2} = \ppru_{n+1}^{-1} A_2 \ppru_n,
	\quad
	B_{n,3} = \ppru_{n+1}^{-1} A_3 \ppru_n.
\end{eqnarray}
%%%%%%%%%%%%%%%%%%%%%%%%%%%%%%%%%%%%%%%%%%%%%%%%%%%%%%%%%%%%
%%%%%%%%%%%%%%%%%%%%%%%%%%%%%%%%%%%%%%%%%%%%%%%%%%%%%%%%%%%%

The computation of the matrices in \eqref{B123} is lengthy but rather straightforward. However, we present some linear algebraic preliminaries which make them quicker and more elegant. 
We will identify vectors and complex numbers in the usual way as
$\displaystyle
v
=
\begin{pmatrix}
	v_1 \\ v_2
\end{pmatrix}
=
v_1 + iv_2
$.
For two vectors $v$ and $w$, we denote by
$
\displaystyle
\begin{pmatrix}
	v \\ w
\end{pmatrix}
$
and
$
\displaystyle
\begin{pmatrix}
	v & w
\end{pmatrix}
$
the matrices whose rows and columns are $v$ and $w$ respectively. 
A matrix with rows $v$ and $w$ can be written in terms of projections as
\begin{eqnarray*}
	M
	=
	\begin{pmatrix}
		v \\ w
	\end{pmatrix}
	=
	\begin{pmatrix}
		\langle v|
		\\
		\langle w|
	\end{pmatrix}.
\end{eqnarray*}
%\begin{eqnarray}
%	M
%	=
%	\begin{pmatrix}
%		v \\ w
%	\end{pmatrix}
%	=
%	\begin{pmatrix}
%		v_1 & v_2 \\
%		w_1 & w_2
%	\end{pmatrix}
%\end{eqnarray}
%then, $M$ can be represented as
%\begin{eqnarray}
%	M
%	=
%	\begin{pmatrix}
%		\langle v|
%		\\
%		\langle w|
%	\end{pmatrix}
%\end{eqnarray}
Noting that $w_2 - i w_1 = -i(w_1+iw_2)$, the inverse matrix can be represented by columns as
\begin{eqnarray*}
	M^{-1}
	=
	\frac{1}{\det ( v \,\, w)}
	\begin{pmatrix}
		w_2 & -v_2 \\
		-w_1 & v_1
	\end{pmatrix}
	=
	\frac{i}{\det ( v \, \, w)}
	( -w \,\, v).
\end{eqnarray*}
%Now, note that $w_2 - i w_1 = -i(w_1+iw_2)$ so that
%\begin{eqnarray}
%	M^{-1}
%	=
%	\frac{i}{\det ( v \, \, w)}
%	( -w \,\, v)
%\end{eqnarray}
%%%%%%%%%%%%%%%%%%%%%%%%%%%%%%%%%%%%%%%%%%%%%%%%%%%%%%%%%%%%
%%%%%%%%%%%%%%%%%%%%%%%%%%%%%%%%%%%%%%%%%%%%%%%%%%%%%%%%%%%%

\subsubsection{First coordinate system}

%%%%%%%%%%%%%%%%%%%%%%%%%%%%%%%%%%%%%%%%%%%%%%%%%%%%%%%%%%%%
%%%%%%%%%%%%%%%%%%%%%%%%%%%%%%%%%%%%%%%%%%%%%%%%%%%%%%%%%%%%

Recall the change of basis matrices \eqref{eq:change-of-basis-first-system}.
Letting
$$\displaystyle 
v_m 
= 
\begin{pmatrix}
	\cos(mk) \\ \sin(mk) 
\end{pmatrix}
=
\e^{imk}
,$$
allows us to represent $\ppru_n$ as
\begin{eqnarray*}
	\ppru_n
	=
	(-1)^{n-1}
	\begin{pmatrix}
		-\sqrt{p_2} \langle v_{2n-1}| \\ \\
		\sqrt{-p_1} \langle v_{2n-2}|
	\end{pmatrix}.
\end{eqnarray*}
Now, since $\det \ppru_n=\det (T^{n-1} \ppru_1)=\det \ppru_1=-\sin(2k)$, the inverse of $\ppru_n$ is given by
\begin{eqnarray*}
	\ppru_n^{-1}
	=
	\frac{(-1)^{n-1}i}{\sin(2k)}
	\left(\sqrt{-p_1}v_{2n-2} \,\, \sqrt{p_2} v_{2n-1} \right).
\end{eqnarray*}
%Actually, we will use
%\begin{eqnarray}
%	\ppru_{n+1}^{-1}
%	=
%	\frac{(-1)^{n}i}{\sin(2k)}
%	\left(\sqrt{-p_1}v_{2n} \,\, \sqrt{p_2} v_{2n+1} \right)
%\end{eqnarray}
Let us compute the matrices $B_{n,1},\, B_{n,2}$ and $B_{n,3}$ following this formalism.
We start with the simplest matrix which is $B_{n,3}$:
\begin{eqnarray}\label{B1}
	B_{n,3}
	&=&
	\ppru_{n+1}^{-1}A_3\ppru_n
	=
	(-1)^{n-1}
	\ppru_{n+1}^{-1}
	\begin{pmatrix}
		-1 & 0 \\
		0 & 0
	\end{pmatrix}
	\begin{pmatrix}
		-\sqrt{p_2} \langle v_{2n-1}| \\ \\
		\sqrt{-p_1} \langle v_{2n-2}|
	\end{pmatrix}
	\notag
	\\
	&=&
	\frac{-i}{\sin(2k)}
	\left(\sqrt{-p_1}v_{2n} \,\, \sqrt{p_2} v_{2n+1} \right)
	\begin{pmatrix}
		\sqrt{p_2} \langle v_{2n-1}| \\ \\
		0
	\end{pmatrix}
	\\
	\notag
	&=&
	\frac{-i}{\sin(2k)} \sqrt{-p_1p_2}
	v_{2n}\langle v_{2n-1}|
	=
	\frac{i}{\sin k} v_{2n}\langle v_{2n-1}|,
\end{eqnarray}
where we used the relation $\sqrt{-p_1p_2}=-2\cos k$.
By a similar procedure, we obtain
\begin{eqnarray*}
	B_{n,1}
	&=&
%	\ppru_{n+1}^{-1}A_1\ppru_n
%	=
%	\ppru_{n+1}^{-1}
%	\begin{pmatrix}
%		p_2 & 0 \\
%		1 & 0
%	\end{pmatrix}
%	\ppru_n
%	\\
%	&=&
%	(-1)^{n-1}
%	\ppru_{n+1}^{-1}
%	\begin{pmatrix}
%		p_2 & 0 \\
%		1 & 0
%	\end{pmatrix}
%	\begin{pmatrix}
%		-\sqrt{p_2} \langle v_{2n-1}| \\
%		\sqrt{-p_1} \langle v_{2n-2}|
%	\end{pmatrix}
%	\\
%	&=&
%	\frac{i \sqrt{p_2}}{\sin(2k)}
%	\left(
%		\sqrt{-p_1} v_{2n} \, \,  \sqrt{p_2} v_{2n+1}
%	\right)
%	\begin{pmatrix}
%		p_2\langle v_{2n-1}| \\
%		\langle v_{2n-1}|
%	\end{pmatrix}
%	\\
%	&=&
%	\frac{ip_2}{\sin(2k)}
%	\left(
%		\sqrt{-p_1p_2} v_{2n} + v_{2n+1} 
%	\right)
%	\langle v_{2n-1}|
%	\\
%	&=&
	-\frac{ip_2}{\sin(2k)}
	\left(
		2 \cos(k)v_{2n} - v_{2n+1} 
	\right)
	\langle v_{2n-1}|.
\end{eqnarray*}
At this point, we use Chebyshev's identity $v_m = 2\cos(k) \, v_{m-1}-v_{m-2}$ which yields
%\begin{eqnarray}\label{eq:chebyshev}
%	v_m = 2\cos(k) \, v_{m-1}-v_{m-2},
%\end{eqnarray} 
%$2 \cos k \, v_{2n} - v_{2n+1}=v_{2n-1}$. Hence,
%\begin{eqnarray}
%	2 \cos k \, v_{2n} - v_{2n+1}
%	=
%	2 \cos k  \, v_{2n} - v_{2n-1} + v_{2n-1} - v_{2n+1}
%	=
%	v_{2n-1}
%\end{eqnarray}
\begin{eqnarray*}
	B_{n,1}
	=
	-\frac{ip_2}{\sin(2k)}
	v_{2n-1}
	\langle v_{2n-1}|.
\end{eqnarray*}
Similarly,
\begin{eqnarray}\label{B2}
	B_{n,2}
	&=&
%	-\frac{i \sqrt{-p_1}}{\sin(2k)}
%	\left(
%		\sqrt{-p_1} v_{2n} \, \,  \sqrt{p_2} v_{2n+1}
%	\right)
%	\begin{pmatrix}
%		\langle v_{2n} | \\ 0
%	\end{pmatrix}
%	\\
%	&=&
	\frac{i p_1}{\sin(2k)}
	v_{2n} \langle v_{2n} |.
\end{eqnarray}
%%%%%%%%%%%%%%%%%%%%%%%%%%%%%%%%%%%%%%%%%%%%%%%%%%%%%%%%%%%%
%%%%%%%%%%%%%%%%%%%%%%%%%%%%%%%%%%%%%%%%%%%%%%%%%%%%%%%%%%%%

\subsubsection{Second coordinate system}

%%%%%%%%%%%%%%%%%%%%%%%%%%%%%%%%%%%%%%%%%%%%%%%%%%%%%%%%%%%%
%%%%%%%%%%%%%%%%%%%%%%%%%%%%%%%%%%%%%%%%%%%%%%%%%%%%%%%%%%%%
The same computations as above yield the following matrices corresponding to the coordinate system $	\ppru_n' $:

\begin{eqnarray}\label{B123'}
	&&\tB_{n,1}
	=
	-\frac{ip_2}{\sin(2k)}v_{2n}\langle v_{2n}|,
	\quad
	\tB_{n,2}
	=
	-\frac{ip_1}{\sin(2k)} v_{2n-1}\langle v_{2n-1}|,
	\notag
	\\
	&&\tB_{n,3}
	=
	-\frac{i}{\sin k}v_{2n}\langle v_{2n-1}|.
\end{eqnarray}
%%%%%%%%%%%%%%%%%%%%%%%%%%%%%%%%%%%%%%%%%%%%%%%%%%%%%%%%%%%%
%%%%%%%%%%%%%%%%%%%%%%%%%%%%%%%%%%%%%%%%%%%%%%%%%%%%%%%%%%%%
%%%%%%%%%%%%%%%%%%%%%%%%%%%%%%%%%%%%%%%%%%%%%%%%%%%%%%%%%%%%
%%%%%%%%%%%%%%%%%%%%%%%%%%%%%%%%%%%%%%%%%%%%%%%%%%%%%%%%%%%%
%%%%%%%%%%%%%%%%%%%%%%%%%%%%%%%%%%%%%%%%%%%%%%%%%%%%%%%%%%%%
%%%%%%%%%%%%%%%%%%%%%%%%%%%%%%%%%%%%%%%%%%%%%%%%%%%%%%%%%%%%
%%%%%%%%%%%%%%%%%%%%%%%%%%%%%%%%%%%%%%%%%%%%%%%%%%%%%%%%%%%%
%%%%%%%%%%%%%%%%%%%%%%%%%%%%%%%%%%%%%%%%%%%%%%%%%%%%%%%%%%%%

\subsection{The Pr\"ufer transform}\label{sec:the-transform}

%%%%%%%%%%%%%%%%%%%%%%%%%%%%%%%%%%%%%%%%%%%%%%%%%%%%%%%%%%%%
%%%%%%%%%%%%%%%%%%%%%%%%%%%%%%%%%%%%%%%%%%%%%%%%%%%%%%%%%%%%

The matrices computed above can be represented as an explicit transformation on the complex plane. Recalling the correspondence between vectors and complex numbers, we have
\begin{eqnarray}\label{eq:matrix-to-transform}
	\e^{i\alpha}\langle \e^{i\beta} | \e^{i\theta}
	=
	\cos(\theta-\beta)\ \e^{i\alpha}
	=
	\cos(\theta-\beta)\ \e^{-i(\theta-\alpha)}\e^{i\theta}.
\end{eqnarray}

%%%%%%%%%%%%%%%%%%%%%%%%%%%%%%%%%%%%%%%%%%%%%%%%%%%%%%%%%%%%
%%%%%%%%%%%%%%%%%%%%%%%%%%%%%%%%%%%%%%%%%%%%%%%%%%%%%%%%%%%%
%%%%%%%%%%%%%%%%%%%%%%%%%%%%%%%%%%%%%%%%%%%%%%%%%%%%%%%%%%%%
%%%%%%%%%%%%%%%%%%%%%%%%%%%%%%%%%%%%%%%%%%%%%%%%%%%%%%%%%%%%
\subsubsection{First coordinate system}
%%%%%%%%%%%%%%%%%%%%%%%%%%%%%%%%%%%%%%%%%%%%%%%%%%%%%%%%%%%%
%%%%%%%%%%%%%%%%%%%%%%%%%%%%%%%%%%%%%%%%%%%%%%%%%%%%%%%%%%%%
With the representation mentioned above, the matrices $B_{n,i}$ for $i\in\{1,2,3\}$ can be expressed as 
%%%%%%%%%%%%%%%%%%%%%%%%%%%%%%%%%%%%%%%%%%%%%%%%%%%%%%%%%%%%
%%%%%%%%%%%%%%%%%%%%%%%%%%%%%%%%%%%%%%%%%%%%%%%%%%%%%%%%%%%%
\begin{eqnarray*}
	B_{n,1}\, \e^{i\theta}
	&=&
	-\frac{ip_2}{\sin(2k)}
	\cos \bar\theta \, 
	\e^{-i\bar\theta}
	\e^{i\theta},
	\\
	B_{n,2} \, \e^{i\theta}
	&=&
	\frac{i p_1}{\sin(2k)}
	\cos( \bar\theta-k)
	\, \e^{-i (\bar\theta-k)}
	\e^{i\theta},
	\\
	B_{n,3} \, \e^{i\theta}
	&=&
	\frac{i}{\sin k}
	\cos \bar\theta \,
	\e^{-i (\bar \theta-k)}
	\e^{i\theta},
\end{eqnarray*}
where $\bar\theta =\theta- (2n-1)k$.
%%%%%%%%%%%%%%%%%%%%%%%%%%%%%%%%%%%%%%%%%%%%%%%%%%%%%%%%%%%%
%%%%%%%%%%%%%%%%%%%%%%%%%%%%%%%%%%%%%%%%%%%%%%%%%%%%%%%%%%%%

Recall that we defined a new representation of $\Phi$ through the relation $\Phi_n = \ppru_n \Psi_n$.
We write $\Psi_n$ in terms of its coordinates $\psi_n^\pm$ and
introduce new variables $\zeta_n,\, R_n$ and $\theta_n$ through the relation $\zeta_n = \psi_n^++i\psi^-_n=R_n \e^{i\theta_n}$.
We also define $\bar \theta_n =  \theta_n-(2n-1)k$. These are called the Pr\"ufer coordinates. Recalling that
\begin{eqnarray}\label{dec eigenfunction B123}
	\Phi_{n+1} 
	= 
	\Big{(}
		I + \vo(n) B_{n,1} + \vt(n+1)B_{n,2} + \vo(n)\vt(n+1)B_{3,n}	
	\Big{)}
	\Phi_n,
\end{eqnarray}
the recursion for the Pr\"ufer coordinates becomes
\begin{eqnarray}
	\nonumber
	\zeta_{n+1}
	&=&
	\Big{(}
	1
	-\frac{ip_2}{\sin(2k)} \vo(n)
	\cos \bar\theta_n \, 
	\e^{-i\bar\theta_n}
	+
	\frac{i p_1}{\sin(2k)}\vt(n+1)
	\cos( \bar\theta_n-k)
	\, \e^{-i (\bar\theta_n-k)}
	\\
		\label{eq:prufer-transform-s1}
	&&
	\phantom{blablablablablablabla}
	+
	\frac{i}{\sin k}\vo(n)\vt(n+1)
	\cos \bar\theta_n \,
	\e^{-i (\bar \theta_n-k)}
	\Big{)}
	\,
	\zeta_n.
	\end{eqnarray}
%%%%%%%%%%%%%%%%%%%%%%%%%%%%%%%%%%%%%%%%%%%%%%%%%%%%%%%%%%%%
%%%%%%%%%%%%%%%%%%%%%%%%%%%%%%%%%%%%%%%%%%%%%%%%%%%%%%%%%%%%
\begin{remark}\label{rk:measurability-s1}
	Let $\mathcal{F}_n$ be the $\sigma-$algebra defined by 
	$\mathcal{F}_n=\sigma(\vo(j),\, \vt(j+1):\, 1\leq j \leq n)$. The previous representation shows that the variables $\zeta_n,\, R_n,\,\theta$ and $\bar{\theta}_n$ are measurable with respect to $\mathcal{F}_{n-1}$ (as a consequence, so is $\Phi_n$, a fact that could be read from the original system of equations). In particular, they are independent of $\vo(n)$ and $\vt(n+1)$. This fact will be crucial in our analysis as it will turn certain objects into martingales with respect to the filtration $(\mathcal{F}_n)_n$.
\end{remark}
%%%%%%%%%%%%%%%%%%%%%%%%%%%%%%%%%%%%%%%%%%%%%%%%%%%%%%%%%%%%
%%%%%%%%%%%%%%%%%%%%%%%%%%%%%%%%%%%%%%%%%%%%%%%%%%%%%%%%%%%%

%%%%%%%%%%%%%%%%%%%%%%%%%%%%%%%%%%%%%%%%%%%%%%%%%%%%%%%%%%%%
%%%%%%%%%%%%%%%%%%%%%%%%%%%%%%%%%%%%%%%%%%%%%%%%%%%%%%%%%%%%
%%%%%%%%%%%%%%%%%%%%%%%%%%%%%%%%%%%%%%%%%%%%%%%%%%%%%%%%%%%%
%%%%%%%%%%%%%%%%%%%%%%%%%%%%%%%%%%%%%%%%%%%%%%%%%%%%%%%%%%%%
\subsubsection{Second coordinate system}

%%%%%%%%%%%%%%%%%%%%%%%%%%%%%%%%%%%%%%%%%%%%%%%%%%%%%%%%%%%%
%%%%%%%%%%%%%%%%%%%%%%%%%%%%%%%%%%%%%%%%%%%%%%%%%%%%%%%%%%%%

Analogously, we define $\Psi'_n$ such that $\Phi'_n = \pru_n \Psi'_n$ and introduce a new set of corresponding Pr\"ufer variables $\zeta'_n = R'_n \e^{i\theta'_n}$. Setting $\bar{\theta}'_n = \theta'_n - (2n-1)k$, we obtain the recursion
\begin{eqnarray}
	\nonumber
	\zeta'_{n+1}
	&=&
	\Big{(}
	1
	-\frac{ip_2}{\sin(2k)} \vo(n)
	\cos( \bar\theta_n-k) \, 
	\e^{-i(\bar{\theta}'_n-k)}
	+
	\frac{i p_1}{\sin(2k)}\vt(n)
	\cos \bar{\theta}'_n
	\, \e^{-i \bar{\theta}'_n}
	\\
		\label{eq:prufer-transform-s2}
	&&
	\phantom{blablablablablablabla}
	+
	\frac{i}{\sin k}\vo(n)\vt(n)
	\cos \bar{\theta}'_n \,
	\e^{-i (\bar{\theta}'_n-k)}
	\Big{)}
	\,
	\zeta'_n.
	\end{eqnarray}
%%%%%%%%%%%%%%%%%%%%%%%%%%%%%%%%%%%%%%%%%%%%%%%%%%%%%%%%%%%%
%%%%%%%%%%%%%%%%%%%%%%%%%%%%%%%%%%%%%%%%%%%%%%%%%%%%%%%%%%%%
\begin{remark}\label{rk:measurability-s2}
	We can see that the new Pr\"ufer variables fulfill the same measurability properties highlighted in Remark \ref{rk:measurability-s1} with respect to the filtration $\mathcal{F}'_n=\sigma(\vo(j),\, \vt(j):\, 1\leq j \leq n)$.
\end{remark}
%%%%%%%%%%%%%%%%%%%%%%%%%%%%%%%%%%%%%%%%%%%%%%%%%%%%%%%%%%%%
%%%%%%%%%%%%%%%%%%%%%%%%%%%%%%%%%%%%%%%%%%%%%%%%%%%%%%%%%%%%

%%%%%%%%%%%%%%%%%%%%%%%%%%%%%%%%%%%%%%%%%%%%%%%%%%%%%%%%%%%%
%%%%%%%%%%%%%%%%%%%%%%%%%%%%%%%%%%%%%%%%%%%%%%%%%%%%%%%%%%%%
%%%%%%%%%%%%%%%%%%%%%%%%%%%%%%%%%%%%%%%%%%%%%%%%%%%%%%%%%%%%
%%%%%%%%%%%%%%%%%%%%%%%%%%%%%%%%%%%%%%%%%%%%%%%%%%%%%%%%%%%%
\subsection{Equivalence of systems}

%%%%%%%%%%%%%%%%%%%%%%%%%%%%%%%%%%%%%%%%%%%%%%%%%%%%%%%%%%%%
%%%%%%%%%%%%%%%%%%%%%%%%%%%%%%%%%%%%%%%%%%%%%%%%%%%%%%%%%%%%
The results of this section hold without any assumption on the potential.
%%%%%%%%%%%%%%%%%%%%%%%%%%%%%%%%%%%%%%%%%%%%%%%%%%%%%%%%%%%%
%%%%%%%%%%%%%%%%%%%%%%%%%%%%%%%%%%%%%%%%%%%%%%%%%%%%%%%%%%%%
The next lemma shows that the asymptotics of the systems in the original coordinates and Pr\"ufer coordinates are equivalent. We will write the result for the first system since the adaptation to the second system is straightforward.
%%%%%%%%%%%%%%%%%%%%%%%%%%%%%%%%%%%%%%%%%%%%%%%%%%%%%%%%%%%%
%%%%%%%%%%%%%%%%%%%%%%%%%%%%%%%%%%%%%%%%%%%%%%%%%%%%%%%%%%%%
\begin{lemma}\label{thm:equivalence-systems}
%Let $\Phi$ be an eigenfunction associated to an energy $E\in\mathring\Sigma$ expressed through Pr\"ufer variables. Then
Let $\Phi\in\h$ be expressed through its Pr\"ufer coordinates associated to a fixed energy $E\in\mathring\Sigma$. Then,
	\begin{eqnarray}\label{bnd Phi}
	\frac{\sin^2 (2k)}{2E} R_n^2 \leq \| \Phi_n \|^2 \leq 4E^2 R_n^2.
\end{eqnarray}
\end{lemma}
%%%%%%%%%%%%%%%%%%%%%%%%%%%%%%%%%%%%%%%%%%%%%%%%%%%%%%%%%%%%
%%%%%%%%%%%%%%%%%%%%%%%%%%%%%%%%%%%%%%%%%%%%%%%%%%%%%%%%%%%%
\begin{proof}
	We assume $E\in (m,\sqrt{m^2+4})$, the other case being similar. We first note that
\begin{eqnarray*}
	\text{Tr} (\ppru^*_n \ppru_n) 
	&=&
	-p_1 + p_2 = 2E \in (2m,2\sqrt{m^2+4}),
	\\
	\det (\ppru^*_n \ppru_n) 
	&=& \sin^2 (2k).
\end{eqnarray*}
Now, let $0<\gamma_1 \leq \gamma_2$ represent the eigenvalues of $\ppru^*_n \ppru_n$. Then,
\begin{eqnarray*}
	\gamma_1+\gamma_2 = 2E,\quad \quad \gamma_1 \gamma_2 = \sin^2(2k).
\end{eqnarray*}
For the lower bound, we have
\begin{eqnarray}\label{eq:lower-bound-eigenvalue}
	\gamma_1 = \frac{\sin^2 (2k)}{\gamma_2} = \frac{\sin^2 (2k)}{\text{Tr} (\ppru^*_n \ppru_n)-\gamma_1} \geq \frac{\sin^2 (2k)}{2E}.
\end{eqnarray}
The upper bound follows from $\| \ppru_n \|^2 = \lambda_2^2 \leq 4E^2$.
\end{proof}
%%%%%%%%%%%%%%%%%%%%%%%%%%%%%%%%%%%%%%%%%%%%%%%%%%%%%%%%%%%%
%%%%%%%%%%%%%%%%%%%%%%%%%%%%%%%%%%%%%%%%%%%%%%%%%%%%%%%%%%%%
It turns out that the Pr\"ufer radii also allow us to control the asymptotics of transfer matrices. Denote $\prodtransn = \transn \cdots T_{\omega,1}$ so that $\Phi_{n} = \prodtransnm \Phi_1$.
%and $\prufprodtransn = \pruftransn \cdots M_{\omega,1}$ so that, this time, $Y_{n} = \prufprodtransnm Y_0$. 
We write $R_n(\theta)$ when the recursion for $(\Psi_n)_n$ is started from $\Psi_1 = \e^{i\theta}$.
%$\widehat Y_1 = \hat \theta := (\cos \theta, \sin \theta)$. 
%%%%%%%%%%%%%%%%%%%%%%%%%%%%%%%%%%%%%%%%%%%%%%%%%%%%%%%%%%%%
%%%%%%%%%%%%%%%%%%%%%%%%%%%%%%%%%%%%%%%%%%%%%%%%%%%%%%%%%%%%
\begin{lemma} \label{thm:comparison}
For any pair of initial angles $\theta_1 \neq \theta_2$, there exist constants $c(\theta_1,\theta_2),\, C(\theta_1,\theta_2)$ such that
\begin{eqnarray*}
	c(\theta_1,\theta_2) \max\{R_n(\theta_1),\, R_n(\theta_2)\}
	\leq
	\| \mathbf{T}_{\omega, n-1} \|^2
	\leq	
	C(\theta_1,\theta_2) \max\{ R_n(\theta_1),\, R_n(\theta_2)\},
\end{eqnarray*}
for all $n\geq 1$.
\end{lemma}
%%%%%%%%%%%%%%%%%%%%%%%%%%%%%%%%%%%%%%%%%%%%%%%%%%%%%%%%%%%%
%%%%%%%%%%%%%%%%%%%%%%%%%%%%%%%%%%%%%%%%%%%%%%%%%%%%%%%%%%%%
\begin{proof}
	From the relation $\prodtransnm \Phi_1 = \ppru_n \Psi_n$ and \eqref{eq:lower-bound-eigenvalue}, we obtain
	\begin{eqnarray*}
		\| \prodtransnm \|^2 \geq \| \ppru_n \Psi_n \|^2 \geq \frac{\sin^2(2 k)}{2E} R_n^2.
	\end{eqnarray*}
	The upper bound is more delicate and follows from a general result on unimodular matrices that we defer to the appendix.
\end{proof}
%%%%%%%%%%%%%%%%%%%%%%%%%%%%%%%%%%%%%%%%%%%%%%%%%%%%%%%%%%%%
%%%%%%%%%%%%%%%%%%%%%%%%%%%%%%%%%%%%%%%%%%%%%%%%%%%%%%%%%%%%

%%%%%%%%%%%%%%%%%%%%%%%%%%%%%%%%%%%%%%%%%%%%%%%%%%%%%%%%%%%%
%%%%%%%%%%%%%%%%%%%%%%%%%%%%%%%%%%%%%%%%%%%%%%%%%%%%%%%%%%%%
%%%%%%%%%%%%%%%%%%%%%%%%%%%%%%%%%%%%%%%%%%%%%%%%%%%%%%%%%%%%
%%%%%%%%%%%%%%%%%%%%%%%%%%%%%%%%%%%%%%%%%%%%%%%%%%%%%%%%%%%%
%%%%%%%%%%%%%%%%%%%%%%%%%%%%%%%%%%%%%%%%%%%%%%%%%%%%%%%%%%%%
%%%%%%%%%%%%%%%%%%%%%%%%%%%%%%%%%%%%%%%%%%%%%%%%%%%%%%%%%%%%
%%%%%%%%%%%%%%%%%%%%%%%%%%%%%%%%%%%%%%%%%%%%%%%%%%%%%%%%%%%%
%%%%%%%%%%%%%%%%%%%%%%%%%%%%%%%%%%%%%%%%%%%%%%%%%%%%%%%%%%%%

\section{Asymptotics of transfer matrices}\label{sec:asymptotics-TM}

%%%%%%%%%%%%%%%%%%%%%%%%%%%%%%%%%%%%%%%%%%%%%%%%%%%%%%%%%%%%
%%%%%%%%%%%%%%%%%%%%%%%%%%%%%%%%%%%%%%%%%%%%%%%%%%%%%%%%%%%%
The results of this section are given for the transfer matrices in the first system. This will be enough for the proof of Theorem \ref{thm:spectrum} and \ref{thm:transport}. Theorem \ref{thm:DL} will require the corresponding estimates for both systems. Once again, the statements and proofs are identical.

%%%%%%%%%%%%%%%%%%%%%%%%%%%%%%%%%%%%%%%%%%%%%%%%%%%%%%%%%%%%
%%%%%%%%%%%%%%%%%%%%%%%%%%%%%%%%%%%%%%%%%%%%%%%%%%%%%%%%%%%%
To stress the dependence on the energy, we write $\transn(E)$ for the transfer matrix at a fixed energy $E$ and $\prodtransn(E)=\transn(E)\cdots T_{\omega,1}(E)$. Recall that the parameters $E$ and $k$ are linked through the relation $E = \sqrt{m^2 + 4\cos^2(k)}$ for $E\in(m,\sqrt{m^2+4})$ and $k\in(-\pi,-\tfrac{3\pi}{2})$. 
%%%%%%%%%%%%%%%%%%%%%%%%%%%%%%%%%%%%%%%%%%%%%%%%%%%%%%%%%%%%
%%%%%%%%%%%%%%%%%%%%%%%%%%%%%%%%%%%%%%%%%%%%%%%%%%%%%%%%%%%%
%%%%%%%%%%%%%%%%%%%%%%%%%%%%%%%%%%%%%%%%%%%%%%%%%%%%%%%%%%%%
%%%%%%%%%%%%%%%%%%%%%%%%%%%%%%%%%%%%%%%%%%%%%%%%%%%%%%%%%%%%
%%%%%%%%%%%%%%%%%%%%%%%%%%%%%%%%%%%%%%%%%%%%%%%%%%%%%%%%%%%%
%%%%%%%%%%%%%%%%%%%%%%%%%%%%%%%%%%%%%%%%%%%%%%%%%%%%%%%%%%%%
%%%%%%%%%%%%%%%%%%%%%%%%%%%%%%%%%%%%%%%%%%%%%%%%%%%%%%%%%%%%
%%%%%%%%%%%%%%%%%%%%%%%%%%%%%%%%%%%%%%%%%%%%%%%%%%%%%%%%%%%%

\subsection{Almost sure asymptotics}\label{sec:asymptotics-TM-as}

%%%%%%%%%%%%%%%%%%%%%%%%%%%%%%%%%%%%%%%%%%%%%%%%%%%%%%%%%%%%
%%%%%%%%%%%%%%%%%%%%%%%%%%%%%%%%%%%%%%%%%%%%%%%%%%%%%%%%%%%%
%%%%%%%%%%%%%%%%%%%%%%%%%%%%%%%%%%%%%%%%%%%%%%%%%%%%%%%%%%%%
%%%%%%%%%%%%%%%%%%%%%%%%%%%%%%%%%%%%%%%%%%%%%%%%%%%%%%%%%%%%

The following theorem gives the asymptotics of transfer matrices for the critical and sub-critical regime. It will be the key to our proof of spectral transition for $\alpha=\frac12$. For the proof of dynamical localization for $\alpha\in(0,\frac12)$, we will need an integrated version which is given in Section \ref{sec:asymptotics-TM-l1}.
%%%%%%%%%%%%%%%%%%%%%%%%%%%%%%%%%%%%%%%%%%%%%%%%%%%%%%%%%%%%
%%%%%%%%%%%%%%%%%%%%%%%%%%%%%%%%%%%%%%%%%%%%%%%%%%%%%%%%%%%%
\begin{theorem}\label{thm:lyapunov-exponents}
	Let $\alpha\in(0,\tfrac12]$.
	Assume \textbf{(A1)}-\textbf{(A3a)} and \textbf{(A4)}.
	For each fixed energy corresponding to a value of $k \in (-\pi, -\tfrac{\pi}{2})$ different from $-\frac{5\pi}{8},-\frac{3\pi}{4}$ and $-\frac{7\pi}{8}$, there exists a measurable set $\Omega_E \subset \Omega$ of full probability such that
	\begin{eqnarray}\label{beta}
		\beta
		=
		\beta(E,\lambda)
		:=
		\lim_{n\to\infty} \frac{\log \| \prodtransn(E) \|}{\sum^n_{j=1} j^{-2\alpha}}
		=
		\frac{\lambda^2( p_1^2+p_2^2)}{8 \sin^2(2k)},
	\end{eqnarray}
	for all $\omega\in\Omega_E$ and all $\lambda>0$.
\end{theorem}
%%%%%%%%%%%%%%%%%%%%%%%%%%%%%%%%%%%%%%%%%%%%%%%%%%%%%%%%%%%%
%%%%%%%%%%%%%%%%%%%%%%%%%%%%%%%%%%%%%%%%%%%%%%%%%%%%%%%%%%%%
\begin{proof}
	According to Lemma \ref{thm:comparison}, it is enough to show that
	\begin{eqnarray*}
		\lim_{n\to\infty} \frac{\log R_n^2}{\sum^n_{j=1} j^{-2\alpha}}
		=
		\frac{\lambda^2( p_1^2+p_2^2)}{4 \sin^2(2k)}.
	\end{eqnarray*}
	%%%%%%%%%%%%%%%%%%%%%%%%%%%%%%%%%%%%%%%%%%%%%%%%%%%%%%%%%%%%
	%%%%%%%%%%%%%%%%%%%%%%%%%%%%%%%%%%%%%%%%%%%%%%%%%%%%%%%%%%%%
	From \eqref{eq:prufer-transform-s1}, we obtain a recursion for the radii $(R_n)_n$ given by
	\begin{eqnarray}
		\nonumber
		R_{n+1}^2
		&=&
		\Big{(}
	1	
		+\frac{p_2^2}{\sin^2(2k)}\cos^2(\bar\theta_n)\vo(n)^2
		+\frac{p_1^2}{\sin^2(2k)}\cos^2(\bar\theta_n-k)\vt(n+1)^2
		\\
		\nonumber
		&&
		\quad
		+\frac{1}{\sin^2k}\cos^2(\bar\theta_n)\vo(n)^2\vt(n+1)^2
		\\
		\label{eq:prufer-transform-s1-radii}
		&&
		\quad
		-\frac{p_2}{\sin(2k)}\sin(2\bar\theta_n)\vo(n)
		+\frac{p_1}{\sin(2k)}\sin(2(\bar\theta_n-k))\vt(n+1)
		\\
		\nonumber
		&&
		\quad
		+ \frac{1}{\sin k}\cos(\bar\theta_n)\sin(\bar\theta_n-k)\vo(n)\vt(n+1)
		\\
		\nonumber
		&&
		\quad
		-2\frac{p_1p_2}{\sin^2(2k)}\cos k\cos(\bar\theta_n)\cos(\bar\theta_n-k) \vo(n)\vt(n+1)
		\\
		\nonumber
		&&
		\quad
		-2\frac{p_2}{\sin k \sin(2k)}\cos k\cos^2(\bar\theta_n)\vo(n)^2\vt(n+1)
		\\
		\nonumber
		&&
		\quad
		+2\frac{p_1}{\sin k\sin(2k)}\cos(\bar\theta_n)\cos(\bar\theta_n-k)\vo(n)\vt(n+1)^2
		\Big{)}
		R_n^2
		\\
		\nonumber
	&=:&
	\left( 1 + \Gamma_{\omega}(n)\right) \, R_n^2.
\end{eqnarray}
%%%%%%%%%%%%%%%%%%%%%%%%%%%%%%%%%%%%%%%%%%%%%%%%%%%%%%%%%%%%
%%%%%%%%%%%%%%%%%%%%%%%%%%%%%%%%%%%%%%%%%%%%%%%%%%%%%%%%%%%%
Note that $R_{n+1}^2,R_n^2>0$, $\p$-almost surely, so that $1 + \Gamma_{\omega}(n)>0$, $\p$-almost surely.
%%%%%%%%%%%%%%%%%%%%%%%%%%%%%%%%%%%%%%%%%%%%%%%%%%%%%%%%%%%%
%%%%%%%%%%%%%%%%%%%%%%%%%%%%%%%%%%%%%%%%%%%%%%%%%%%%%%%%%%%%
	Iterating this relation and using the expansion $\log(1+\varepsilon)\simeq \varepsilon - \frac12 \varepsilon^2$, we get
	\begin{eqnarray*}
		\log R_n^2
		&=&
		\log \prod^{n-1}_{j=1}
		\Big{\{} 1+\Gamma_{\omega}(j)\Big{\}}
		\\
		&=&
		\sum^{n-1}_{j=1}
		\left\{
		\frac{p_2^2}{\sin^2(2k)} \left( \cos^2(\bar\theta_j) - \frac12\sin^2(2\bar\theta_j) \right)\vo(j)^2
		\right.
		\\
		&&\phantom{blablab}
		+ \frac{p_1^2}{\sin^2(2k)} \left( \cos^2(\bar\theta_j-k) - \frac12\sin^2(2(\bar\theta_j-k)) \right)\vt(j+1)^2
		\\
		&&\phantom{blablab}
		-\frac{p_2}{\sin(2k)} \sin(2\bar\theta_j)\vo(j)
		+\frac{p_1}{\sin(2k)} \sin(2(\bar\theta_j-k))\vt(j+1)
		\\
		&&\phantom{blablab}
		+\frac{1}{\sin k}\cos(\bar \theta_n)\sin(\bar\theta_j-k)\vo(j)\vt(j+1)
		\\
		&&\phantom{blablab}
		\left.
		-\frac{2p_1p_2}{\sin^2(2k)} \cos k \cos(\bar\theta_j)\cos(\bar\theta_j-k) \vo(j)\vt(j+1)
		+ \K_{\omega}(j)
		\right\},
	\end{eqnarray*}
	where $\K_{\omega}(j)$ collects all the monomials of order higher than $3$ in the disorder variables. 
Now we use the identity $\cos^2(\beta)-\frac12 \sin^2(2\beta) = \frac14 + \frac12 \cos(2\theta)-\frac14 \cos(4\theta)$ to rewrite $\log R_n^2 $ as	\begin{eqnarray}
\notag
	\log R_n^2
	&=&
	\frac{p_2^2}{4\sin^2(2k)}\sum^{n-1}_{j=1}\esp[\vo(j)^2]+\frac{p_1^2}{4\sin^2(2k)}\sum^{n-1}_{j=1}\esp[\vt(j+1)^2]
	\\
	\label{eq:martingale-decomposition}
	&&
	+
	\sum^6_{k=1} M_{n,k}
	+Q_{n,1}+Q_{n,2}
	+
	\sum^n_{j=1} \K_{\omega}(j),
	\end{eqnarray}
	where
	\begin{eqnarray*}
		M_{n,1}
		&=&
		\frac{p_2^2}{\sin^2(2k)}\sum^{n-1}_{j=1}
		\left( \frac14 + \frac12\cos(2\bar\theta_j)-\frac14 \cos(4\bar\theta_j)\right)
		\left( \vo(j)^2-\esp[\vo(j)^2]\right),
		\\
		M_{n,2}
		&=&
		\frac{p_1^2}{\sin^2(2k)}\sum^{n-1}_{j=1}
		\left( \frac14 + \frac12\cos(2(\bar\theta_j-k))-\frac14 \cos(4(\bar\theta_j-k))\right)
		\\
		&&
		\phantom{blablablablablablablabla}
		\times
		\left( \vt(j+1)^2-\esp[\vt(j+1)^2]\right),
		\\
		M_{n,3}
		&=&
		-\frac{p_2}{\sin(2k)} \sum^{n-1}_{j=1}\sin(2\bar\theta_j)\vo(j),
		\\
%	\end{eqnarray}
%	\begin{eqnarray}
		M_{n,4}
		&=&
		\frac{p_1}{\sin(2k)} \sum^{n-1}_{j=1}\sin(2(\bar\theta_j-k))\vt(j+1),
		\\
		M_{n,5}
		&=&
		\frac{1}{\sin k}\sum^{n-1}_{j=1}\cos(\bar \theta_n)\sin(\bar\theta_j-k)\vo(j)\vt(j+1),
		\\
		M_{n,6}
		&=&
		-\frac{2p_1p_2}{\sin^2(2k)} \cos k \sum^{n-1}_{j=1}\cos(\bar\theta_j)\cos(\bar\theta_j-k) \vo(j)\vt(j+1),
	\end{eqnarray*}
	are $\mathcal{F}_n$-martingales according to Lemma \ref{thm:martingales} from Appendix \ref{sec:martingales}, and
	\begin{eqnarray}
		Q_{n,1}
		&=&
		\frac{p_2^2}{\sin^2(2k)}\sum^{n-1}_{j=1}
		\left(  \frac12\cos(2\bar\theta_j)-\frac14 \cos(4\bar\theta_j)\right)
		\esp[\vo(j)^2],
		\\
		\label{eq:sum-of-phases}
		Q_{n,2}
		&=&
		\frac{p_1^2}{\sin^2(2k)}\sum^{n-1}_{j=1}
		\left(\frac12\cos(2(\bar\theta_j-k))-\frac14 \cos(4(\bar\theta_j-k))\right)
		\esp[\vt(j+1)^2].
	\end{eqnarray}
	We apply Lemma \ref{thm:martingales} from Appendix \ref{sec:martingales} to control the martingale terms $M_{n,i}$ with $\gamma=2\alpha$ for $i=1,2,5$ and $6$, and $\gamma=\alpha$ for $i=3$ and $4$, showing that
	\begin{eqnarray*}
		|M_{n,i}| = o\left( \sum^n_{j=1} j^{-2\alpha}\right),\qquad i=1,\dots,6.
	\end{eqnarray*}
	The control of $Q_{n,1}$ and $Q_{n,2}$ is rather lengthy but quite elementary and requires to take $k$ different from $-\frac{5\pi}{8},-\frac{3\pi}{4}$ and $-\frac{7\pi}{8}$. We defer it to Lemma \ref{thm:control-phases} in Appendix \ref{sec:control-phases}.
	Finally, to estimate the error term $\K_{\omega}(j)$, we use the bound $|\log(1+\varepsilon)-\varepsilon+\frac12 \varepsilon| \leq C_1 \varepsilon^3$ for some $C_1$ which, together with \textbf{(A4)}, yields
	\begin{eqnarray*}
		|\K_{\omega}(j)| \leq C_1 |\Gamma_{\omega}(j)|^3 \leq C_2 n^{-(2\alpha+3\varepsilon)},
	\end{eqnarray*}
	for some $C_2=C_2(E,\omega)\in(0,\infty)$, $\p$-almost surely.
	This shows that
	\begin{eqnarray*}
		\sum^n_{j=1} |\K_{\omega}(j)|
		=
		o\left( \sum^n_{j=1} j^{-2\alpha}\right).
	\end{eqnarray*}
\end{proof}
%%%%%%%%%%%%%%%%%%%%%%%%%%%%%%%%%%%%%%%%%%%%%%%%%%%%%%%%%%%%
%%%%%%%%%%%%%%%%%%%%%%%%%%%%%%%%%%%%%%%%%%%%%%%%%%%%%%%%%%%%
%\begin{remark}
%	Checking the martingale property for sums of type $\sum_j F_j X_j$ with $F_j$ bounded and functions of $X_1,\cdots,X_{j-1}$, $X_j$ centered. Above, say whenever $X_j = V_j$ or a product. 
%\end{remark}
%%%%%%%%%%%%%%%%%%%%%%%%%%%%%%%%%%%%%%%%%%%%%%%%%%%%%%%%%%%%
%%%%%%%%%%%%%%%%%%%%%%%%%%%%%%%%%%%%%%%%%%%%%%%%%%%%%%%%%%%%
%%%%%%%%%%%%%%%%%%%%%%%%%%%%%%%%%%%%%%%%%%%%%%%%%%%%%%%%%%%%
%%%%%%%%%%%%%%%%%%%%%%%%%%%%%%%%%%%%%%%%%%%%%%%%%%%%%%%%%%%%
%%%%%%%%%%%%%%%%%%%%%%%%%%%%%%%%%%%%%%%%%%%%%%%%%%%%%%%%%%%%
%%%%%%%%%%%%%%%%%%%%%%%%%%%%%%%%%%%%%%%%%%%%%%%%%%%%%%%%%%%%
%%%%%%%%%%%%%%%%%%%%%%%%%%%%%%%%%%%%%%%%%%%%%%%%%%%%%%%%%%%%
%%%%%%%%%%%%%%%%%%%%%%%%%%%%%%%%%%%%%%%%%%%%%%%%%%%%%%%%%%%%
%%%%%%%%%%%%%%%%%%%%%%%%%%%%%%%%%%%%%%%%%%%%%%%%%%%%%%%%%%%%
%%%%%%%%%%%%%%%%%%%%%%%%%%%%%%%%%%%%%%%%%%%%%%%%%%%%%%%%%%%%
%%%%%%%%%%%%%%%%%%%%%%%%%%%%%%%%%%%%%%%%%%%%%%%%%%%%%%%%%%%%
%%%%%%%%%%%%%%%%%%%%%%%%%%%%%%%%%%%%%%%%%%%%%%%%%%%%%%%%%%%%

\subsection{Averaged asymptotics}\label{sec:asymptotics-TM-l1}

%%%%%%%%%%%%%%%%%%%%%%%%%%%%%%%%%%%%%%%%%%%%%%%%%%%%%%%%%%%%
%%%%%%%%%%%%%%%%%%%%%%%%%%%%%%%%%%%%%%%%%%%%%%%%%%%%%%%%%%%%
%%%%%%%%%%%%%%%%%%%%%%%%%%%%%%%%%%%%%%%%%%%%%%%%%%%%%%%%%%%%
%%%%%%%%%%%%%%%%%%%%%%%%%%%%%%%%%%%%%%%%%%%%%%%%%%%%%%%%%%%%
We now present the basic estimates that will be used in our proof of dynamical localization.
Let ${\bf T}_{\omega,[u,n]}(E)=T_{\omega,n-1}(E)\cdots T_{\omega,u}(E)$ be the truncated transfer matrix.
%%%%%%%%%%%%%%%%%%%%%%%%%%%%%%%%%%%%%%%%%%%%%%%%%%%%%%%%%%%%
%%%%%%%%%%%%%%%%%%%%%%%%%%%%%%%%%%%%%%%%%%%%%%%%%%%%%%%%%%%%
\begin{corollary}\label{thm:integrated-lyapunov}
	Let $0<\alpha\leq \tfrac12$.
	Assume \textbf{(A1)}-\textbf{(A3a)} and \textbf{(A5)}.
	 Let $I\subset \mathring\Sigma$ be a compact interval. Then one has
	\begin{eqnarray}\label{log T}
		\lim_{\substack{u,n\to\infty \\ u<n}}
		\frac{\esp\left[ \log \| {\bf T}_{\omega,[u,n]}(E) \varphi_0\|\right]}{\sum^{n-1}_{j=l}j^{	-2\alpha}}
		=
		\frac{\lambda^2(p_1^2+p_2^2)}{8\sin^2 k},
	\end{eqnarray}
	uniformly over $\|\varphi_0\|=1$ and $E\in I$ corresponding to values of $k$ different from $-\frac{5\pi}{8},-\frac{3\pi}{4}$ and $-\frac{7\pi}{8}$.
\end{corollary}
%%%%%%%%%%%%%%%%%%%%%%%%%%%%%%%%%%%%%%%%%%%%%%%%%%%%%%%%%%%%
%%%%%%%%%%%%%%%%%%%%%%%%%%%%%%%%%%%%%%%%%%%%%%%%%%%%%%%%%%%%
\begin{proof}
	Observing that all the martingale terms in \eqref{eq:martingale-decomposition} have $0$ expected value, we get
	\begin{eqnarray*}
		\esp\left(\log \frac{R_n^2}{R_u^2}\right)
		&=&
		\frac{p_2^2}{4\sin^2(2k)}\sum^{n-1}_{j=u}\esp[\vo(j)^2]	
		+
		\frac{p_1^2}{4\sin^2(2k)}\sum^{n-1}_{j=u}\esp[\vt(j+1)^2]
		\\
		&&
		+
		\esp\left[Q_{u,n,1}+Q_{u,n,2}\right]u
		+
		\sum^n_{j=u} \esp\left[ \K_{\omega}(j)\right],
	\end{eqnarray*}
	where $Q_{u,n,1}$ and $Q_{u,n,2}$ are defined as $Q_{n,1}$ and $Q_{n,2}$ with the sum starting from $u$ instead of $1$. The error terms $\K_{\omega}(j)$ can be controlled as in the proof of Proposition \ref{thm:lyapunov-exponents}: there exists $C_1>0$ such that
	\begin{eqnarray*}
		\esp\left[ |\K_{\omega}(j)|\right]
		\leq
		C_1  \esp\left[ |\Gamma_{\omega}(j)|^3\right]
		\leq
		C_2 n^{-(2\alpha+3\epsilon)},
	\end{eqnarray*}
	for some $C_2=C_2(I)\in(0,\infty)$
	in virtue of \textbf{(A5)}.
%	\begin{eqnarray}
%		\esp[|E_{\omega}(j)|] 
%		\leq
%		C_1
%		\sum_{l\geq 3} \frac{1}{l} \esp[|V_j|^p]
%		\leq
%		C_1
%		n^{-(2\alpha+\varepsilon)}
%		\sum_{l\geq 3} \frac{1}{l}C_2^l n^{-(\frac{2\alpha}{3}+\varepsilon)(l-3)},
%	\end{eqnarray}
%	for all $E\in I$
%	where $C_1=C_1(I)$ and $C_2$ is the constant in condition \textbf{(A4')}. The sum above is finite for $j\geq j_0$. To estimate the terms corresponding to $j<j_0$, use again
%	\begin{eqnarray}
%		\esp\left[ |E_{\omega}(j)|\right]
%		\leq
%		C  \esp\left[ \min\{|\Gamma_{\omega}(j)|^2,|\Gamma_{\omega}(j)|^3\}\right],
%	\end{eqnarray}
%	which is uniformly bounded in $E\in I$. 
	
	The control of the terms $Q_{u,n,1}$ and $Q_{u,n,2}$ is deferred to Lemma \ref{thm:control-phases-integrated} in Appendix \ref{sec:control-phases} and is uniform as well. The corollary then follows from Lemma \ref{thm:comparison}.
\end{proof}
%An inspection at the proof of Theorem \ref{thm:lyapunov-exponents},
%together with the comparison bounds in Lemma \ref{thm:comparison} and the fact that all the martingale terms have $0$ expected value, shows that, for each compact energy interval $I\subset\Sigma^{\circ}$,
%\begin{eqnarray}\label{eq:prufer-bound-mn}
%	\esp\left[ \log \| {\bf T}_{\omega,[m,n]}(E) {\bf x}_0\|\right]
%	=
%	\frac{\lambda^2(p_1^2+p_2^2)}{8\sin^2 k}\sum^{n-1}_{j=m}\frac{1}{j^{2\alpha}} + o\left( \sum^{n-1}_{j=m} j^{-2\alpha}\right),
%\end{eqnarray}
%uniformly in $\| {\bf x}_0\|=1$ and $E\in I$ 
%%%%%%%%%%%%%%%%%%%%%%%%%%%%%%%%%%%%%%%%%%%%%%%%%%%%%%%%%%%%
%%%%%%%%%%%%%%%%%%%%%%%%%%%%%%%%%%%%%%%%%%%%%%%%%%%%%%%%%%%%
We collect two non-asymptotic bounds in the lemma below.
%%%%%%%%%%%%%%%%%%%%%%%%%%%%%%%%%%%%%%%%%%%%%%%%%%%%%%%%%%%%
%%%%%%%%%%%%%%%%%%%%%%%%%%%%%%%%%%%%%%%%%%%%%%%%%%%%%%%%%%%%
\begin{lemma}\label{thm:bounds-on-Tmn}
	Let $0<\alpha\leq \tfrac12$.
	Assume \textbf{(A1)}-\textbf{(A3a)} and \textbf{(A5)}.
	Let $I\subset\mathring\Sigma$ be a compact interval. Then for all $\beta'$ such that $0<\beta'<\displaystyle\inf_{E\in I}\beta(\lambda,E)$, there exists $n_0=n_0(I)$ so that one has
%	\begin{eqnarray}
%		\label{eq:lower-bound-Tmn}
%		\esp\left[ \log \|T_{\omega,ln_0 }(E)\cdots T_{\omega,(l-1)n_0+1}(E){\bf x}_0\|\right]
%		&\geq&
%		%\beta' \sum^{kn_0}_{j=(k-1)n_0+1} \frac{1}{j^{2\alpha}}
%%		%+
%%		%o\left(\frac{n_0^{1-2\alpha}}{k^{2\alpha}}\right).
%		%\geq 
%		\beta'
%		  \frac{n_0^{1-2\alpha}}{(1-2\alpha)l^{2\alpha}},
%	\end{eqnarray}
\begin{eqnarray}
		\label{eq:lower-bound-Tmn}
		\esp\left[ \log \|T_{\omega,ln_0 }(E)\cdots T_{\omega,(l-1)n_0+1}(E)\varphi_0\|\right]
		&\geq&
		\beta'
		  \sum^{ln_0}_{j=(l-1)n_0+1} \frac{1}{j^{2\alpha}},
	\end{eqnarray}
	for all $l\geq 1$, $\|\varphi_0\|=1$ and $E\in I$.
	\newline
	Furthermore, there exists a constant $C=C(I)$ such that
	\begin{eqnarray}
		\label{eq:upper-bound-Tmn}
		\esp\left[\left( \log \|T_{\omega,ln_0 }(E)\cdots T_{\omega,(l-1)n_0+1}(E)\varphi_0\| \right)^4\right]
		&\leq&
		C
		\sum^{ln_0}_{j=(l-1)n_0+1} \frac{1}{j^{2\alpha}},
	\end{eqnarray}
	for all $l\geq 1$, $\|\varphi_0\|=1$ and $E\in I$.
\end{lemma}
%%%%%%%%%%%%%%%%%%%%%%%%%%%%%%%%%%%%%%%%%%%%%%%%%%%%%%%%%%%%
%%%%%%%%%%%%%%%%%%%%%%%%%%%%%%%%%%%%%%%%%%%%%%%%%%%%%%%%%%%%
\begin{proof}
	From Corollary \ref{thm:integrated-lyapunov}, we can find $n_0$ large enough such that
	\begin{eqnarray*}
		\esp\left[ \log \|T_{\omega,ln_0 }(E)\cdots T_{\omega,(l-1)n_0+1}(E)\varphi_0\|\right]
		&\geq&
		\beta' \sum^{ln_0}_{j=(l-1)n_0+1} \frac{1}{j^{2\alpha}},
	\end{eqnarray*}
	for all $l\geq 1$, $\|\varphi_0\|=1$ and all $E\in I$ corresponding to values of $k$ different from $-\frac{5\pi}{8},-\frac{3\pi}{4}$ and $-\frac{7\pi}{8}$. The bound for all energies in $I$ then follows by continuity of the left-hand-side above with respect to $E$. This proves \eqref{eq:lower-bound-Tmn}. The estimate \eqref{eq:upper-bound-Tmn} is a crude $\mathbb{L}^2$ bound and follows by an inspection of the decomposition \eqref{eq:prufer-transform-s1-radii}.
\end{proof}
%%%%%%%%%%%%%%%%%%%%%%%%%%%%%%%%%%%%%%%%%%%%%%%%%%%%%%%%%%%%
%%%%%%%%%%%%%%%%%%%%%%%%%%%%%%%%%%%%%%%%%%%%%%%%%%%%%%%%%%%%

%%%%%%%%%%%%%%%%%%%%%%%%%%%%%%%%%%%%%%%%%%%%%%%%%%%%%%%%%%%%
%%%%%%%%%%%%%%%%%%%%%%%%%%%%%%%%%%%%%%%%%%%%%%%%%%%%%%%%%%%%
%%%%%%%%%%%%%%%%%%%%%%%%%%%%%%%%%%%%%%%%%%%%%%%%%%%%%%%%%%%%
%%%%%%%%%%%%%%%%%%%%%%%%%%%%%%%%%%%%%%%%%%%%%%%%%%%%%%%%%%%%
%%%%%%%%%%%%%%%%%%%%%%%%%%%%%%%%%%%%%%%%%%%%%%%%%%%%%%%%%%%%
%%%%%%%%%%%%%%%%%%%%%%%%%%%%%%%%%%%%%%%%%%%%%%%%%%%%%%%%%%%%
%%%%%%%%%%%%%%%%%%%%%%%%%%%%%%%%%%%%%%%%%%%%%%%%%%%%%%%%%%%%
%%%%%%%%%%%%%%%%%%%%%%%%%%%%%%%%%%%%%%%%%%%%%%%%%%%%%%%%%%%%
%%%%%%%%%%%%%%%%%%%%%%%%%%%%%%%%%%%%%%%%%%%%%%%%%%%%%%%%%%%%
%%%%%%%%%%%%%%%%%%%%%%%%%%%%%%%%%%%%%%%%%%%%%%%%%%%%%%%%%%%%
%%%%%%%%%%%%%%%%%%%%%%%%%%%%%%%%%%%%%%%%%%%%%%%%%%%%%%%%%%%%
%%%%%%%%%%%%%%%%%%%%%%%%%%%%%%%%%%%%%%%%%%%%%%%%%%%%%%%%%%%%

\section{Super-critical regime: a.c. spectrum}\label{sec:ac}

%%%%%%%%%%%%%%%%%%%%%%%%%%%%%%%%%%%%%%%%%%%%%%%%%%%%%%%%%%%%
%%%%%%%%%%%%%%%%%%%%%%%%%%%%%%%%%%%%%%%%%%%%%%%%%%%%%%%%%%%%
This is based on a criterion of Last and Simon \cite{LS} that relates spectral properties to transfer matrices behavior. Let $\mathcal{T}_n(E)$ denote the product of transfer matrices associated to a bounded Schr\"odinger operator $H$ on $\ell^2(\n^*)$ and consider an energy $E$.
%%%%%%%%%%%%%%%%%%%%%%%%%%%%%%%%%%%%%%%%%%%%%%%%%%%%%%%%%%%%
%%%%%%%%%%%%%%%%%%%%%%%%%%%%%%%%%%%%%%%%%%%%%%%%%%%%%%%%%%%%
\begin{theorem}\cite[Teorem 1.3]{LS}
	Suppose that
	\begin{eqnarray}\label{eq:LS-criterion}
		\liminf_n \int^b_a \| \mathcal{T}_n(E)\|^4dE <\infty.
	\end{eqnarray}
	Then, $(a,b)\subset \sigma(H)$ and the spectral measure is purely absolutely continuous on $(a,b)$.
\end{theorem}
 %%%%%%%%%%%%%%%%%%%%%%%%%%%%%%%%%%%%%%%%%%%%%%%%%%%%%%%%%%%%
%%%%%%%%%%%%%%%%%%%%%%%%%%%%%%%%%%%%%%%%%%%%%%%%%%%%%%%%%%%%
 The criterion is valid for any power larger than 2. There is nothing special about the power 4 except that it makes the computations easier. %According to our observation on the equivalence of the $(X_n)_n$ and $(Y_n)_n$ systems, it is enough to check to criterion for $\mathcal{M}_n(E)$ instead of  $\mathcal{T}_n(E)$.
 
 In the following, we write $\mathbf{T}_{\omega,n}(E)$, $\Psi_n(E)$ and $R_n(E)$ when we want to emphasise the dependence on the energy $E$. 
% We write $V_{\omega,n}=\lambda a_n\omega_n$ for $n\geq 0$.
 %%%%%%%%%%%%%%%%%%%%%%%%%%%%%%%%%%%%%%%%%%%%%%%%%%%%%%%%%%%%
%%%%%%%%%%%%%%%%%%%%%%%%%%%%%%%%%%%%%%%%%%%%%%%%%%%%%%%%%%%%
\begin{proof}[Proof of Theorem \ref{thm:spectrum}, Part 1 (super-critical case)]
Let $\theta_0$ be any initial angle and let $[a,b]\subset \mathring\Sigma$. 
According to Lemma \ref{thm:comparison}, it is enough to show
\begin{eqnarray}\label{eq:bound-int-rho}
	\liminf_n \esp\left[  \int^b_a R_n^4(E)dE\right]<\infty,
\end{eqnarray}
since Fatou's lemma yields
\begin{eqnarray*}
	\esp\left[ \liminf_n \int^b_a R_n^4(E)dE \right] \leq \liminf_n \esp\left[  \int^b_a R_n^4(E) dE\right]<\infty,
\end{eqnarray*}
which implies that \eqref{eq:LS-criterion} holds almost surely. 
Squaring \eqref{eq:prufer-transform-s1-radii}, we obtain
\begin{eqnarray*}
	R_{n+1}^4 (E)
	&=& 
	\Big{\{} 
		1 
		- \frac{2p_2}{\sin(2k)}\sin(2\bar{\theta}_n)\vo(n)
		\\
		&&
		\phantom{blablablabla}
		+\frac{2p_1}{\sin(2k)}\sin(2(\bar{\theta}_n-k))\vt(n+1)
		+ A_{\omega,n}(E) 
	\Big{\}} 
	R_n^4(E),
\end{eqnarray*}
where $A_{\omega,n}(E)$ collects all the terms of degree 1, 2 and 3 in the disorder variables. An inspection at those terms shows that there exists $c=c(a,b)\in(0,\infty)$ such that 
\begin{eqnarray*}
	\esp\left[|A_{\omega,n}(E)|\, | \mathcal{F}_{n-1}\right]
	\leq
	c
	\left(
		\esp[|V_{\omega,i}|^2]
		+
		\esp[|V_{\omega,i}|^3]
		+
		\esp[|V_{\omega,i}|^4]
	\right).
\end{eqnarray*}
Using \textbf{(A3a)} and \textbf{(A3b)}, and the simple inequality
\begin{eqnarray*}
	\esp[|V_{\omega,i}|^3]
	\leq
	\esp[|V_{\omega,i}|^2]^{\frac12}\ 
	\esp[|V_{\omega,i}|^4]^{\frac12},
\end{eqnarray*}
we conclude that 
\begin{eqnarray*}
	\esp[|A_{\omega,n}(E)|\, | \mathcal{F}_{n-1}] \leq c' n^{-2\alpha},
\end{eqnarray*}
for some $c'=c'(a,b)\in(0,\infty)$.

Recalling  that $\bar \theta_n$ is $\mathcal{F}_{n-1}$-measurable and $\vo(n)$ is independent of $\mathcal{F}_{n-1}$, centered and integrable, we get
\begin{eqnarray*}
	\esp\left[ \sin(2\bar{\theta}_n)\vo(n) \Big{|} \mathcal{F}_{n-1} \right]
	=
	\sin(2\bar{\theta}_n)
	\esp\left[ \vo(n) \right]
	=
	0.
\end{eqnarray*}
The same argument gives
\begin{eqnarray*}
	\esp\left[ \sin(2(\bar{\theta}_n-k))\vt(n+1) \Big{|} \mathcal{F}_{n-1} \right]
	=
	0.
\end{eqnarray*}
Hence, as $R_n$ is $\mathcal{F}_{n-1}$-measurable and $R_n^4$ is integrable,  we conclude that
\begin{eqnarray*}
	\esp\left[ R_{n+1}^4 (E) \Big{|} \mathcal{F}_{n-1} \right]
	=
	\esp\left[ 1 + A_{\omega,n}(E) \Big{|} \mathcal{F}_{n-1} \right]\, R_{n}^4 (E) 
	\leq
	\left(
		1 + c'n^{-2\alpha}
	\right)
	\, R_{n}^4 (E), 
\end{eqnarray*}
where we used the uniform bound on $A_{\omega,n}(E)$.  Integrating with respect to $\p$ and iterating, we obtain
\begin{eqnarray*}
	\esp\left[ R^4_{n+1}(E) \right] 
	\leq 
	\left(
		1 + \frac{c'}{n^{2\alpha} }
	\right)
	\esp\left[ R^4_{n}(E) \right] 
	\leq 
	\prod^n_{j=1}
	\left(
		1 + \frac{c'}{j^{2\alpha} }
	\right), 
%	 \cdot 
%	 \esp\left[ R^4_1(E) \right],
\end{eqnarray*}
for all $E\in [a,b]$ and all $n\geq 1$. Since $\displaystyle \sum_j j^{-2\alpha} < \infty$ for $\alpha>\tfrac12$, the product above is bounded uniformly in $n$ and $E\in[a,b]$.  This finishes the proof.
%uniformly in $n$ and $E$ (higher moments being controlled by the second moment for bounded random variables). 
%Observing that $\bar \theta_n$ and $R_n$ are $\mathcal{F}_{n-1}$-measurable and $V_{\omega,n}$ is independent of $\mathcal{F}_{n-1}$ and centered, then taking the expectation with respect to the probability measure $\p$, we obtain
%\begin{eqnarray}
%	\esp\left[ R^4_{n+1}(E) \right] \leq \{ 1 + b_n \} \esp\left[ R^4_{n}(E) \right] \leq \prod^n_{j=1}\{ 1 + b_j \} \cdot \esp\left[ R^4_1(E) \right],
%\end{eqnarray}
%where $b_j \leq C j^{-2\alpha}$ uniformly in $j$ and $E\in[a,b]$. As $\displaystyle \sum_n b_n < \infty$, the product above is bounded uniformly in $n$ and $E$. 
\end{proof}
%%%%%%%%%%%%%%%%%%%%%%%%%%%%%%%%%%%%%%%%%%%%%%%%%%%%%%%%%%%%
%%%%%%%%%%%%%%%%%%%%%%%%%%%%%%%%%%%%%%%%%%%%%%%%%%%%%%%%%%%%

%%%%%%%%%%%%%%%%%%%%%%%%%%%%%%%%%%%%%%%%%%%%%%%%%%%%%%%%%%%%
%%%%%%%%%%%%%%%%%%%%%%%%%%%%%%%%%%%%%%%%%%%%%%%%%%%%%%%%%%%%
%%%%%%%%%%%%%%%%%%%%%%%%%%%%%%%%%%%%%%%%%%%%%%%%%%%%%%%%%%%%
%%%%%%%%%%%%%%%%%%%%%%%%%%%%%%%%%%%%%%%%%%%%%%%%%%%%%%%%%%%%
%%%%%%%%%%%%%%%%%%%%%%%%%%%%%%%%%%%%%%%%%%%%%%%%%%%%%%%%%%%%
%%%%%%%%%%%%%%%%%%%%%%%%%%%%%%%%%%%%%%%%%%%%%%%%%%%%%%%%%%%%
%%%%%%%%%%%%%%%%%%%%%%%%%%%%%%%%%%%%%%%%%%%%%%%%%%%%%%%%%%%%
%%%%%%%%%%%%%%%%%%%%%%%%%%%%%%%%%%%%%%%%%%%%%%%%%%%%%%%%%%%%

\section{Critical regime: spectral transition and transport}\label{sec:critical}

%%%%%%%%%%%%%%%%%%%%%%%%%%%%%%%%%%%%%%%%%%%%%%%%%%%%%%%%%%%%
%%%%%%%%%%%%%%%%%%%%%%%%%%%%%%%%%%%%%%%%%%%%%%%%%%%%%%%%%%%%

%%%%%%%%%%%%%%%%%%%%%%%%%%%%%%%%%%%%%%%%%%%%%%%%%%%%%%%%%%%%
%%%%%%%%%%%%%%%%%%%%%%%%%%%%%%%%%%%%%%%%%%%%%%%%%%%%%%%%%%%%
%%%%%%%%%%%%%%%%%%%%%%%%%%%%%%%%%%%%%%%%%%%%%%%%%%%%%%%%%%%%
%%%%%%%%%%%%%%%%%%%%%%%%%%%%%%%%%%%%%%%%%%%%%%%%%%%%%%%%%%%%

\subsection{Spectral transition}

%%%%%%%%%%%%%%%%%%%%%%%%%%%%%%%%%%%%%%%%%%%%%%%%%%%%%%%%%%%%
%%%%%%%%%%%%%%%%%%%%%%%%%%%%%%%%%%%%%%%%%%%%%%%%%%%%%%%%%%%%
The absence of absolutely continuous spectrum is a consequence of the following criterion of Last and Simon \cite{LS}. With the notations of the beginning of Section \ref{sec:ac}:
%%%%%%%%%%%%%%%%%%%%%%%%%%%%%%%%%%%%%%%%%%%%%%%%%%%%%%%%%%%%
%%%%%%%%%%%%%%%%%%%%%%%%%%%%%%%%%%%%%%%%%%%%%%%%%%%%%%%%%%%%
\begin{theorem}\cite[Theorem 1.2]{LS}
	Suppose $\displaystyle\lim_{n\to\infty} \| \mathcal{T}_n(E)\| = \infty$ for a.e. $E\in[a,b]$. Then, $\mu_{ac}([a,b])=0$, where $\mu_{ac}$ is the absolutely continuous spectral measure associated to $H$.
\end{theorem}
%%%%%%%%%%%%%%%%%%%%%%%%%%%%%%%%%%%%%%%%%%%%%%%%%%%%%%%%%%%%
%%%%%%%%%%%%%%%%%%%%%%%%%%%%%%%%%%%%%%%%%%%%%%%%%%%%%%%%%%%%
It follows from Theorem \ref{thm:lyapunov-exponents} that for almost every $|E| \in (m, \sqrt{m^2}+4)$, $\displaystyle\lim_{n\to\infty} \| \prodtransn(E)\| = \infty$ for $\p$-almost every $\omega$. By Fubini's theorem, we conclude that, $\p$-almost surely,  $\displaystyle\lim_{n\to\infty} \| \prodtransn(E)\| = \infty$ for almost every $|E| \in (m, \sqrt{m^2}+4)$ and we can apply the above theorem.
%%%%%%%%%%%%%%%%%%%%%%%%%%%%%%%%%%%%%%%%%%%%%%%%%%%%%%%%%%%%
%%%%%%%%%%%%%%%%%%%%%%%%%%%%%%%%%%%%%%%%%%%%%%%%%%%%%%%%%%%%
Now, to determine the nature of the spectrum, it will be enough to determine whether the generalized eigenfunctions are $\ell^2$ or not. For an angle $\vartheta$, we denote 
$\displaystyle \widehat{\vartheta}
=
\begin{pmatrix}
	\cos \vartheta \\ \sin \vartheta
\end{pmatrix}
$.
%%%%%%%%%%%%%%%%%%%%%%%%%%%%%%%%%%%%%%%%%%%%%%%%%%%%%%%%%%%%
%%%%%%%%%%%%%%%%%%%%%%%%%%%%%%%%%%%%%%%%%%%%%%%%%%%%%%%%%%%%
\begin{proposition}\label{thm:decay-eigenfunctions}
	Let $\alpha=\frac12$, assume \textbf{(A1)}-\textbf{(A3a)} and \textbf{(A4)}, and let $k\neq -\frac{\pi}{2}, -\frac{3\pi}{4}, -\pi$. Then, for $\p$-almost every $\omega$, there exists an initial angle $\vartheta_0 = \vartheta_0(\omega)$ such that
	\begin{eqnarray*}
		\lim_{n\to\infty} \frac{\log \| \mathbf{T}_{\omega,n}(E) \widehat\vartheta_0\|}{\log n} = -\beta(\lambda,E),\quad \p-a.s.,
	\end{eqnarray*}
	for all $\lambda > 0$.
\end{proposition}
%%%%%%%%%%%%%%%%%%%%%%%%%%%%%%%%%%%%%%%%%%%%%%%%%%%%%%%%%%%%
%%%%%%%%%%%%%%%%%%%%%%%%%%%%%%%%%%%%%%%%%%%%%%%%%%%%%%%%%%%%
The proof of this proposition is given in details in Appendix \ref{app:unimodular}. We are now ready to prove Part $(ii)$ in Theorem \ref{thm:spectrum}.
%%%%%%%%%%%%%%%%%%%%%%%%%%%%%%%%%%%%%%%%%%%%%%%%%%%%%%%%%%%%
%%%%%%%%%%%%%%%%%%%%%%%%%%%%%%%%%%%%%%%%%%%%%%%%%%%%%%%%%%%%
\begin{proof}[Proof of Theorem \ref{thm:spectrum}, Part $(ii)$]
	We have just established that there is no absolutely continuous spectrum.
	From Proposition \ref{thm:decay-eigenfunctions}, we see that the generalized eigenfunction corresponding to $E$ and $\lambda$ are $\ell^2$ if and only if $\beta>\frac12$ which can be seen to be equivalent to
	\begin{eqnarray}\label{eq:l2-criterion}
		\lambda^2 > \frac12 \frac{(E^2-m^2)(m^2+4-E^2)}{m^2+E^2} =: F(E).
	\end{eqnarray}
	This function $F$ satisfies $F(m)=F(\sqrt{m^2+4})=0$ and reaches its maximum $\lambda^*(m)$ at a unique point $E^*(m)\in (m,\sqrt{m^2+4})$. In particular, if $\lambda>\lambda^*(m)$ then the $\ell^2$-condition \eqref{eq:l2-criterion} is always fulfilled and the corresponding generalized eigenvalue is a bona fide eigenvalue. If $0<\lambda\leq \lambda^*(m)$, there exist two values $m < E^*_-(\lambda,m) < E^*_+(\lambda,m) < \sqrt{m^2+4}$ such that the criterion \eqref{eq:l2-criterion} is met. Note that $E^*_{\pm}(\lambda,m)$ are the two roots of the equation $\lambda^2 = F(E)$.	
	The result then follows from the theory of rank one perturbations \cite{Si95}.
	In all the cases above, the spectrum is pure point. Otherwise, it is a fortiori singular continuous.
\end{proof}
%%%%%%%%%%%%%%%%%%%%%%%%%%%%%%%%%%%%%%%%%%%%%%%%%%%%%%%%%%%%
%%%%%%%%%%%%%%%%%%%%%%%%%%%%%%%%%%%%%%%%%%%%%%%%%%%%%%%%%%%%
%%%%%%%%%%%%%%%%%%%%%%%%%%%%%%%%%%%%%%%%%%%%%%%%%%%%%%%%%%%%
%%%%%%%%%%%%%%%%%%%%%%%%%%%%%%%%%%%%%%%%%%%%%%%%%%%%%%%%%%%%

\subsection{Lower bounds on eigenfunctions and transport}

%%%%%%%%%%%%%%%%%%%%%%%%%%%%%%%%%%%%%%%%%%%%%%%%%%%%%%%%%%%%
%%%%%%%%%%%%%%%%%%%%%%%%%%%%%%%%%%%%%%%%%%%%%%%%%%%%%%%%%%%%

The next lemma provides a lower bound on any non-trivial solution of $ \D\Phi = E\Phi$ for $\alpha=\frac12$ and $\lambda> 0$, uniformly in $E$ ranging over compact intervals of $\sigma_{pp}(\D)$.
%%%%%%%%%%%%%%%%%%%%%%%%%%%%%%%%%%%%%%%%%%%%%%%%%%%%%%%%%%%%
%%%%%%%%%%%%%%%%%%%%%%%%%%%%%%%%%%%%%%%%%%%%%%%%%%%%%%%%%%%%
\begin{lemma}\label{thm:lower-bound-eigenfunctions}
	Let $\alpha=\frac12$ and fix $\lambda> 0$. Assume \textbf{(A1)}-\textbf{(A3a)}.
	For $E\in \sigma_{pp}(\D)$, define $\Phi_{\omega,E}=(\Phi_{\omega,E,n})_n\in (\C^2)^{\n^*}$ as the solution of 
	$\D\Phi = E\Phi$ with a possibly random initial condition $\widehat{\vartheta}_0$.
	%For $E\in\sigma_{pp}(H_{\omega,\lambda})$, denote by $(x_{\omega,E,n})_n$ the (unique) eigenfunction corresponding to the energy $E$.
	Then, for each compact interval $I\subset \sigma_{pp}(\D)$, there exists a deterministic constant $\kappa=\kappa(I)>0$ such that, for $\p$-almost every $\omega$, there exists $c_{\omega}=c_{\omega}(I)>0$ such that
	\begin{eqnarray}\label{lower bnd eigenfunction}
	%\label{eq:lower-bound-eigenfunctions}
		\| \Phi_{\omega,E,n}\| \geq c_{\omega} n^{-\kappa}, \quad \forall \, n\in \n^*.
	\end{eqnarray}	 
\end{lemma}
%%%%%%%%%%%%%%%%%%%%%%%%%%%%%%%%%%%%%%%%%%%%%%%%%%%%%%%%%%%%
%%%%%%%%%%%%%%%%%%%%%%%%%%%%%%%%%%%%%%%%%%%%%%%%%%%%%%%%%%%%
\begin{proof}
	%We prove the bound for all $n\geq 1$, the opposite case being similar.
	%By Proposition \ref{thm:decay-eigenfunctions}, there exists an initial angle $\vartheta_0=\vartheta_0(\omega)$ such that
	We can reconstruct $\Phi_{\omega,E}$ through the recurrence
	$\displaystyle
		\Phi_{\omega,E,n}			
	=
	{\bf T}_{\omega,n-1}(E)
	\widehat{\vartheta}_0
	$.
	This implies in particular that
	\begin{eqnarray*}
		\| \Phi_{\omega,E,n}\|
		\geq 
		\| {\bf T}_{\omega,n-1}(E)\|^{-1}.
	\end{eqnarray*}	
	Hence, using Lemma \ref{thm:comparison} with some $\vartheta_1 \neq \vartheta_2$, 
	\begin{eqnarray*}
		\p\left[ \| \Phi_{\omega,E,n}\|  \leq n^{-\kappa}\right]
		&\leq&
		\p\left[ \| {\bf T}_{\omega,n-1}(E) \| \geq n^{\kappa}\right]\\
		&\leq& 
		n^{-2\kappa} \ \esp\left[ \| {\bf T}_{\omega,n-1} (E)\|^2\right]\\
		&\leq& 
		C_1(\vartheta_1,\vartheta_2) n^{-2\kappa}
		\left(
			 \esp\left[ R_n^2(E,\vartheta_1)\right]+\esp\left[ R_n^2(E,\vartheta_2)\right]
		\right),
	\end{eqnarray*}
	for some $C_1(\vartheta_1,\vartheta_2) >0$.
	Keeping in mind the recursion \eqref{eq:prufer-transform-s1-radii}, the argument of the proof of Part 1 of Theorem \ref{thm:spectrum} given in Section \ref{sec:ac} can be reproduced and yields
	\begin{eqnarray*}
		\esp\left[ R_n^2(E,\vartheta_1)\right]
		&\leq& 
		\prod^n_{j=1} \left( 1 + \frac{b}{j} \right)
		\leq
		C_2 \ n^{b-2\kappa},
	\end{eqnarray*}
	for some constants $b=b(I)$ and $C_2=C_2(b)$. The estimate for $R_n^2(E,\vartheta_2)$ is of course similar. Taking $b-2\kappa<-1$, the result follows by Borel-Cantelli.
\end{proof}
%%%%%%%%%%%%%%%%%%%%%%%%%%%%%%%%%%%%%%%%%%%%%%%%%%%%%%%%%%%%
%%%%%%%%%%%%%%%%%%%%%%%%%%%%%%%%%%%%%%%%%%%%%%%%%%%%%%%%%%%%
\begin{proof}[Proof of Theorem \ref{thm:transport}]
	Let $c_{\omega}$ and $\kappa$ be as in Lemma \ref{thm:lower-bound-eigenfunctions}. Let $(\varphi_l)_{l}$ be a basis of ${\text Ran}P_I(\D)$ consisting of eigenfunctions of the operator $\D$, with corresponding eigenvalues $(E_l)_l$.
	 Define the truncated position operator $\X_N= {\bf X}\ \chi_{[-N,N]}$. Then, taking $p>2\kappa-1$,
	\begin{eqnarray*}
		\left\| |\X_N |^{p/2} \varphi_l\right\|^2
		&=&
		\sum_{1\leq n \leq N} |n|^p \|\varphi_l(n)\|^2\\
		&\geq&
		c_{\omega} \sum_{1\leq n \leq N} |n|^{p-2\kappa}
		\geq
		c'_{\omega}\ N^{p-2\kappa+1},
	\end{eqnarray*}		 
	for some $c'_{\omega}>0$ and for all $l$. Let $\varphi \in {\text Ran}P_I(\D)$ and write $\varphi = \sum_l a_l \varphi_l$ with $\sum_l |a_l|^2=1$. Then,
	\begin{eqnarray*}
		\left\| |\X_N|^{p/2} \e^{-i t \D}\psi\right\|^2
		=
		\sum_{l,l'} a_l \overline{a}_{l'}\ \e^{-it(E_l - E_{l'})}
		\langle \varphi_{l'}, |\X_N|^{p/2} \varphi_l \rangle.
	\end{eqnarray*}
	After a careful application of the dominated convergence theorem to exchange sums and integrals, we obtain
	\begin{eqnarray*}
		\lim_{T\to\infty} \frac{1}{T} \int^T_0 \left\| |\X_N|^{p/2} \e^{-i t \D}\psi\right\|^2 dt
		=
		\sum_l |a_l|^2 \left\| |\X|^{p/2} \varphi_l\right\|^2
		\geq 
		c'_{\omega}\ N^{p-2\kappa+1}.
	\end{eqnarray*}
	Hence, there exists an diverging (random) sequence $(T_N)_N$ such that
	\begin{eqnarray}\label{lower bnd moment}
		\frac{1}{T_N} \int^{T_N}_0 \left\| |{\bf X}_N|^{p/2} \e^{-i t \D}\varphi\right\|^2 dt
		\geq 
		\frac{c'_{\omega}}{2}\ N^{p-2\kappa+1},
	\end{eqnarray}
	for all $N\geq 1$. We can then find a diverging (random) sequence $(t_N)_N$ such that
	\begin{eqnarray*}
		\left\| |\X_N|^{p/2} \e^{-i t_N \D}\varphi\right\|^2
		\geq 
		\frac{c'_{\omega}}{4}\ N^{p-2\kappa+1},
	\end{eqnarray*}
	for all $N\geq 1$. This finishes the proof.
\end{proof}

%%%%%%%%%%%%%%%%%%%%%%%%%%%%%%%%%%%%%%%%%%%%%%%%%%%%%%%%%%%%
%%%%%%%%%%%%%%%%%%%%%%%%%%%%%%%%%%%%%%%%%%%%%%%%%%%%%%%%%%%%
%%%%%%%%%%%%%%%%%%%%%%%%%%%%%%%%%%%%%%%%%%%%%%%%%%%%%%%%%%%%
%%%%%%%%%%%%%%%%%%%%%%%%%%%%%%%%%%%%%%%%%%%%%%%%%%%%%%%%%%%%
%%%%%%%%%%%%%%%%%%%%%%%%%%%%%%%%%%%%%%%%%%%%%%%%%%%%%%%%%%%%
%%%%%%%%%%%%%%%%%%%%%%%%%%%%%%%%%%%%%%%%%%%%%%%%%%%%%%%%%%%%
%%%%%%%%%%%%%%%%%%%%%%%%%%%%%%%%%%%%%%%%%%%%%%%%%%%%%%%%%%%%
%%%%%%%%%%%%%%%%%%%%%%%%%%%%%%%%%%%%%%%%%%%%%%%%%%%%%%%%%%%%

\section{Sub-critical regime: pure point spectrum}\label{sec:DL-pp}
%%%%%%%%%%%%%%%%%%%%%%%%%%%%%%%%%%%%%%%%%%%%%%%%%%%%%%%%%%%%
%%%%%%%%%%%%%%%%%%%%%%%%%%%%%%%%%%%%%%%%%%%%%%%%%%%%%%%%%%%%

Part 3 of Theorem \ref{thm:spectrum} follows from the theory of rank one perturbations once we establish the following Proposition which is a direct consequence of Proposition \ref{thm:lyapunov-exponents} and \cite[Theorem 8.3]{LS} stated in Appendix \ref{app:unimodular} as Theorem \ref{thm:LS-osc}. 
%%%%%%%%%%%%%%%%%%%%%%%%%%%%%%%%%%%%%%%%%%%%%%%%%%%%%%%%%%%%
%%%%%%%%%%%%%%%%%%%%%%%%%%%%%%%%%%%%%%%%%%%%%%%%%%%%%%%%%%%%
\begin{proposition}\label{thm:asymptotics-eigenfunctions-pp}
	Let $0<\alpha<\frac12$. Assume \textbf{(A1)}-\textbf{(A3a)} and \textbf{(A4)}, and let $k\neq -\frac{\pi}{2}, -\frac{3\pi}{4}, -\pi$. Then, for $\p$-almost every $\omega$, there exists an initial angle $\vartheta_0 = \vartheta_0(\omega)$ such that
	\begin{eqnarray}\label{log T-beta}
		\lim_{n\to\infty} \frac{\log \| \mathbf{T}_{\omega,n}(E) \widehat\vartheta_0\|}{\sum^n_{j=1}j^{-2\alpha}} = -\beta(\lambda,E),\quad \p-a.s.,
	\end{eqnarray}
	for all $\lambda > 0$.
\end{proposition}
%%%%%%%%%%%%%%%%%%%%%%%%%%%%%%%%%%%%%%%%%%%%%%%%%%%%%%%%%%%%
%%%%%%%%%%%%%%%%%%%%%%%%%%%%%%%%%%%%%%%%%%%%%%%%%%%%%%%%%%%%

%%%%%%%%%%%%%%%%%%%%%%%%%%%%%%%%%%%%%%%%%%%%%%%%%%%%%%%%%%%%
%%%%%%%%%%%%%%%%%%%%%%%%%%%%%%%%%%%%%%%%%%%%%%%%%%%%%%%%%%%%
%%%%%%%%%%%%%%%%%%%%%%%%%%%%%%%%%%%%%%%%%%%%%%%%%%%%%%%%%%%%
%%%%%%%%%%%%%%%%%%%%%%%%%%%%%%%%%%%%%%%%%%%%%%%%%%%%%%%%%%%%
%%%%%%%%%%%%%%%%%%%%%%%%%%%%%%%%%%%%%%%%%%%%%%%%%%%%%%%%%%%%
%%%%%%%%%%%%%%%%%%%%%%%%%%%%%%%%%%%%%%%%%%%%%%%%%%%%%%%%%%%%
%%%%%%%%%%%%%%%%%%%%%%%%%%%%%%%%%%%%%%%%%%%%%%%%%%%%%%%%%%%%
The next section is dedicated to the dynamical localization result.

%%%%%%%%%%%%%%%%%%%%%%%%%%%%%%%%%%%%%%%%%%%%%%%%%%%%%%%%%%%%

\section{Sub-critical regime: dynamical localization}\label{sec:DL}

%%%%%%%%%%%%%%%%%%%%%%%%%%%%%%%%%%%%%%%%%%%%%%%%%%%%%%%%%%%%
%%%%%%%%%%%%%%%%%%%%%%%%%%%%%%%%%%%%%%%%%%%%%%%%%%%%%%%%%%%%
We start our proof of Theorem \ref{thm:DL}. In Section \ref{sec:FME}, we present our estimates on the fractional moments of the Green's function. We then use these estimates to show the stretched exponential decay of the correlators in Section \ref{sec:correlators}. We prove Propositions \ref{thm:consequences-DL}, \ref{thm:SULE} and Theorem \ref{thm:lower-bound-DL} in Section \ref{sec:consequences-DL}.

The reader will see that some estimates require the asymptotics for the second system of coordinates. Once again, we will give full proofs only in the cases requiring the first system, the other cases being handled in the exact same way.
%In the following, we drop the dependence on $\lambda$ to lighten the notation.
%%%%%%%%%%%%%%%%%%%%%%%%%%%%%%%%%%%%%%%%%%%%%%%%%%%%%%%%%%%%
%%%%%%%%%%%%%%%%%%%%%%%%%%%%%%%%%%%%%%%%%%%%%%%%%%%%%%%%%%%%
%%%%%%%%%%%%%%%%%%%%%%%%%%%%%%%%%%%%%%%%%%%%%%%%%%%%%%%%%%%%
%%%%%%%%%%%%%%%%%%%%%%%%%%%%%%%%%%%%%%%%%%%%%%%%%%%%%%%%%%%%

\subsection{Fractional moments estimates}\label{sec:FME}

%%%%%%%%%%%%%%%%%%%%%%%%%%%%%%%%%%%%%%%%%%%%%%%%%%%%%%%%%%%%
%%%%%%%%%%%%%%%%%%%%%%%%%%%%%%%%%%%%%%%%%%%%%%%%%%%%%%%%%%%%
The key tool of our proof of dynamical localization will be an estimate on the Green's function of the operator $\D$ in boxes contained in Theorem \ref{thm:FM-Green}. The organization of this section is as follows: we start with some simple results involving the resolvent identities in Section \ref{sec:FM-preliminaries}. Then, we use these to bound the fractional moments of the Green's function by negative fractional moments of the norm of transfer matrices in Section \ref{sec:FM-green-to-TM}. Finally, we show their stretched exponential decay in Section \ref{sec:FM-TM}.

%%%%%%%%%%%%%%%%%%%%%%%%%%%%%%%%%%%%%%%%%%%%%%%%%%%%%%%%%%%%
%%%%%%%%%%%%%%%%%%%%%%%%%%%%%%%%%%%%%%%%%%%%%%%%%%%%%%%%%%%%
We define two collections of boxes: for $l\geq 1$, let
\begin{eqnarray}\label{boxes}
	\Lambda_l &=& \{ (u,+),\, 1 \leq u\leq l-1,\, (u,-),\, 1 \leq u\leq l\}\notag \\
	\Lambda' _l&=& \{ (u,\pm),\, 1 \leq u\leq l\}.
\end{eqnarray} 
%%%%%%%%%%%%%%%%%%%%%% Boxes
%%%%%%%%%%%%%%%%%%%%%%%%%%%%%%%%%%%%%%%%%%%%%%%%%%%%%%%%%%%%
\begin{figure}[h!]
\begin{center}
\includegraphics[scale=0.4]{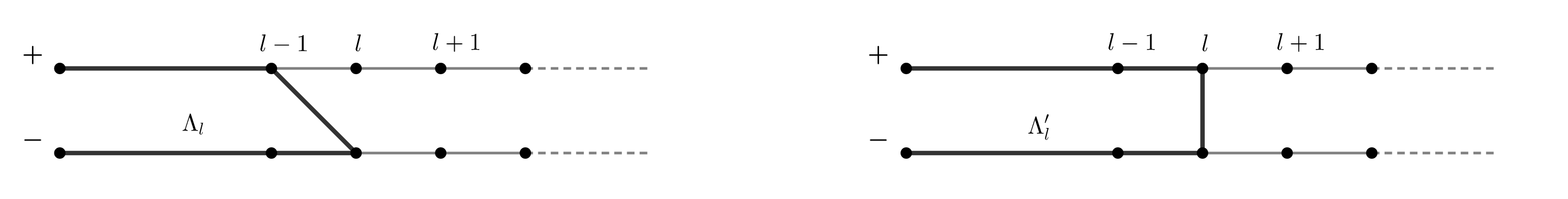}
\vspace*{-0.2cm}
\caption{Boxes $\Lambda_l$ and $\Lambda_l'$.} \label{fig:boxes}
\end{center}
\end{figure}
%%%%%%%%%%%%%%%%%%%%%%%%%%%%%%%%%%%%%%%%%%%%%%%%%%%%%%%%%%%%
%%%%%%%%%%%%%%%%%%%%%%%%%%%%%%%%%%%%%%%%%%%%%%%%%%%%%%%%%%%%

We let $P_l$ to be the projection on $\Lambda_l$ and  $D_{\omega,l}=P_l \D P_l$ is the restriction of $\D$ to the box $\Lambda_l$ acting on $\ell^2(\Lambda_l)$. We denote its resolvent by $R_{\omega,l}(E) = (D_{\omega,l}-E)^{-1}$ and by $G_{\omega,l}$ the corresponding Green's function 
\begin{eqnarray*}
	G_{\omega,l}(u,\sigma;n,\sigma';E)= \bra \delta^{\sigma}_u, R_{\omega,l}(E) \, \delta^{\sigma'}_n \ket.
\end{eqnarray*}
We define $P'_l,\, D'_{\omega,l}, \, R'_{\omega,l}$ and $G'_{\omega,l}$ in the same way.
%%%%%%%%%%%%%%%%%%%%%%%%%%%%%%%%%%%%%%%%%%%%%%%%%%%%%%%%%%%%

\begin{theorem}\label{thm:FM-Green}
	Let $0<\alpha<\frac12$ and $\lambda>0$. Assume \textbf{(A1)}-\textbf{(A3a)} and \textbf{(A5)}. Then for all $u\in\n^*$ and all compact energy interval $I\subset \mathring\Sigma$, there exist constants $c=c(u,I)>0$ and $C=C(u,I)>0$ such that 
	\begin{eqnarray*}
		\esp\left[ \left| G_{\omega,L}(u,\pm;n,-;E)\right|^s\right] \leq C \lambda^{-2s} a_n^{-2s} \e^{-cn^{1-2\alpha}},
	\end{eqnarray*}
	and
	\begin{eqnarray*}
		\esp\left[ \left| G'_{\omega,L}(u,\pm;n,+;E)\right|^s\right] \leq C \lambda^{-2s} a_n^{-2s} \e^{-cn^{1-2\alpha}},
	\end{eqnarray*}
	for all $1\leq u \leq L$, $1 \leq n \leq L$ and $E\in I$.
\end{theorem}
%%%%%%%%%%%%%%%%%%%%%%%%%%%%%%%%%%%%%%%%%%%%%%%%%%%%%%%%%%%%
%%%%%%%%%%%%%%%%%%%%%%%%%%%%%%%%%%%%%%%%%%%%%%%%%%%%%%%%%%%%
The reason to introduce two different systems of boxes comes from the fact that the estimates above require to use the first and second system of coordinates respectively. The scheme of proof in both cases is exactly the same.
%%%%%%%%%%%%%%%%%%%%%%%%%%%%%%%%%%%%%%%%%%%%%%%%%%%%%%%%%%%%
%%%%%%%%%%%%%%%%%%%%%%%%%%%%%%%%%%%%%%%%%%%%%%%%%%%%%%%%%%%%
%%%%%%%%%%%%%%%%%%%%%%%%%%%%%%%%%%%%%%%%%%%%%%%%%%%%%%%%%%%%
%%%%%%%%%%%%%%%%%%%%%%%%%%%%%%%%%%%%%%%%%%%%%%%%%%%%%%%%%%%%

\subsubsection{Preliminaries}\label{sec:FM-preliminaries}

%%%%%%%%%%%%%%%%%%%%%%%%%%%%%%%%%%%%%%%%%%%%%%%%%%%%%%%%%%%%
%%%%%%%%%%%%%%%%%%%%%%%%%%%%%%%%%%%%%%%%%%%%%%%%%%%%%%%%%%%%
We express the full operator in terms of the canonical basis so that
\begin{eqnarray}\label{Dirac expansion}
	D_{\omega} 
	&=& 
	\sum_{j\geq 2} 
	\Big{\{}
		\deltap_j\left( m \langle \deltap_j| + \langle \deltam_j| - \langle \deltam_{j+1}| \right)
		+
		 \deltam_j\left( \langle \deltap_j| - \langle \deltap_{j-1}| + m\langle \deltam_j|\right)
	\Big{\}}
	\\
	&&
	\phantom{blablablabla}
	+
	\deltap_1\left( m \langle \deltap_1| + \langle \deltam_1| - \langle \deltam_{2}| \right)
	+
		\deltam_1\left( \langle \deltap_1| + m\langle \deltam_1|\right)
	\\
	&+&	
	\sum_{j\geq 1}
	\Big{\{}
		\vo(j)\deltap_j\langle \deltap_j|
		+
		\vt(j)\deltam_j\langle \deltam_j|
	\Big{\}}.
\end{eqnarray}
%We define two collections of boxes:
%\begin{eqnarray}
%	\Lambda_l &=& \{ (m,+),\, m\leq l-1,\, (m,-),\, m\leq l\}\\
%	\Lambda' _l&=& \{ (m,\pm),\, m\leq l\}
%\end{eqnarray} 
%We denote by $P_l$ the projection on $\Lambda_l$ and let $D_l=P_lDP_l$, $R_l = (D_l-E)^{-1}$ and $G_l$ the Green's function. We define $P'_l,\, D'_l, \, R'_l$ and $G'_l(m,n)$ in the same way.
%%%%%%%%%%%%%%%%%%%%%%%%%%%%%%%%%%%%%%%%%%%%%%%%%%%%%%%%%%%%
%%%%%%%%%%%%%%%%%%%%%%%%%%%%%%%%%%%%%%%%%%%%%%%%%%%%%%%%%%%%

%%%%%%%%%%%%%%%%%%%%%%%%%%%%%%%%%%%%%%%%%%%%%%%%%%%%%%%%%%%%
%%%%%%%%%%%%%%%%%%%%%%%%%%%%%%%%%%%%%%%%%%%%%%%%%%%%%%%%%%%% Operator action 
\begin{figure}[h!]
\begin{center}
\includegraphics[scale=0.6]{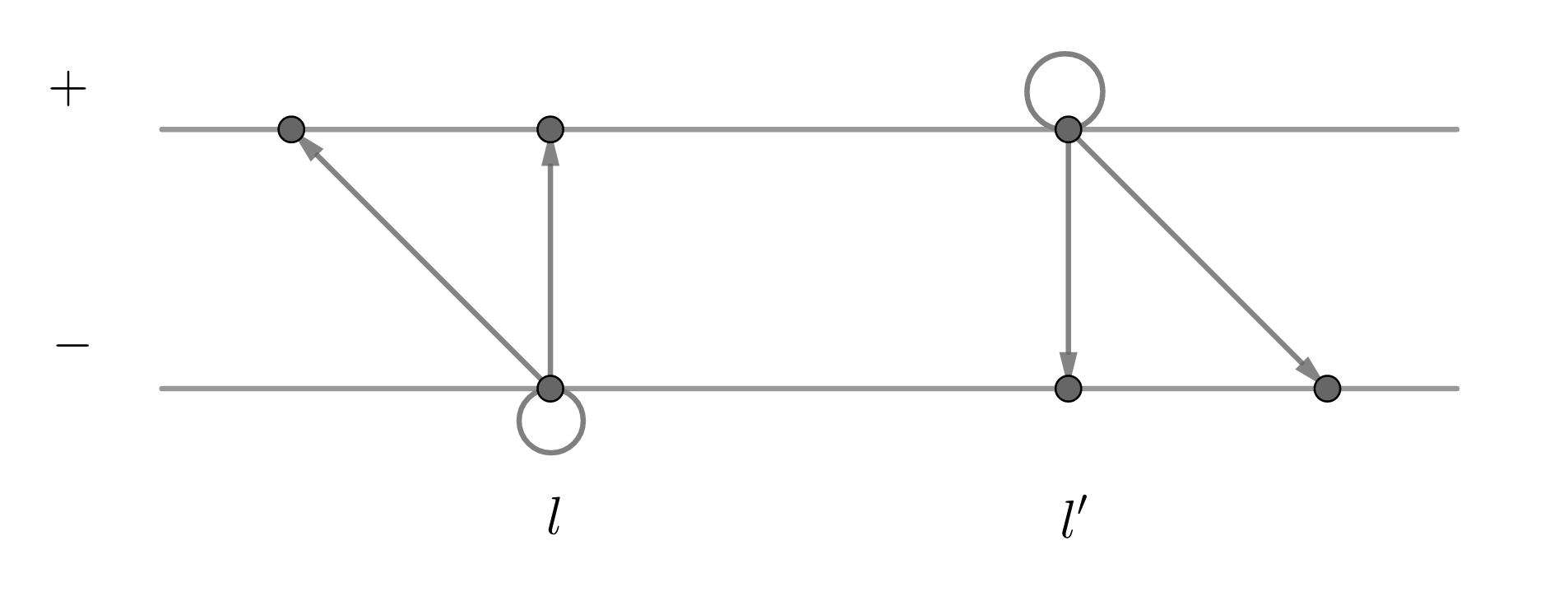}
\vspace*{-0.2cm}
\caption{The action of the operator $\D$ according to the spin position.} \label{figd2}
\end{center}
\end{figure}
%%%%%%%%%%%%%%%%%%%%%%%%%%%%%%%%%%%%%%%%%%%%%%%%%%%%%%%%%%%%
%%%%%%%%%%%%%%%%%%%%%%%%%%%%%%%%%%%%%%%%%%%%%%%%%%%%%%%%%%%%

Let $\Phi$ be the vector obtained from $\Phi_1$ through the transfer matrix recurrence: $\Phi_{n+1} = \transn \Phi_n$. This way, $(\D-E)\Phi=0$.
Note that we are using the first system of coordinates. The second system will appear naturally.
%%%%%%%%%%%%%%%%%%%%%%%%%%%%%%%%%%%%%%%%%%%%%%%%%%%%%%%%%%%%
%%%%%%%%%%%%%%%%%%%%%%%%%%%%%%%%%%%%%%%%%%%%%%%%%%%%%%%%%%%%
We begin with some identities involving $\Phi$ and the resolvents.
%%%%%%%%%%%%%%%%%%%%%%%%%%%%%%%%%%%%%%%%%%%%%%%%%%%%%%%%%%%%
%%%%%%%%%%%%%%%%%%%%%%%%%%%%%%%%%%%%%%%%%%%%%%%%%%%%%%%%%%%%
\begin{lemma}\label{thm:resolvents-to-eigenfunctions}
	For all $n\geq 1$ and $u\in[1,n]$, we have
	\begin{eqnarray}
		\label{eq:resolvents-to-eigenfunctions-1}
		&&
		G_{\omega,n}(u,\pm; n,-;E) = -\frac{\vp^{\pm}_u}{\vp^+_n}=\frac{\vp^{\pm}_u}{\vp^-_n} \, G_{\omega,n}(n,-;n,-;E),
		\\
		\label{eq:resolvents-to-eigenfunctions-2}
		&&
		G'_{\omega,n}(u,\pm; n,+;E) = \frac{\vp^{\pm}_u}{\vp^-_{n+1}}=\frac{\vp^{\pm}_u}{\vp^+_n} \, G'_{\omega,n}(n,+;n,+;E).
	\end{eqnarray}
\end{lemma}
%%%%%%%%%%%%%%%%%%%%%%%%%%%%%%%%%%%%%%%%%%%%%%%%%%%%%%%%%%%%
%%%%%%%%%%%%%%%%%%%%%%%%%%%%%%%%%%%%%%%%%%%%%%%%%%%%%%%%%%%%
\begin{proof}
Using the expansion \eqref{Dirac expansion}, we decompose $\D$ as 
\begin{eqnarray*}
	D_{\omega}-E = D_{\omega,n}-EP_n + \deltam_n\langle \deltap_n|  + P_n^{\perp}C,
\end{eqnarray*}
for some bounded operator $C$. In particular, we have
\begin{eqnarray*}
	0 = (D_{\omega,n}-EP_n)\Phi + \vp^+_n \deltam_n,
\end{eqnarray*}
and
\begin{eqnarray*}
	\Phi = -\vp^+_n R_{\omega,n}(E)\deltam_n,
\end{eqnarray*}
which yields
\begin{eqnarray*}
	\vp^{\pm}_u = -\vp^+_n G_{\omega,n}(u,\pm; n,-;E)
\end{eqnarray*}
The first identity in \eqref{eq:resolvents-to-eigenfunctions-2} holds in the same way from the decomposition
\begin{eqnarray*}
	D_{\omega}-E = D'_{\omega,n}-EP_n -\deltap_n \langle \deltam_{n+1}|  + (P'_n)^{\perp}C',
\end{eqnarray*}
for some bounded operator $C'$.
Now, observe that the restrictions $D_{\omega,n}$ and $D_{\omega,n}'$ are related by
\begin{eqnarray*}
	D_{\omega,n}
	=
	D'_{\omega,n-1}-\deltap_{n-1}\langle \deltam_n| -\deltam_n\langle \deltap_{n-1}|, 
\end{eqnarray*}
so that
\begin{eqnarray*}
	D_{\omega,n} - EP_n
	=
	D'_{\omega,n-1}-EP'_{n-1} -E\deltam_n \langle \deltam_n|-\deltap_{n-1}\langle \deltam_n| -\deltam_n\langle \deltap_{n-1}|. 
\end{eqnarray*}
It follows from the resolvent identity that
\begin{eqnarray*}
	R_{\omega,n} (E)- R'_{\omega,n-1}(E)
	=
	R'_{\omega,n-1}(E)
	\left(
		E\deltam_n \langle \deltam_n|+\deltap_{n-1}\langle \deltam_n| +\deltam_n\langle \deltap_{n-1}| 
	\right)
	R_{\omega,n}(E),
\end{eqnarray*}
from where we obtain
\begin{eqnarray*}
	G_{\omega,n}(u,\pm;n,-;E)
	&=&
	G'_{\omega,n-1}(u,\pm;n-1,+;E)\ G_{\omega,n}(n,-;n,-;E)
	\\
	&=&
	\frac{\vp^{\pm}_u}{\vp^-_n} \, G_{\omega,n}(n,-;n,-;E),
\end{eqnarray*}
using first identity in \eqref{eq:resolvents-to-eigenfunctions-2}. The second identity in \eqref{eq:resolvents-to-eigenfunctions-2} is obtained in a similar spirit.
\end{proof}
%%%%%%%%%%%%%%%%%%%%%%%%%%%%%%%%%%%%%%%%%%%%%%%%%%%%%%%%%%%%
%%%%%%%%%%%%%%%%%%%%%%%%%%%%%%%%%%%%%%%%%%%%%%%%%%%%%%%%%%%%

%%%%%%%%%%%%%%%%%%%%%%%%%%%%%%%%%%%%%%%%%%%%%%%%%%%%%%%%%%%%
%%%%%%%%%%%%%%%%%%%%%%%%%%%%%%%%%%%%%%%%%%%%%%%%%%%%%%%%%%%%
\begin{lemma}\label{thm:resolvents-big-to-small-box}
	For each $L\geq 1$, $1\leq n \leq L$ and $u\in [1,n]$, we have
	\begin{eqnarray*}
		&&G_{\omega,L}(u,\pm;n,-;E) 
		=
		\left( 1-G_{\omega,L}(n,+;n,-;E)\right)G_{\omega,n}(u,\pm;n,-;E),
		\\
		&&G'_{\omega,L}(u,\pm;n,+;E) 
		=
		\left( 1 + G'_{\omega,L}(n+1,-;n,+;E) \right) G'_{\omega,n}(u,\pm;n,+;E).
	\end{eqnarray*}
\end{lemma}
%%%%%%%%%%%%%%%%%%%%%%%%%%%%%%%%%%%%%%%%%%%%%%%%%%%%%%%%%%%%
%%%%%%%%%%%%%%%%%%%%%%%%%%%%%%%%%%%%%%%%%%%%%%%%%%%%%%%%%%%%
\begin{proof}
	Let $\hat{D}_{\omega,n}=P_n D_{\omega,L} P_n + P_n^{\perp}D_{\omega,L}P_n^{\perp}$ and $\hat{R}_{\omega,n}(E)$ and $\hat{G}_{\omega,n}$ be the corresponding resolvent and Green's function. Then,
	\begin{eqnarray*}
		D_{\omega,L} = \hat{D}_{\omega,n} + \deltap_n \bra \deltam_n | + \deltam_n \bra \deltap_n |
	\end{eqnarray*}
	By the resolvent identity, one has
	\begin{eqnarray*}
		R_{\omega,L}(E) - \hat{R}_{\omega,n}(E)
		&=&
		\hat{R}_{\omega,n}(E)
		\left(
			\hat{D}_{\omega,n} - D_{\omega,L}
		\right)
		R_{\omega,L}(E)
		\\
		&=&
		-\hat{R}_{\omega,n}
		\left(
			\deltap_n \bra \deltam_n | + \deltam_n \bra \deltap_n |
		\right)
		R_{\omega,L}(E).
	\end{eqnarray*}
	Hence,
	\begin{eqnarray*}
		G_{\omega,L}(u,\pm;n,-;E)
		&=&
		\hat{G}_{\omega,n}(u,\pm;n,-;E)
		-
		\hat{G}_{\omega,n}(u,\pm;n,-;E)
		G_{\omega,L}(n,+;n,-;E)
		\\
		&=&
		{G}_{\omega,n}(u,\pm;n,-;E)
		-
		{G}_{\omega,n}(u,\pm;n,-;E)
		G_{\omega,L}(n,+;n,-;E).
	\end{eqnarray*}
	For the second identity, consider $\check{D}_{\omega,n} = P'_n D_{\omega,L} P'_n+(P'_n)^{\perp}D_{\omega,L}(P'_n)^{\perp}$ and proceed in the same way.
%	let $\check{R}_n$ and $\check{G}_n$ be the corresponding resolvent and Green's function. Then,
%	\begin{eqnarray}
%		D_L
%		=
%		\check{D}_n - \deltap_n \bra \deltam_{n+1} | - \deltam_{n+1} \bra \deltap_n |
%	\end{eqnarray}
%	Hence,
%	\begin{eqnarray}
%		R_L - \check{R}_n
%		&=&
%		\check{R}_n
%		\left(
%			\check{D}_n - D_L
%		\right)
%		R_L
%		\\
%		&=&
%		\check{R}_n
%		\left(
%			\deltap_n \bra \deltam_{n+1} | + \deltam_{n+1} \bra \deltap_n |
%		\right)
%		R_L,
%	\end{eqnarray}
%	and
%	\begin{eqnarray}
%		G_L(m,\pm;n,+)
%		=
%		\check{G}_n(m,\pm;n,+)
%		+
%		\check{G}_n(m,\pm;n,+)
%		G_L(n+1,-;n,+).
%	\end{eqnarray}
\end{proof}
%%%%%%%%%%%%%%%%%%%%%%%%%%%%%%%%%%%%%%%%%%%%%%%%%%%%%%%%%%%%
%%%%%%%%%%%%%%%%%%%%%%%%%%%%%%%%%%%%%%%%%%%%%%%%%%%%%%%%%%%%

\begin{figure}[h!]
\begin{center}
\includegraphics[scale=0.5]{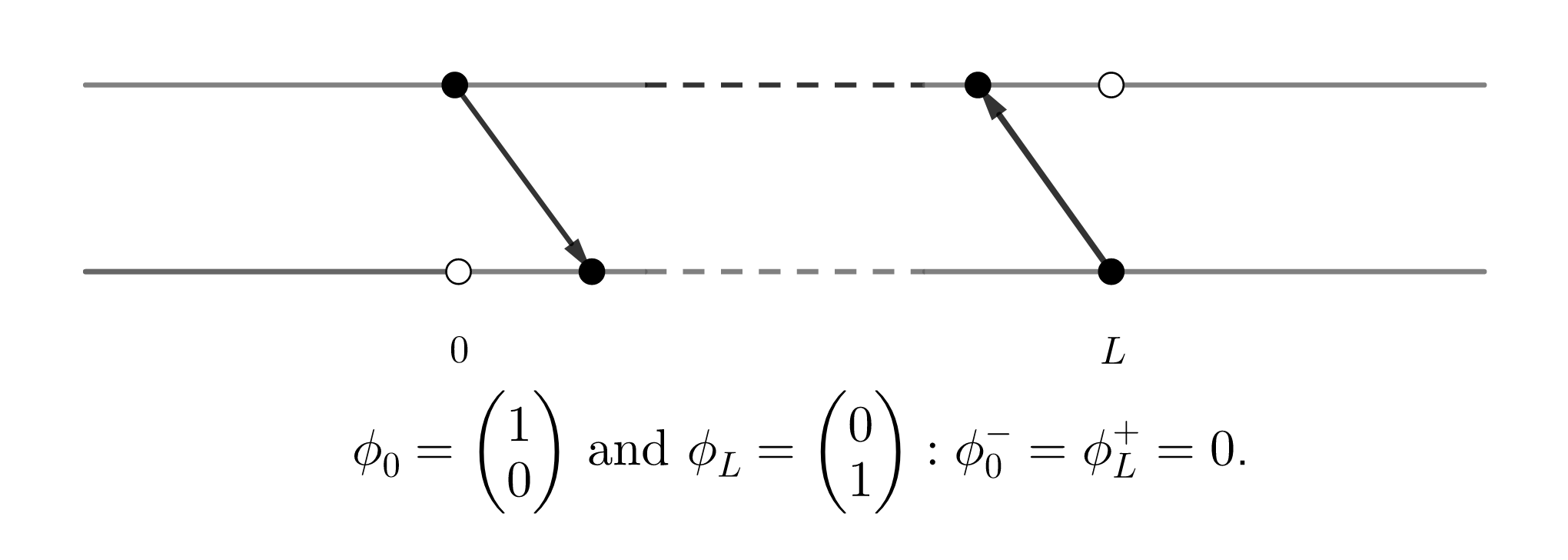}
\vspace*{-0.2cm}
\caption{Boundary conditions on $[0,L]$.} \label{figd3}
\end{center}
\end{figure}

%%%%%%%%%%%%%%%%%%%%%%%%%%%%%%%%%%%%%%%%%%%%%%%%%%%%%%%%%%%%
%%%%%%%%%%%%%%%%%%%%%%%%%%%%%%%%%%%%%%%%%%%%%%%%%%%%%%%%%%%%
%%%%%%%%%%%%%%%%%%%%%%%%%%%%%%%%%%%%%%%%%%%%%%%%%%%%%%%%%%%%
%%%%%%%%%%%%%%%%%%%%%%%%%%%%%%%%%%%%%%%%%%%%%%%%%%%%%%%%%%%%

\subsubsection{From Green's functions to transfer matrices}\label{sec:FM-green-to-TM}

%%%%%%%%%%%%%%%%%%%%%%%%%%%%%%%%%%%%%%%%%%%%%%%%%%%%%%%%%%%%
%%%%%%%%%%%%%%%%%%%%%%%%%%%%%%%%%%%%%%%%%%%%%%%%%%%%%%%%%%%%
We consider $1\leq u \leq n$ and
apply Lemma \ref{thm:resolvents-to-eigenfunctions} and \ref{thm:resolvents-big-to-small-box} 
 dropping temporarily the dependence on $E$ to lighten the notation,
\begin{eqnarray*}
	G_{\omega,L}(u,\pm;n,-)
	&=&
	\left( 1-G_{\omega,L}(n,+;n,-)\right)G_{\omega,n}(u,\pm;n,-)
	\\
	&=&
	-\left( 1-G_{\omega,L}(n,+;n,-)\right) \frac{\vp^{\pm}_u}{\vp^+_n}
	\\
	&=&
	\left( 1-G_{\omega,L}(n,+;n,-)\right)G_{\omega,n}(n,-;n,-) \frac{\vp^{\pm}_u}{\vp^-_n}.
\end{eqnarray*}
Hence,
\begin{eqnarray*}
	|G_{\omega,L}(u,\pm;n,-)|
	&\leq&
	\left( 1+|G_{\omega,L}(n,+;n,-)|\right) \left( 1 + |G_{\omega,n}(n,-;n,-)| \right)
	\frac{|\vp^{\pm}|}{\max\{|\vp^+_n|,|\vp^-_n|\}}
	\\
	&\leq&
	C
	\left( 1+|G_{\omega,L}(n,+;n,-)|\right) \left( 1 + |G_{\omega,n}(n,-;n,-)| \right)
	\frac{\| \Phi_u \|}{\| \Phi_n\|},
\end{eqnarray*}
for some $C>0$.
Now, we note that 
\begin{eqnarray*}
	\frac{\| \Phi_u \|}{\| \Phi_n\|} = \| {\bf T}_{\omega,[u,n]} \hat{\Phi}_u\|^{-1},
\end{eqnarray*}
where $\displaystyle \hat{\Phi}_u= \frac{\Phi_u}{\| \Phi_u \|}$. Let $s\in(0,1)$. By subadditivity and H\"older's inequality, we have
\begin{eqnarray}
	\nonumber
	\esp\left[ |G_{\omega,L}(u,\pm;n,-;E)|^s \right]
	&\leq&
	C
	\left( 1+\esp\left[|G_{\omega,L}(n,+;n,-;E)|^{4s}\right]\right)^{1/4}
	\\
	&&
	\nonumber
	\times
	\left( 1+\esp\left[|G_{\omega,n}(n,-;n,-;E)|^{4s}\right]\right)^{1/4}
	\\
	&&
	\label{eq:resolvent-to-transfer-matrix-0}
	\times
	\esp\left[ \| {\bf T}_{\omega,[u,n]}(E) \hat{\Phi}_u\|^{-2s}\right]^{1/2}.
\end{eqnarray}
%%%%%%%%%%%%%%%%%%%%%%%%%%%%%%%%%%%%%%%%%%%%%%%%%%%%%%%%%%%%
%%%%%%%%%%%%%%%%%%%%%%%%%%%%%%%%%%%%%%%%%%%%%%%%%%%%%%%%%%%%
The following is an a priori estimate on the moments of the Green's function (see \cite[Corollary 8.4]{AW}). The lemma requires some regularity of the law. Let $I\subset\R$ be a bounded interval. Using \textbf{(A6a)},
\begin{eqnarray*}
	\p[V_{\omega,i}(n) \in I]
	=
	\int_I \rho_{n,i}(y)\, dy
	\leq
	\left(
		\int \rho_{n,i}(y)^p\, dy
	\right)^{\frac{1}{p}}
	|I|^{\frac{p-1}{p}}
	\leq
	C (\lambda a_n)^{-\frac{\gamma}{p}}|I|^{\frac{p-1}{p}},
\end{eqnarray*}
so that, in the terminology of \cite[Definition 4.5]{AW} the law of $V_{\omega,i}(n)$ is uniformly $\tau$-H\"older continuous with $\tau=\frac{p-1}{p}$ and \cite[Corollary 8.4]{AW} can be applied to show that the fractional moments of the Green's function are bounded for $s\in(0,\frac{p-1}{p})$.
We state the result for the Green function $G$ but the exact same bound holds for $G'$.
%%%%%%%%%%%%%%%%%%%%%%%%%%%%%%%%%%%%%%%%%%%%%%%%%%%%%%%%%%%%
%%%%%%%%%%%%%%%%%%%%%%%%%%%%%%%%%%%%%%%%%%%%%%%%%%%%%%%%%%%%
\begin{lemma}\label{thm:FM-a-priori}
	Assume \textbf{(A1)} and \textbf{(A6a)}.
	For each compact energy interval $I\subset\mathring\Sigma$ and each $s\in(0,\frac{p-1}{p})$, there exists  $C=C(s,I)\in(0,\infty)$ 
	%and $\kappa=\kappa(s,\gamma)\geq 0$ 
	such that
	\begin{eqnarray*}
		\esp\left[ \left| G_{\omega,L}(u,\sigma;n,\sigma';E) \right|^s\right]
		\leq C
		\lambda^{-s} \left( a_u^{-s}+a_n^{-s}\right),
	\end{eqnarray*}
	for all $L\geq 1$, $u,n\in[1,L]$, $\sigma, \sigma'\in\{+,-\}$ and $E\in I$.
\end{lemma}
%%%%%%%%%%%%%%%%%%%%%%%%%%%%%%%%%%%%%%%%%%%%%%%%%%%%%%%%%%%%
%%%%%%%%%%%%%%%%%%%%%%%%%%%%%%%%%%%%%%%%%%%%%%%%%%%%%%%%%%%%
We state the final form of estimate \eqref{eq:resolvent-to-transfer-matrix-0} as a lemma.
%%%%%%%%%%%%%%%%%%%%%%%%%%%%%%%%%%%%%%%%%%%%%%%%%%%%%%%%%%%%
%%%%%%%%%%%%%%%%%%%%%%%%%%%%%%%%%%%%%%%%%%%%%%%%%%%%%%%%%%%%
\begin{lemma}
	Assume \textbf{(A1)} and \textbf{(A6a)}.
	For each compact energy interval $I\subset \mathring\Sigma$ and each $s\in(0,\frac{p-1}{p})$, there exists  $C=C(s,I)\in(0,\infty)$ and $\kappa=\kappa(s,\gamma)\geq 0$ such that
	\begin{eqnarray}\label{eq:resolvent-to-transfer-matrix}
		\esp\left[ |G_{\omega,L}(u,\pm;n,-;E)|^s \right]
		&\leq&
		C
		(\lambda a_n)^{-\kappa}
		\esp\left[ \| {\bf T}_{\omega,[u,n]}(E) \hat{\Phi}_u\|^{-2s}\right]^{1/2},
	\end{eqnarray}
	for all $u,n\in[1,L]$ and $E\in I$.
\end{lemma}
%%%%%%%%%%%%%%%%%%%%%%%%%%%%%%%%%%%%%%%%%%%%%%%%%%%%%%%%%%%%
%%%%%%%%%%%%%%%%%%%%%%%%%%%%%%%%%%%%%%%%%%%%%%%%%%%%%%%%%%%%

%%%%%%%%%%%%%%%%%%%%%%%%%%%%%%%%%%%%%%%%%%%%%%%%%%%%%%%%%%%%
%%%%%%%%%%%%%%%%%%%%%%%%%%%%%%%%%%%%%%%%%%%%%%%%%%%%%%%%%%%%
%%%%%%%%%%%%%%%%%%%%%%%%%%%%%%%%%%%%%%%%%%%%%%%%%%%%%%%%%%%%
%%%%%%%%%%%%%%%%%%%%%%%%%%%%%%%%%%%%%%%%%%%%%%%%%%%%%%%%%%%%
%CHECK AND WIPE FROM HERE-
%\begin{eqnarray}
%	G'_L(m,\pm;n,-)
%	&=&
%	\left( 1+G'_L(n+1,-;n,+)\right)G'_n(m,\pm;n,+)
%	\\
%	&=&
%	\left( 1-G'_L(n+1,-;n,+)\right) \frac{x^{\pm}_m}{x^-_{n+1}}
%	\\
%	&=&
%	\left( 1-G'_L(n+1,-;n,+)\right)G'_n(n,+;n,+) \frac{x^{\pm}_m}{x^+_n}.
%\end{eqnarray}
%Hence,
%\begin{eqnarray}
%	\nonumber
%	|G'_L(m,\pm;n,+)|
%	&\leq&
%	\left( 1+|G'_L(n+1,-;n,+)|\right) \left( 1 + |G'_n(n,+;n,+)| \right)
%	\frac{|x^{\pm}_m|}{\max\{|x^-_{n+1}|,|x^+_n|\}}.
%	\\
%	&&
%\end{eqnarray}
%We break this into two cases:
%\begin{eqnarray}
%	|G'_L(m,+;n,+)|
%	&\leq&
%	C
%	\left( 1+|G'_L(n+1,-;n,+)|\right) \left( 1 + |G'_n(n,+;n,+)| \right)
%	\frac{\| X'_m \|}{\|X_n'\|},
%	\\
%	|G'_L(m,-;n,+)|
%	&\leq&
%	C
%	\left( 1+|G'_L(n+1,-;n,+)|\right) \left( 1 + |G'_n(n,+;n,+)| \right)
%	\frac{\| X'_{m-1} \|}{\|X_n'\|}.
%\end{eqnarray}
%- TO HERE.

%%%%%%%%%%%%%%%%%%%%%%%%%%%%%%%%%%%%%%%%%%%%%%%%%%%%%%%%%%%%
%%%%%%%%%%%%%%%%%%%%%%%%%%%%%%%%%%%%%%%%%%%%%%%%%%%%%%%%%%%%
By similar arguments,
\begin{eqnarray}
	\label{eq:resolvent-to-transfer-matrix-2}
	\esp\left[ |G_{\omega,L}(u,+;n,+;E)|^s \right]
	&\leq&
	C
	\lambda^{-\kappa} a_n^{-\kappa}
	\esp\left[ \| {\bf T}'_{\omega,[u,n]}(E) \hat{\Phi}_u\|^{-2s}\right]^{1/2},
	\\
	\label{eq:resolvent-to-transfer-matrix-3}
	\esp\left[ |G_{\omega,L}(u,-;n,+;E)|^s \right]
	&\leq&
	C
	\lambda^{-\kappa} a_n^{-\kappa}
	\esp\left[ \| {\bf T}'_{\omega,[u-1,n]}(E) \hat{\Phi}_{u-1}\|^{-2s}\right]^{1/2}.
\end{eqnarray}
%%%%%%%%%%%%%%%%%%%%%%%%%%%%%%%%%%%%%%%%%%%%%%%%%%%%%%%%%%%%
%%%%%%%%%%%%%%%%%%%%%%%%%%%%%%%%%%%%%%%%%%%%%%%%%%%%%%%%%%%%

%%%%%%%%%%%%%%%%%%%%%%%%%%%%%%%%%%%%%%%%%%%%%%%%%%%%%%%%%%%%
%%%%%%%%%%%%%%%%%%%%%%%%%%%%%%%%%%%%%%%%%%%%%%%%%%%%%%%%%%%%
%%%%%%%%%%%%%%%%%%%%%%%%%%%%%%%%%%%%%%%%%%%%%%%%%%%%%%%%%%%%
%%%%%%%%%%%%%%%%%%%%%%%%%%%%%%%%%%%%%%%%%%%%%%%%%%%%%%%%%%%%

\subsubsection{Estimates on transfer matrices}\label{sec:FM-TM}

%%%%%%%%%%%%%%%%%%%%%%%%%%%%%%%%%%%%%%%%%%%%%%%%%%%%%%%%%%%%
%%%%%%%%%%%%%%%%%%%%%%%%%%%%%%%%%%%%%%%%%%%%%%%%%%%%%%%%%%%%

The next key lemma provides the decay of the negative moments of transfer matrices needed to complete the proof of Theorem \ref{thm:FM-Green}. Its proof is inspired by \cite[Lemma 5.1]{CKM} where exponential decay was obtained in the ergodic case and used as an input for a multi-scale analysis. Our non-ergodic case requires some finer estimates and leads to stretched exponential decay.

All the estimates in this section are stated for the first system of coordinates. Once again, the exact same bounds hold for the second system.
%%%%%%%%%%%%%%%%%%%%%%%%%%%%%%%%%%%%%%%%%%%%%%%%%%%%%%%%%%%%
%%%%%%%%%%%%%%%%%%%%%%%%%%%%%%%%%%%%%%%%%%%%%%%%%%%%%%%%%%%%
\begin{lemma}\label{thm:key-lemma}
Let $0<\alpha< \frac12$, assume \textbf{(A1)}-\textbf{(A3a)} and \textbf{(A5)}. For each compact interval $I\subset \mathring\Sigma$, there exists $n_0 = n_0(I) \geq 1$, $s_0=s_0(I)\in(0,\frac12)$ and $c=c(I)>0$ such that for each $u\in\n^*$, there exists $C=C(m,I)\in(0,\infty)$  such that 
	\begin{eqnarray}\label{eq:bound-key-lemma}
		\esp\left[
			\left\|
			{\bf T}_{[u,n]}(E)
			{\varphi}_0
			\right\|^{-s} 
			\right] 
			\leq 
			C\e^{-cn^{1-2\alpha}},
	\end{eqnarray}
	for all $s\in(0,s_0]$, $\|\varphi_0\|=1$, $E\in I$ and $n\geq u+n_0$.
\end{lemma}
%%%%%%%%%%%%%%%%%%%%%%%%%%%%%%%%%%%%%%%%%%%%%%%%%%%%%%%%%%%%
%%%%%%%%%%%%%%%%%%%%%%%%%%%%%%%%%%%%%%%%%%%%%%%%%%%%%%%%%%%%

We start with some preliminaries. Lemma \ref{thm:bound-T} is a simple bound on the moments of the norm of the transfer matrices. Lemma \ref{thm:first-bound-Tmn} is the initial step of the recursion in the proof of Lemma \ref{thm:key-lemma}.

%%%%%%%%%%%%%%%%%%%%%%%%%%%%%%%%%%%%%%%%%%%%%%%%%%%%%%%%%%%%
%%%%%%%%%%%%%%%%%%%%%%%%%%%%%%%%%%%%%%%%%%%%%%%%%%%%%%%%%%%%
\begin{lemma}\label{thm:bound-T}
	Assume \textbf{(A3a)}. For all  compact interval $I\subset\mathring\Sigma$, there exists a constant $M=M(I)\in(0,\infty)$ such that
	\begin{eqnarray}\label{eq:bound-T}
		\esp[\| \transn(E) \|^s] \leq M,
	\end{eqnarray}
	for all $s\in[0,1]$, $E\in I$ and $n\geq 1$.
\end{lemma}
%%%%%%%%%%%%%%%%%%%%%%%%%%%%%%%%%%%%%%%%%%%%%%%%%%%%%%%%%%%%
%%%%%%%%%%%%%%%%%%%%%%%%%%%%%%%%%%%%%%%%%%%%%%%%%%%%%%%%%%%%
\begin{proof}
	From \eqref{eq:decomposition-transfer-matrices}, we can see that there exists a constant $C=C(I)\in(0,\infty)$ such that
	\begin{eqnarray*}
		\| \transn(E) \|
		&\leq&
		C
		\Big{(}
			1+|\vo(n)|+|\vt(n+1)|
			+
			|\vo(n)|\, |\vt(n+1)|
		\Big{)}.
	\end{eqnarray*}
	Hence,
	\begin{eqnarray*}
		\esp[\| \transn(E) \|^s]
		&\leq&
		C^s
		\Big{(}
			1+\esp[|\vo(n)|^s]+\esp[|\vt(n+1)|^s]
			\\
			&&
			\phantom{blablablabla}
			+\esp[|\vo(n)|^p]\, \esp[|\vt(n+1)|^s]
		\Big{)},
	\end{eqnarray*}
	if $s\in[0,1]$. Now, we have
	\begin{eqnarray*}
		\esp[|\vo(n)|^s] 
		\leq 
		1 + \esp[|\vo(n)|]
		\leq
		1 + \esp[|\vo(n)|^2]^{1/2}
		=
		1 + \lambda a_n,
	\end{eqnarray*}
	which is uniformly bounded in $n$. The same holds for $\vt(n+1)$.
\end{proof}
%%%%%%%%%%%%%%%%%%%%%%%%%%%%%%%%%%%%%%%%%%%%%%%%%%%%%%%%%%%%
%%%%%%%%%%%%%%%%%%%%%%%%%%%%%%%%%%%%%%%%%%%%%%%%%%%%%%%%%%%%

%%%%%%%%%%%%%%%%%%%%%%%%%%%%%%%%%%%%%%%%%%%%%%%%%%%%%%%%%%%%
%%%%%%%%%%%%%%%%%%%%%%%%%%%%%%%%%%%%%%%%%%%%%%%%%%%%%%%%%%%%
\begin{lemma}\label{thm:first-bound-Tmn}
	Let $0<\alpha< \frac12$, assume \textbf{(A1)}-\textbf{(A3a)} and \textbf{(A5)}. Then, for all compact interval $I\subset \mathring\Sigma$, there exist
	$n_0=n_0(I)\geq 1$, $s_0=s_0(I)\in(0,\frac12)$ and $c=c(I)>0$ such that
	\begin{eqnarray*}
		\esp\left[ \|T_{\omega,ln_0 }(E)\cdots T_{\omega,(l-1)n_0+1}(E)\varphi_0\|^{-s} \right]
		\leq
		1 - \frac{c}{l^{2\alpha}},
	\end{eqnarray*}
	for all $s\in(0,s_0]$, $l\geq 1$, $\|\varphi_0\|=1$ and $E\in I$.
\end{lemma}
%%%%%%%%%%%%%%%%%%%%%%%%%%%%%%%%%%%%%%%%%%%%%%%%%%%%%%%%%%%%
%%%%%%%%%%%%%%%%%%%%%%%%%%%%%%%%%%%%%%%%%%%%%%%%%%%%%%%%%%%%
\begin{proof}
	We remove $E$ from the notation to lighten the presentation.
	From Lemma \ref{thm:bounds-on-Tmn}, we obtain $n_0=n_0(I)\geq 1$, $c_1=c_1(I)>0$ and $c_2=c_2(I)>0$ such that
	\begin{eqnarray*}
		\esp\left[ \log \|T_{\omega,ln_0 }\cdots T_{\omega,(l-1)n_0+1} \varphi_0\|\right]
		\geq c_1
		 \frac{n_0^{1-2\alpha}}{l^{2\alpha}},
	\end{eqnarray*}
	and
	\begin{eqnarray*}
		\esp\left[
			\left(
				\log \|T_{\omega,ln_0 }\cdots T_{\omega,(l-1)n_0+1}\varphi_0\|
			\right)^4
		\right]
		\leq c_2
		\frac{n_0^{1-2\alpha}}{l^{2\alpha}},
	\end{eqnarray*}
	for all $l\geq 1$, $\|\varphi_0\|=1$ and $E\in I$.
	Now, we apply the inequality $\e^y\leq 1 + y + y^2 e^{|y|}$ to $y=-s \log \|T_{\omega,ln_0} \cdots T_{\omega,(l-1)n_0+1}\varphi_0\|$ with $s\in(0,1)$ to be fixed later, so that
	\begin{eqnarray*}
		&&\esp[ \|T_{\omega,ln_0 }\cdots T_{\omega,(l-1)n_0+1}\varphi_0\|^{-s} ]
		\leq 
		1- s \esp[ \log \|T_{\omega,ln_0 }\cdots T_{\omega,(l-1)n_0+1}\varphi_0\|]
		\\
		&&\quad + \quad
		s^2 
		\esp\left[
				\left( 
					\log \|T_{\omega,ln_0 }\cdots T_{\omega,(l-1)n_0+1}\varphi_0\|
				\right)^4
			\right]^{1/2}
		\esp\left[
				\e^{2s \left|\log \|T_{\omega,ln_0 }\cdots T_{\omega,(l-1)n_0+1}\varphi_0\| \right|}
			\right]^{1/2}.
	\end{eqnarray*}
	Now,
	\begin{eqnarray*}
		\log \|T_{\omega,ln_0 }\cdots T_{\omega,(l-1)n_0+1}\varphi_0\| \leq \sum^{ln_0}_{j=(l-1)n_0+1}\log \| T_{\omega,j}\|.
	\end{eqnarray*}
	On the other hand,
	\begin{eqnarray*}
		1 &=& \| T_{\omega,(j-1)n_0+1}^{-1}\cdots T_{\omega,ln_0}^{-1} T_{\omega,ln_0}\cdots T_{\omega,(l-1)n_0+1}\varphi_0\| \\
			&\leq&
			\| T_{\omega,(l-1)n_0+1} \| \cdots \| T_{\omega,ln_0}\| \| T_{\omega,ln_0}\cdots T_{\omega,(l-1)n_0+1}\varphi_0\|,
	\end{eqnarray*}
	since $\| T_{\omega,j}^{-1}\| = \|T_{\omega,j}\|$, so that we have
	\begin{eqnarray*}
		\log \|T_{\omega,ln_0 }\cdots T_{\omega,(l-1)n_0+1}\varphi_0\| 
		\geq 
		-\sum^{ln_0}_{j=(l-1)n_0+1}\log \| T_{\omega,j}\|.
	\end{eqnarray*}
	Piecing these bounds together, we obtain
	\begin{eqnarray*}
		%\esp\left[
		\left| \log \|T_{\omega,ln_0 }\cdots T_{\omega,(l-1)n_0+1}\varphi_0\| \right| 
		%\right]
		\leq 
		\sum^{ln_0}_{j=(l-1)n_0+1}\log \| T_{\omega,j}\|.
	\end{eqnarray*}
	Remembering \eqref{eq:bound-T}, we get
	\begin{eqnarray*}
		\esp\left[
			\exp\left\{
				2s \left| \log \|T_{\omega,ln_0 }\cdots T_{\omega,(l-1)n_0+1}\varphi_0\| \right|
			\right\}
		\right]
		\leq
		\prod^{ln_0}_{j=(l-1)n_0+1} \esp[\| T_{\omega,l}\|^{2s}]
		\leq
		\e^{c_3 n_0},
	\end{eqnarray*}
	for some $c_3>0$ for all $s\in(0,\frac12)$.
	Hence, 
	\begin{eqnarray*}
		\esp\left[ \|T_{\omega,ln_0 }\cdots T_{\omega,(l-1)n_0+1}\varphi_0\|^{-s} \right]
		\leq
		1 - c_1 s \frac{n_0^{1-2\alpha}}{l^{2\alpha}} 
		+
		c_2 s^2 \e^{c_3 n_0} \frac{n_0^{\frac12-\alpha}}{l^{\alpha}},
	\end{eqnarray*}
	for all $l\ge 1$, $\|\varphi_0\|=1$ and $E\in I$. We can now find  $s_0\in(0,\frac12)$ such that
	\begin{eqnarray*}
		\esp\left[ \|T_{\omega,ln_0 }\cdots T_{\omega,(l-1)n_0+1}\varphi_0\|^{-s} \right]
		\leq
		1 - \frac{c_4}{l^{2\alpha}},
	\end{eqnarray*}
	for some $c_4>0$, for all $s\in(0,s_0)$, $l\geq 1$, $\norm{\varphi_0}=1$ and $E\in I$. 
\end{proof}
%%%%%%%%%%%%%%%%%%%%%%%%%%%%%%%%%%%%%%%%%%%%%%%%%%%%%%%%%%%%
%%%%%%%%%%%%%%%%%%%%%%%%%%%%%%%%%%%%%%%%%%%%%%%%%%%%%%%%%%%%

We can now proceed with the proof of Lemma \ref{thm:key-lemma}.

%%%%%%%%%%%%%%%%%%%%%%%%%%%%%%%%%%%%%%%%%%%%%%%%%%%%%%%%%%%%
%%%%%%%%%%%%%%%%%%%%%%%%%%%%%%%%%%%%%%%%%%%%%%%%%%%%%%%%%%%%
\begin{proof}[Proof of Lemma \ref{thm:key-lemma}]
	Once again, we remove $E$ from the notation to lighten the presentation.
	Let $u\geq 1$, $n\geq u + n_0$ and $s\in(0,s_0]$ where $n_0=n_0(I)$ and $s_0=s_0(I)$ are taken from Lemma \ref{thm:first-bound-Tmn}.
	Write $u= l_1n_0 - r_1$ and $n=l_2n_0 + r_2$ with $0\leq r_1,\, r_2 < n_0$.
	Then,
	\begin{eqnarray*}
		\| T_{\omega,l_2n_0} \cdots T_{\omega,u} \varphi_0 \|
		&=&
		\| T_{\omega,l_2n_0+1}^{-1} T_{\omega,n}^{-1}T_{\omega,n}\cdots T_{\omega,u} \varphi_0 \|\\
		&\leq&
		\prod^n_{j=l_2n_0+1}\| T_{\omega,j} \| 
		\cdot
		\| T_{\omega,n} \cdots T_{\omega,u} \varphi_0\|.
	\end{eqnarray*}
	Hence,
	\begin{eqnarray*}
		\esp[\| T_{\omega,n} \cdots T_{\omega,u}\varphi_0\|^{-s}]
		&\leq&
		\prod^n_{j=l_2n_0+1}\esp[ \| T_{\omega,j} \|^s] 
		\cdot
		\esp[\| T_{\omega,l_2n_0} \cdots T_{\omega,u} \varphi_0 \|^{-s}]
		\\
		&\leq&
		C_1
		\esp[\| T_{\omega,l_2n_0} \cdots T_{\omega,u} \varphi_0 \|^{-s}],
	\end{eqnarray*}
	for some $C_1=C_1(I)\in (0,\infty)$ by \eqref{eq:bound-T}, (recall that $n_0$ is fixed). 
	The rest of the proof is  based on a careful conditioning that we now detail.	
	Let
	\begin{eqnarray*}
		\varphi_{l} = \frac{T_{ln_0} \cdots T_{u} \varphi_0}{\| T_{ln_0} \cdots T_{u} \varphi_0 \|},
	\end{eqnarray*}
	ans observe that $\varphi_{l-1}$ is measurable with respect to $\mathcal{F}_{l-1}$.
	Hence, Lemma \ref{thm:first-bound-Tmn} can be applied to obtain
	\begin{eqnarray*}
		\esp
		\left[
		\| T_{\omega,ln_0} \cdots T_{\omega,(l-1)n_0 +1} \varphi_{l-1} \|^{-s}
		\Big{|}
		\mathcal{F}_{l-1}
		\right]
		\leq
		1-\frac{c_4}{l^{2\alpha}},
	\end{eqnarray*}
	where $c_4=c_4(I)>0$ is the constant from Lemma \ref{thm:first-bound-Tmn}. Hence,
	\begin{eqnarray*}
		&&
		\esp\left[
			\| T_{\omega,n} \cdots T_{\omega,u} \varphi_0\|^{-s}
		\right]
		\leq
		C_1
		\esp\left[
			\| T_{\omega,l_2n_0} \cdots T_{\omega,u} \varphi_0\|^{-s}
		\right]
		\\
		&&
		=
		C_1
		\esp\left[
			\| T_{\omega,(l_2-1)n_0} \cdots T_{\omega,u} \varphi_0 \|^{-s}
			\| T_{\omega,l_2n_0} \cdots T_{\omega,(l_2-1)n_0+1} \varphi_{l_2-1}\|^{-s} 
		\right]
		\\
		&&
		=
		C_1
		\esp\left[ \esp\left[
			\| T_{\omega,(l_2-1)n_0} \cdots T_{\omega,u} \varphi_0 \|^{-s}
			\| T_{\omega,l_2n_0} \cdots T_{(l_2-1)n_0+1} \varphi_{l_2-1}\|^{-s} 
		\Big{|}
		\mathcal{F}_{l_2-1}
		\right]\right]
		\\
		&&
		=
		C_1
		\esp\left[ [
			\| T_{\omega,(l_2-1)n_0} \cdots T_{\omega,u} \varphi_0 \|^{-s}
			 \esp\left[
			\| T_{\omega,l_2n_0} \cdots T_{\omega,(l_2-1)n_0+1} \varphi_{l_2-1}\|^{-s} 
		\Big{|}
		\mathcal{F}_{l_2-1}
		\right]\right]
		\\
		&&
		\leq
		C_1 \left( 1-\frac{c_4}{l_2^{2\alpha}}\right)
		\esp\left[
		\| T_{\omega,(l_2-1)n_0} \cdots T_{\omega,u} \varphi_0 \|^{-s}
		\right].
	\end{eqnarray*}
	Iterating,
	\begin{eqnarray*}
		\esp\left[
			\| T_{\omega,n} \cdots T_{\omega,u} \varphi_0\|^{-s}
		\right]
		&\leq&
		C_1
		\prod^{l_2}_{j=l_1} \left( 1-\frac{c_4}{j^{2\alpha}}\right)
		\esp\left[
		\| T_{\omega,m+r_1} \cdots T_{\omega,u} \varphi_0\|^{-s}
		\right].
	\end{eqnarray*}
	Just as we did in the previous lemma,
	\begin{eqnarray*}
		1 = \|T_{\omega,u}^{-1} \cdots T_{\omega,u+r_1}^{-1} T_{\omega,u+r_1} \cdots T_{\omega,u}\varphi_0\|
		\leq 
		\prod^{u+r_1}_{j=u} \| T_{\omega,j} \| \cdot  \|T_{\omega,u+r_1} \cdots T_{\omega,u}\varphi_0\|,
	\end{eqnarray*}
	so that, by \eqref{eq:bound-T},
	\begin{eqnarray*}
		\esp\left[
			\|T_{\omega,u+r_1} \cdots T_{\omega,u}\varphi_0\|^{-s}
		\right]
		\leq C_2,
	\end{eqnarray*}
	for some $C_2=C_2(I)>0$.
	Hence,
	\begin{eqnarray*}
		\esp\left[
			\| T_{\omega,n} \cdots T_{\omega,u}\varphi_0\|^{-s}
		\right]
		\leq
		C
		\prod^{l_2}_{j=l_1} \left( 1-\frac{c_4}{j^{2\alpha}}\right)
		\leq
		C
		\e^{cu^{1-2\alpha}}\e^{-cn^{1-2\alpha}},
	\end{eqnarray*}
	for some suitable $C=C(I)\in(0,\infty)$ and $c=c(I)>0$.

\end{proof}
%%%%%%%%%%%%%%%%%%%%%%%%%%%%%%%%%%%%%%%%%%%%%%%%%%%%%%%%%%%%
%%%%%%%%%%%%%%%%%%%%%%%%%%%%%%%%%%%%%%%%%%%%%%%%%%%%%%%%%%%%

\begin{proof}[Proof of Theorem \ref{thm:FM-Green}]
	Let $u\geq 1$ and $n\geq u + n_0$ where $n_0=n_0(I)\geq 1$ is taken from Lemma \ref{thm:first-bound-Tmn}. Lemma \ref{thm:key-lemma} can be combined with the bound \eqref{eq:resolvent-to-transfer-matrix} to finish the proof of the first estimate of Theorem \ref{thm:FM-Green}. The second estimate can be proved using \eqref{eq:resolvent-to-transfer-matrix-2} and \eqref{eq:resolvent-to-transfer-matrix-3} instead of \eqref{eq:resolvent-to-transfer-matrix} and the analogue of Lemma \ref{thm:key-lemma} for the transfer matrices in the second system of coordinates.
	\newline
	If $1\leq n \leq u+n_0$, we just use the a priori estimate appearing in Lemma \ref{thm:FM-a-priori}.
\end{proof}

%%%%%%%%%%%%%%%%%%%%%%%%%%%%%%%%%%%%%%%%%%%%%%%%%%%%%%%%%%%%
%%%%%%%%%%%%%%%%%%%%%%%%%%%%%%%%%%%%%%%%%%%%%%%%%%%%%%%%%%%%
%%%%%%%%%%%%%%%%%%%%%%%%%%%%%%%%%%%%%%%%%%%%%%%%%%%%%%%%%%%%
%%%%%%%%%%%%%%%%%%%%%%%%%%%%%%%%%%%%%%%%%%%%%%%%%%%%%%%%%%%%

\subsection{Correlators and dynamical localization}\label{sec:correlators}

We follow the approach of \cite[Chapter 7]{AW}.
%%%%%%%%%%%%%%%%%%%%%%%%%%%%%%%%%%%%%%%%%%%%%%%%%%%%%%%%%%%%
%%%%%%%%%%%%%%%%%%%%%%%%%%%%%%%%%%%%%%%%%%%%%%%%%%%%%%%%%%%%
Recall the definition of the correlator \eqref{correlator},
\begin{equation*} 
	Q_\omega(u,\sigma;n,\sigma';I)=\sup_{\substack{f\in C_0(I)\\ \norm{f}_\infty\le1}}
	 \left| \langle \delta_u^{\sigma} , P_I(D_{\omega})f(\D) \delta_n^{\sigma'} \rangle \right|,
\end{equation*}
We will need to work with finite volume correlators
in boxes $\Lambda_L$,
\begin{equation*} 
	Q_{\omega,L}(u,\sigma;n,\sigma';I)=\sup_{\substack{f\in C_0(I)\\ \norm{f}_\infty\le1}}
	 \left| \langle \delta_u^{\sigma} , P_I(D_{\omega,L})f(D_{\omega,L}) \delta_n^{\sigma'} \rangle \right|.
\end{equation*}
Here, we allowed ourselves a slight abuse of notation where the operator $D_{\omega,L}$ coincides with the definition given at the beginning of Section \ref{sec:FME} if $\sigma'=-$ but will be understood to be $D'_{\omega,L}$ if $\sigma'=+$. The argument is of course identical in both cases. 
Since $D_{\omega,L}\to D_{\omega}$ as $L\to\infty$ in the strong resolvent sense, $\p$-a.s,  we have
\begin{eqnarray*}
	Q_{\omega}(u,\sigma;n,\sigma';I) \leq \liminf_L Q_{\omega,L}(u,\sigma;n,\sigma';I).
\end{eqnarray*}
By Fatou's lemma, we get
\begin{eqnarray}
	\label{eq:sum-Q-1}
	\esp\left[ Q_{\omega}(u,\sigma;n,\sigma';I)^2\right]
	&\leq&\liminf_L \esp\left[  Q_{\omega,L}(u,\sigma;n,\sigma';I)^2\right],
\end{eqnarray}
so that it is enough to bound the correlators in boxes, uniformly in the size of the box.
At this point, it is convenient to work with the interpolated eigenfunction correlator or $s$-correlator defined as
\begin{eqnarray}\label{s correlator}
	Q_{\omega,L}(u,\sigma;n,\sigma';I,s)
	 = 
	 \sum_{E\in I \cap \sigma(D_{\omega,L})} 
	 \left| \langle \delta_u,\, P_{{E}}(D_{\omega,L}) \delta_u \rangle\right|^{1-s}
	 \left| \langle \delta_u,\, P_{{E}}(D_{\omega,L}) \delta_n \rangle\right|^{s},
\end{eqnarray}
for $s\in[0,1]$,
where $ P_{{E}}(D_{\omega,L})$ is the projection of $D_{\omega,L}$ on the eigenspace corresponding to $E$. By Cauchy-Schwarz inequality for the kernel of $P_{{E}}(D_{\omega,L})$, we obtain
\begin{eqnarray*}
	Q_{\omega,L}(u,\sigma;n,\sigma';I)^2 \leq Q_{\omega,L}(u,\sigma;n,\sigma';I,s)\, Q_{\omega,L}(n,\sigma';u,\sigma;I,s), \quad \text{for all}  \quad s\in[0,1].
\end{eqnarray*}
The bound \eqref{eq:sum-Q-1} becomes
\begin{eqnarray}\label{eq:correlator-to-s-correlator}
	\esp\left[ Q_{\omega}(u,\sigma;n,\sigma';I)^2\right]
	&\leq&
	\liminf_L \esp\left[  Q_{\omega,L}(u,\sigma;n,\sigma';I,s)\, Q_{\omega,L}(n,\sigma';u,\sigma;I,s)\right]\\
	&\leq&
	\liminf_L  \esp\left[  Q_{\omega,L}(u,\sigma;n,\sigma';I,s)\right],
\end{eqnarray}
since $Q_{\omega,L}(n,\sigma';u,\sigma;I)\leq 1$. In Lemma \ref{thm:correlator-FM} below, we will show that the expected value of the $s$-correlators \eqref{s correlator} can be estimated in terms of the fractional moments of the Green's function. Together with with Theorem \ref{thm:FM-Green}, this will finish the proof of Theorem \ref{thm:DL}. We defer the details to the end of the section.

%%%%%%%%%%%%%%%%%%%%%%%%%%%%%%%%%%%%%%%%%%%%%%%%%%%%%%%%%%%%
%%%%%%%%%%%%%%%%%%%%%%%%%%%%%%%%%%%%%%%%%%%%%%%%%%%%%%%%%%%%
%%%%%%%%%%%%%%%%%%%%%%%%%%%%%%%%%%%%%%%%%%%%%%%%%%%%%%%%%%%%
%%%%%%%%%%%%%%%%%%%%%%%%%%%%%%%%%%%%%%%%%%%%%%%%%%%%%%%%%%%%

The last step consists then in controling the $s$-correlator \eqref{s correlator} by the resolvent.
Let $D^v_{L,u,\sigma} = D_{\omega,L} +(v-V_{\omega}(u,\sigma)) {\bf 1}_{u,\sigma}$ be the operator resulting from setting the disorder at site $(u,\sigma)$ to the value $v$, where we simply denoted
\begin{eqnarray*}
	V_{\omega}(u,\sigma)
	=
	\left\{
		\begin{array}{ll}
			\vo(u), & \text{ if} \, \sigma=+,
			\\
			\vt(u), & \text{ if} \, \sigma=-.
		\end{array}
	\right.
\end{eqnarray*}
We borrow the following identity from \cite[Lemma 7.10]{AW}:
\begin{eqnarray}\label{eq:identity-correlator}
	Q_{\omega,L}(u,\sigma;n,\sigma';I,s)
	 =
	 \int _I
	 \left| \langle \delta_u^{\sigma},\, (D^v_{L,u,\sigma}-E)^{-1} \delta_n^{\sigma'} \rangle\right|^{s}
	 \left| V_{\omega}(u,\sigma)-v\right|^{s}
	 \mu_{\delta_{u^\sigma}}(dE),
\end{eqnarray}
where $\mu_{\delta_u,\sigma}$ is the spectral measure of $D_{\omega,L}$ on $\delta_u^{\sigma}$. 
\newline
We also recall the spectral averaging principle
\begin{eqnarray*}
	\int_{\R}
	\int _I
	 \left| \langle \delta_u^{\sigma},\, (D^v_{L,u,\sigma}-E)^{-1} \delta_n^{\sigma'} \rangle\right|^{s}
	 \mu_{\delta_u^\sigma}(dE)
	 \, dv
	=
	 \int _I
	 \left| \langle \delta_u^{\sigma},\, (D^v_{L,u,\sigma}-E)^{-1} \delta_n^{\sigma'} \rangle\right|^{s}
	 dE,
\end{eqnarray*}
where the dependence in $v$ in the integrand on the left-hand-side is hidden in $\mu_{\delta_u,\sigma}$.
\bigskip

The following is a straightforward adaptation of \cite[Theorem 7.11]{AW} to the inhomogeneous potential case.
%%%%%%%%%%%%%%%%%%%%%%%%%%%%%%%%%%%%%%%%%%%%%%%%%%%%%%%%%%%%
%%%%%%%%%%%%%%%%%%%%%%%%%%%%%%%%%%%%%%%%%%%%%%%%%%%%%%%%%%%%
\begin{lemma}\label{thm:correlator-FM}
	Assume \textbf{(A1)}, \textbf{(A5)} and \textbf{(A6)}, and let $q=\tfrac{p}{p-1}$.
	For all $s\in[0,\frac{p-1}{p})$, there exists constants $C=C(s)\in(0,\infty)$ and $\kappa=\kappa(s)\geq 0$ such that
	\begin{eqnarray}\label{eq:correlator-FM}
		&&
		\esp\left[Q_{\omega,L}(u,\sigma;n,\sigma';I,s)\right]
%		\\
%		&&
%		\phantom{blablablabla}
		\leq 
		C (\lambda a_u)^{-\kappa}
		\left(
			 \int _I
		 		\esp\left[
	 				\left| G_{\omega,L}(u,\sigma;n,\sigma';E)\right|^{qs}
		 		\right]
		 	dE
		 \right)^{\frac{1}{q}},
	\end{eqnarray}
	for all $L\geq 1$, $u,n\in[1,L]$, $\sigma,\sigma' \in \{+,-\}$ and all interval $I\subset \R$.
\end{lemma}
%%%%%%%%%%%%%%%%%%%%%%%%%%%%%%%%%%%%%%%%%%%%%%%%%%%%%%%%%%%%
%%%%%%%%%%%%%%%%%%%%%%%%%%%%%%%%%%%%%%%%%%%%%%%%%%%%%%%%%%%%
\begin{proof}
	Let $\esp_{u,\sigma}$ denote the expected value with respect to the random variable $V_{\omega}(u,\sigma)$ and let $\rho_{u,\sigma}$ be the corresponding density.  Let $q=\frac{p}{p-1}$. Using \eqref{eq:identity-correlator} and H\"older's inequality,
	\begin{align*}
		\esp_{u,\sigma}&\left[Q_{\omega,L}(u,\sigma;n,\sigma';I,s)\right]
=
		\esp_u\left[
		\int _I
	 \left| \langle \delta_u^{\sigma},\, (D^v_{L,u,\sigma}-E)^{-1} \delta_n^{\sigma'} \rangle\right|^{s}
	 \left| V_\omega(u,\sigma)-v\right|^{s}
	 \mu_{\delta_{u,\sigma}}(dE)
	 \right]
	 \\
		&=
		\int_{\R}
		\int _I
	 \left| \langle \delta_u^{\sigma},\, (D^v_{L,u,\sigma}-E)^{-1} \delta_n^{\sigma'} \rangle\right|^{s}
	 |\tau-v|^s
	 \mu_{\delta_{u,\sigma}}(dE)
	 \rho_{u,\sigma}(\tau)\, d\tau\\
	 &\le
	 \left(
	 	\int_{\R} |\tau-v|^{ps} \rho_{u,\sigma}(\tau)^p d\tau
	 \right)^{\frac{1}{p}}
	 \left(
	 	\int_{\R}
		\int _I
	 	\left| \langle \delta_u^{\sigma},\, (D^v_{L,u,\sigma}-E)^{-1} \delta_n^{\sigma'} \rangle\right|^{qs}
	 	\mu_{\delta_{u,\sigma}}(dE)
	 	\, d\tau
	 \right)^{\frac{1}{q}}
	 \\
	 &\le 
	 \left(
	 	\int_{\R} |\tau-v|^{ps} \rho_{u,\sigma}(\tau)^p d\tau
	 \right)^{\frac{1}{p}}
	 \left(
		\int _I
	 	\left| \langle \delta_u^{\sigma},\, (D^v_{L,u,\sigma}-E)^{-1} \delta_n^{\sigma'} \rangle\right|^{qs}
	 	dE
	 \right)^{\frac{1}{q}},
	\end{align*}
	by the spectral averaging principle. Now, observe that
	\begin{eqnarray*}
		\int_{\R} |\tau-v|^{ps} \rho_{u,\sigma}(\tau)^p d\tau
		&\leq&
		(1+|v|)^{ps}
		\int_{\R} (1+|\tau|)^{ps} \rho_{u,\sigma}(\tau)^p d\tau
		\\
		&\leq&
		C (\lambda a_u)^{-\gamma} (1+|v|)^{ps},
	\end{eqnarray*}
	for some $C\in(0,\infty)$ and $\gamma \geq 0$ thanks to \textbf{(A6a)}.
	So far,
	\begin{eqnarray*}
		&&
		\esp_{u,\sigma}\left[Q_L(u,\sigma;n,\sigma';I,s)\right]^q
	 	\leq
	 	C (\lambda a_u)^{-\frac{\gamma q}{p}} (1+|v|)^{qs}
			\int _I
	 		\left| \langle \delta_u^{\sigma},\, (D^v_{L,u,\sigma}-E)^{-1} \delta_n^{\sigma'} \rangle\right|^{qs}
	 		dE,
	\end{eqnarray*}
	for all $v\in\R$.
	This can be integrated against $c_{u,\sigma}(1+|v|)^{-qs} \rho_{u,\sigma}(v)\, dv$ where
	\begin{eqnarray*}
		c_{u,\sigma}^{-1} 
		= 
		\int_{\R}(1+|v|)^{-qs} \rho_{u,\sigma}(v)\, dv
		\geq
		c (\lambda a_u)^{\gamma},
	\end{eqnarray*}
	for some $c>0$ thanks to \textbf{(A6b)},
	to obtain
	\begin{eqnarray*}
		&&
		\esp_{u,\sigma}\left[Q_{\omega,L}(u,\sigma;n,\sigma';I,s)\right]^q
	 	\leq
	 	C (\lambda a_u)^{-\kappa}
			 \int _I
		 	\esp_{u,\sigma}\left[
	 		\left| \langle \delta_u^{\sigma},\, (D^v_{L,u,\sigma}-E)^{-1} \delta_n^{\sigma'} \rangle\right|^{qs}
		 \right]
		 dE,
	\end{eqnarray*}
	for some $C\in(0,\infty)$ and $\kappa = \kappa(s,\gamma)>0$. Finally, denoting by $\widehat{\esp}_{u,\sigma}$ the average with respect to the remaining disorder variables, we have
	\begin{eqnarray*}
		&&
		\esp\left[Q_{\omega,L}(u,\sigma;n,\sigma';I,s)\right]^q
		=
		\widehat{\esp}_{u,\sigma}\left[
			\esp_{u,\sigma}\left[
				Q_{\omega,L}(u,\sigma;n,\sigma';I,s)
			\right]
		\right]^q
		\\
		&&
		\phantom{blablablabla}
		\leq
		\widehat{\esp}_{u,\sigma}\left[
			\esp_{u,\sigma}\left[
				Q_{\omega,L}(u,\sigma;n,\sigma';I,s)
			\right]^q
		\right]
		\\
		&&
		\phantom{blablablabla}
		\leq 
		C (\lambda a_u)^{-\kappa}
			 \int _I
		 	\esp\left[
	 		\left| \langle \delta_u^{\sigma},\, (D^v_{L,u,\sigma}-E)^{-1} \delta_n^{\sigma'} \rangle\right|^{qs}
		 \right]
		 dE.
	\end{eqnarray*}
	This finishes the proof.
\end{proof}
%%%%%%%%%%%%%%%%%%%%%%%%%%%%%%%%%%%%%%%%%%%%%%%%%%%%%%%%%%%%
%%%%%%%%%%%%%%%%%%%%%%%%%%%%%%%%%%%%%%%%%%%%%%%%%%%%%%%%%%%%
We complete the proof of Theorem \ref{thm:DL}.
%%%%%%%%%%%%%%%%%%%%%%%%%%%%%%%%%%%%%%%%%%%%%%%%%%%%%%%%%%%%
%%%%%%%%%%%%%%%%%%%%%%%%%%%%%%%%%%%%%%%%%%%%%%%%%%%%%%%%%%%%
\begin{proof}[Proof of Theorem \ref{thm:DL}]
	Fix $u\in\n^*$.
	Combining the bound \eqref{eq:correlator-to-s-correlator} together with Lemma \ref{thm:correlator-FM}, we obtain
	\begin{eqnarray*}
		\esp[Q_\omega(u,\pm;n,-;I)^2]
		\leq
		Ca_u^{-1} 
		\liminf
		\int_I
		\esp\left[
			\left| G_{\omega,L}(u,\pm;n,-;E)\right|^s
		\right]\, dE.
	\end{eqnarray*}
	We can use Theorem \ref{thm:FM-Green} with $s$ small enough to bound this last quantity uniformly in the size of the box and over the interval $I$. This gives the bound
	\begin{eqnarray*}
		\esp[Q_\omega(u,\pm;n,-;I)^2]
		\leq
		C \lambda^{-2s} a_n^{-2s} \e^{-c n^{1-2\alpha}},
	\end{eqnarray*}
	for all $n$, for some constants $C=C(u,s)\in(0,\infty)$ and $c=c(s,I)>0$. All the other entries $(u,\sigma)$ and $(n,\sigma')$ can be handled similarly using
	completely analogous estimates for the corresponding resolvents.
	This proves \eqref{eq:FM-bound}. Finally,
	\begin{eqnarray*}
		\sum_{n,\sigma'} \esp[Q_\omega(u,\sigma;n,\sigma';I)^2] 
		\leq 
		C\lambda^{-2s} \sum_n a_n^{-2s} \e^{-c n^{1-2\alpha}}<\infty,
	\end{eqnarray*}
	with $C$ as above, which shows dynamical localization.
\end{proof}

%%%%%%%%%%%%%%%%%%%%%%%%%%%%%%%%%%%%%%%%%%%%%%%%%%%%%%%%%%%%
%%%%%%%%%%%%%%%%%%%%%%%%%%%%%%%%%%%%%%%%%%%%%%%%%%%%%%%%%%%%
%%%%%%%%%%%%%%%%%%%%%%%%%%%%%%%%%%%%%%%%%%%%%%%%%%%%%%%%%%%%
%%%%%%%%%%%%%%%%%%%%%%%%%%%%%%%%%%%%%%%%%%%%%%%%%%%%%%%%%%%%

 \subsection{Proof of Proposition \ref{thm:consequences-DL} and \ref{thm:SULE} and Theorem \ref{thm:lower-bound-DL}}\label{sec:consequences-DL}

%%%%%%%%%%%%%%%%%%%%%%%%%%%%%%%%%%%%%%%%%%%%%%%%%%%%%%%%%%%%
%%%%%%%%%%%%%%%%%%%%%%%%%%%%%%%%%%%%%%%%%%%%%%%%%%%%%%%%%%%%

We prove Proposition \ref{thm:consequences-DL}:
%%%%%%%%%%%%%%%%%%%%%%%%%%%%%%%%%%%%%%%%%%%%%%%%%%%%%%%%%%%%
%%%%%%%%%%%%%%%%%%%%%%%%%%%%%%%%%%%%%%%%%%%%%%%%%%%%%%%%%%%%
\begin{proof}[Proof of Proposition \ref{thm:consequences-DL}]
	Suppose that \eqref{eq:DL} holds in an energy interval $I$. We will prove that the spectrum of $\D$ is almost surely pure point in $I$. This is a consequence of the RAGE Theorem \cite{CFKS}. Let $\chi_R$ the characteric function of the box $[0,R]$. Since $\chi_R$ converges strongly to the identity as $R\to\infty$, it is enough to show that, $\p$-almost surely,
	\begin{eqnarray}\label{eq:RAGE}
		\lim_{R\to\infty} \sup_t \left\| \left( 1 - \chi_R \right) \e^{-it\D}P_I(\D) \delta_u^{\sigma} \right\|^2 =0,
	\end{eqnarray}
	for all canonical vectors $\delta_u^{\sigma}$ since this implies that the range of $P_I(D_{\omega,\lambda})$ is almost surely included in the point spectrum of $\D$. Now,
	\begin{eqnarray*}
		\left\| \left( 1 - \chi_R \right) \e^{-it\D}P_I(\D) \delta_u^{\sigma} \right\|^2 
		&=&
		\sum_{\substack{|n|>R \\ \sigma'=\pm}} \left| \langle \delta_n^{\sigma'}, \e^{-it\D}P_I(\D) \delta_u^{\sigma} \rangle\right|^2\\
		&=&
		\sum_{\substack{|n|>R \\ \sigma'=\pm}} \left| \langle P_I(\D) \ \e^{-it\D} \delta_n^{\sigma'}, \delta_u^{\sigma} \rangle\right|^2\\
		&\leq&
		\sum_{|u|>R} Q_\omega(u,\sigma;n, \sigma';I)^2,
	\end{eqnarray*}
	for all $t\in\R$. By Fatou's lemma,
	\begin{eqnarray*}
		\esp\left[  \lim_{R\to\infty} \sup_t \left\| \left( 1 - \chi_R \right) \e^{-it\D}P_I(\D) \delta_u^{\sigma} \right\|^2 \right]
		&\leq&
		\esp\left[  \liminf_{R\to\infty} 
		\sum_{\substack{|n|>R \\ \sigma'=\pm}} Q_\omega(u,\sigma;n, \sigma';I)^2
		\right]\\
		&\le&
		\liminf_{R\to\infty} \esp\left[   
		\sum_{\substack{|n|>R \\ \sigma'=\pm}} Q_\omega(u,\sigma;n, \sigma';I)^2
		\right]
		=0,
	\end{eqnarray*}
	which shows that \eqref{eq:RAGE} holds $\p$-almost surely for each $u$ and $\sigma$. Hence, for each $u$ and $\sigma$, there exists $\Omega_{u,\sigma} \subset \Omega$ with $\p(\Omega_{u,\sigma})=1$ such that \eqref{eq:RAGE} holds for all $\omega\in\Omega_{u,\sigma}$. Finally, the set $\displaystyle\widetilde \Omega := \bigcap_{u,\sigma} \Omega_{u,\sigma}$ is so that $\p(\widetilde \Omega)=1$ and \eqref{eq:RAGE} holds for all $u$ and $\sigma$ simultaneously for all $\omega \in \widetilde \Omega$.
%%%%%%%%%%%%%%%%%%%%%%%%%%%%%%%%%%%%%%%%%%%%%%%%%%%%%%%%%%%%
%%%%%%%%%%%%%%%%%%%%%%%%%%%%%%%%%%%%%%%%%%%%%%%%%%%%%%%%%%%%

	We turn to the moments bounds. Let $p>0$,
	\begin{eqnarray*}
		&&\esp\left[\sup_t
					\left\|
						|\X|^{\frac{p}{2}}\e^{-it\D} P_I(\D)\delta_u^{\sigma}
					\right\|^2
				\right]\\
		&& \phantom{blablablab} =
		\esp\left[
					\langle
						|{\bf X}|^{p}\e^{-it\D}P_I(\D)\delta_u^{\sigma},
						e^{-itD_{\omega,\lambda}}\delta_u
					\rangle
				\right]\\
		&& \phantom{blablablab} =
		\esp\left[\sup_t
					\sum_{n,\sigma'}
					\langle
						|{\bf X}|^{p}\e^{-it\D}P_I(\D)\delta_u^{\sigma}, \delta_n^{\sigma'}
					\rangle
					\langle
						\delta_n^{\sigma'},
						\e^{-it\D}P_I(\D)\delta_u^{\sigma}
					\rangle
				\right]\\
		&& \phantom{blablablab} \leq
		\esp\left[\sup_t
					\sum_{n,\sigma'}
					|n|^p
					\left|
					\langle
						\delta_n^{\sigma'},
						\e^{-it\D}P_I(\D)\delta_u^{\sigma}
					\rangle
					\right|^2
				\right]\\
		&& \phantom{blablablab} \leq
		\esp\left[
					\sum_{n,\sigma'}
					|n|^p
					Q_\omega(u,\sigma;n, \sigma';I)^2
				\right],
	\end{eqnarray*}
	which is bounded for each $u$ and $\sigma$ in virtue of \eqref{eq:FM-bound}.
\end{proof}
%%%%%%%%%%%%%%%%%%%%%%%%%%%%%%%%%%%%%%%%%%%%%%%%%%%%%%%%%%%%
%%%%%%%%%%%%%%%%%%%%%%%%%%%%%%%%%%%%%%%%%%%%%%%%%%%%%%%%%%%%
Finally, we provide sketches of proof for Proposition \ref{thm:SULE} and Theorem \ref{thm:lower-bound-DL} as the arguments are either standard or have been developed elsewhere in this work.
%%%%%%%%%%%%%%%%%%%%%%%%%%%%%%%%%%%%%%%%%%%%%%%%%%%%%%%%%%%%
%%%%%%%%%%%%%%%%%%%%%%%%%%%%%%%%%%%%%%%%%%%%%%%%%%%%%%%%%%%%
\begin{proof}[Proof of Proposition \ref{thm:SULE}]
	The upper bound is standard and follows from the estimate in Theorem \ref{thm:DL} (see \cite[Theorem 9.22]{CFKS}). The lower bound can be obtained as in Lemma \ref{thm:lower-bound-eigenfunctions} above.
\end{proof}
%%%%%%%%%%%%%%%%%%%%%%%%%%%%%%%%%%%%%%%%%%%%%%%%%%%%%%%%%%%%
%%%%%%%%%%%%%%%%%%%%%%%%%%%%%%%%%%%%%%%%%%%%%%%%%%%%%%%%%%%%
\begin{proof}[Proof of Proposition \ref{thm:lower-bound-DL}]
	The first statement is completely analogous to the bound on the moments in Proposition \ref{thm:consequences-DL}. The second statement can be obtained following the strategy of proof of Theorem \ref{thm:transport} using the lower bound in Proposition \ref{thm:SULE} instead of Lemma \ref{thm:lower-bound-eigenfunctions}.
\end{proof}
%%%%%%%%%%%%%%%%%%%%%%%%%%%%%%%%%%%%%%%%%%%%%%%%%%%%%%%%%%%%
%%%%%%%%%%%%%%%%%%%%%%%%%%%%%%%%%%%%%%%%%%%%%%%%%%%%%%%%%%%%

%%%%%%%%%%%%%%%%%%%%%%%%%%%%%%%%%%%%%%%%%%%%%%%%%%%%%%%%%%%%
%%%%%%%%%%%%%%%%%%%%%%%%%%%%%%%%%%%%%%%%%%%%%%%%%%%%%%%%%%%%
%%%%%%%%%%%%%%%%%%%%%%%%%%%%%%%%%%%%%%%%%%%%%%%%%%%%%%%%%%%%
%%%%%%%%%%%%%%%%%%%%%%%%%%%%%%%%%%%%%%%%%%%%%%%%%%%%%%%%%%%%
%%%%%%%%%%%%%%%%%%%%%%%%%%%%%%%%%%%%%%%%%%%%%%%%%%%%%%%%%%%%
%%%%%%%%%%%%%%%%%%%%%%%%%%%%%%%%%%%%%%%%%%%%%%%%%%%%%%%%%%%%
%%%%%%%%%%%%%%%%%%%%%%%%%%%%%%%%%%%%%%%%%%%%%%%%%%%%%%%%%%%%
%%%%%%%%%%%%%%%%%%%%%%%%%%%%%%%%%%%%%%%%%%%%%%%%%%%%%%%%%%%%

\appendix

%%%%%%%%%%%%%%%%%%%%%%%%%%%%%%%%%%%%%%%%%%%%%%%%%%%%%%%%%%%%
%%%%%%%%%%%%%%%%%%%%%%%%%%%%%%%%%%%%%%%%%%%%%%%%%%%%%%%%%%%%
%%%%%%%%%%%%%%%%%%%%%%%%%%%%%%%%%%%%%%%%%%%%%%%%%%%%%%%%%%%%
%%%%%%%%%%%%%%%%%%%%%%%%%%%%%%%%%%%%%%%%%%%%%%%%%%%%%%%%%%%%
%%%%%%%%%%%%%%%%%%%%%%%%%%%%%%%%%%%%%%%%%%%%%%%%%%%%%%%%%%%%
%%%%%%%%%%%%%%%%%%%%%%%%%%%%%%%%%%%%%%%%%%%%%%%%%%%%%%%%%%%%
%%%%%%%%%%%%%%%%%%%%%%%%%%%%%%%%%%%%%%%%%%%%%%%%%%%%%%%%%%%%
%%%%%%%%%%%%%%%%%%%%%%%%%%%%%%%%%%%%%%%%%%%%%%%%%%%%%%%%%%%%

\section{Some technical estimates}

%%%%%%%%%%%%%%%%%%%%%%%%%%%%%%%%%%%%%%%%%%%%%%%%%%%%%%%%%%%%
%%%%%%%%%%%%%%%%%%%%%%%%%%%%%%%%%%%%%%%%%%%%%%%%%%%%%%%%%%%%
%%%%%%%%%%%%%%%%%%%%%%%%%%%%%%%%%%%%%%%%%%%%%%%%%%%%%%%%%%%%
%%%%%%%%%%%%%%%%%%%%%%%%%%%%%%%%%%%%%%%%%%%%%%%%%%%%%%%%%%%%

%%%%%%%%%%%%%%%%%%%%%%%%%%%%%%%%%%%%%%%%%%%%%%%%%%%%%%%%%%%%
%%%%%%%%%%%%%%%%%%%%%%%%%%%%%%%%%%%%%%%%%%%%%%%%%%%%%%%%%%%%
%%%%%%%%%%%%%%%%%%%%%%%%%%%%%%%%%%%%%%%%%%%%%%%%%%%%%%%%%%%%
%%%%%%%%%%%%%%%%%%%%%%%%%%%%%%%%%%%%%%%%%%%%%%%%%%%%%%%%%%%%

\subsection{Unimodular matrices}\label{app:unimodular}

The following lemmas correspond to \cite[Lemma 2.2 and 8.7]{KLS}. The first one allows us to establish the upper bound in Lemma \ref{thm:comparison}. The proof of Proposition \ref{thm:decay-eigenfunctions} is given after the second one. At the end of the section, we state \cite[Theorem 8.3]{LS} which is used to prove pure point spectrum in the sub-critical regime.

%%%%%%%%%%%%%%%%%%%%%%%%%%%%%%%%%%%%%%%%%%%%%%%%%%%%%%%%%%%%
%%%%%%%%%%%%%%%%%%%%%%%%%%%%%%%%%%%%%%%%%%%%%%%%%%%%%%%%%%%%
\begin{lemma}
	Let $A$ be an unimodular matrix and let $\hat \theta = (\cos \theta,\, \sin \theta)$. Then, for all pair of angles $|\theta_1-\theta_2|\leq \frac{\pi}{2}$,
	\begin{eqnarray*}
		\| A \| \leq \sin \left( \tfrac{|\theta_1-\theta_2|}{2}\right)^{-1} \max \{ \| A \hat \theta_1\|,\, \| A\hat \theta_2\| \}.
	\end{eqnarray*}
\end{lemma}
\begin{proof}
	See \cite[Lemma 2.2]{KLS}.
\end{proof}
%\begin{proof}
%	First, there exists angles $\theta_0$ and $\sigma_0$ such that $A^* \hat \sigma_0 = \| A \| \hat \theta_0$. Then,
%	\begin{eqnarray}\nonumber
%		| \cos(\theta-\theta_0) | = |\langle \hat \theta, \hat \theta_0 \rangle |
%		= \frac{1}{\| A \|} | \langle\hat \theta, A^*\hat \sigma_0 \rangle |
%		= \frac{1}{\| A \|} | \langle A \hat \theta, \hat \sigma_0 \rangle |
%		\leq \frac{\| A \hat \theta \|}{\| A \|}.
%	\end{eqnarray}
%	Replacing $\theta$ by $\frac{\pi}{2}-\theta$, this becomes $\displaystyle \| A \| |\sin(\theta-\theta_0) | \leq \| A\hat \theta\|$. This way, for any pair of angles, we obtain
%	\begin{eqnarray}\nonumber
%		\| A \| \max\{ | \sin(\theta_1-\theta_0) |,\, | \sin(\theta_2-\theta_0) |  \} \leq \max\{ \| A\hat \theta_1\|,\, \| A\hat \theta_2\|\}
%	\end{eqnarray}
%	To conclude, just note that for all $|\gamma|\leq \pi/2$, the minimum of the function 
%	$$x\mapsto \max\{|\sin(x)|,\, |\sin(x+\gamma)|\}$$
%	 is attained at $\gamma/2$.
%\end{proof}
%%%%%%%%%%%%%%%%%%%%%%%%%%%%%%%%%%%%%%%%%%%%%%%%%%%%%%%%%%%%
%%%%%%%%%%%%%%%%%%%%%%%%%%%%%%%%%%%%%%%%%%%%%%%%%%%%%%%%%%%%
The following lemma is used to find eigenfunctions with the proper decay and is the key to Proposition \ref{thm:decay-eigenfunctions}. 

%%%%%%%%%%%%%%%%%%%%%%%%%%%%%%%%%%%%%%%%%%%%%%%%%%%%%%%%%%%%
%%%%%%%%%%%%%%%%%%%%%%%%%%%%%%%%%%%%%%%%%%%%%%%%%%%%%%%%%%%%
\begin{lemma}\label{thm:eigendirection}
	For a unimodular matrix with $\| A \| >1$, define $\vartheta=\vartheta(A)$ as the unique angle $\vartheta \in (-\frac{\pi}{2},\frac{\pi}{2}]$ such that $\| A \hat\vartheta \| = \| A \|^{-1}$. We also define $r(A)=\norm{A\begin{pmatrix}1\\0\end{pmatrix}}.\norm{A\begin{pmatrix}0\\1\end{pmatrix}}^{-1} $.
	\newline
	Let $(A_n)_n$ be a sequence of unimodular matrices with $\| A_n \| >1$ and write $\vartheta_n=\vartheta(A_n)$ and $r_n = r(A_n)$.
	Assume that
	\begin{enumerate}[label=\alph*.-]
		\item[(i)] $\displaystyle \lim_{n\to \infty} \| A_n\| = \infty$,
		
		\vspace{1ex}
		
		\item[(ii)] $\displaystyle \lim_{n\to\infty} \frac{\| A_{n+1}A_n^{-1}\|}{\| A_n \| \, \| A_{n+1}\|}=0$.
	\end{enumerate}
	Then,
	\begin{enumerate}
		\item $(\vartheta_n)_n$ has a limit $\vartheta_{\infty}\in (-\pi/2,\pi/2)$ if and only if $(r_n)_n$ has a limit $r_{\infty}\in[0,\infty)$. If $\vartheta_n\to\pm\pi/2$, then $r_n\to\infty$ but, if $r_n\to\infty$, we can only conclude that $|\vartheta_n|\to\pi/2$.
		
		\vspace{1ex}
		
		\item Suppose $(\vartheta_n)_n$ has a limit $\vartheta_{\infty}\neq 0,\, \frac{\pi}{2}$. Then,
		\begin{eqnarray}\label{eq:asymptotic-conditions}
			\lim_{n\to\infty} \frac{\log \| A_n \hat\vartheta_{\infty}\|}{\log \| A_n \|} = -1
			\quad
			\text{if and only if}
			\quad
			\limsup_n \frac{\log |r_n - r_{\infty}|}{\log \| A_n \|} \leq -2.
		\end{eqnarray}
	\end{enumerate}
\end{lemma}
%%%%%%%%%%%%%%%%%%%%%%%%%%%%%%%%%%%%%%%%%%%%%%%%%%%%%%%%%%%%
%%%%%%%%%%%%%%%%%%%%%%%%%%%%%%%%%%%%%%%%%%%%%%%%%%%%%%%%%%%%
\begin{proof}
	See \cite[Lemma 8.7]{KLS}.
\end{proof}
We apply this with $A_n = {\bf T}_{\omega,n}$ in order to prove Proposition \ref{thm:decay-eigenfunctions}. For positive sequences $(b_n)_n$ and $(c_n)_n$, we write $b_n \simeq c_n$ if $\displaystyle\lim_{n\to\infty} \frac{b_n}{c_n}=1$ and denote $b_n \simless c_n$ if there exists a constant $K>0$ such that $b_n \leq K c_n$ for $n$ large enough, and $b_n \asymp c_n$ if $b_n \simless c_n \simless b_n$.
%%%%%%%%%%%%%%%%%%%%%%%%%%%%%%%%%%%%%%%%%%%%%%%%%%%%%%%%%%%%
%%%%%%%%%%%%%%%%%%%%%%%%%%%%%%%%%%%%%%%%%%%%%%%%%%%%%%%%%%%%
\begin{proof}[Proof of Proposition \ref{thm:decay-eigenfunctions}]
Define
\begin{eqnarray*}
	\Phi^{(1)}_n
	=
	\begin{pmatrix}
 		\vp^{(1)}_{+,n}
 		\\
 		\vp^{(1)}_{-,n}
 	\end{pmatrix}
	=
	{\bf T}_{\omega,n-1}
	\begin{pmatrix}
		1 \\ 0
	\end{pmatrix},
	\quad 
		\quad
	\Phi^{(2)}_n
	=
	\begin{pmatrix}
 		\vp^{(2)}_{+,n}
 		\\
 		\vp^{(2)}_{-,n}
 	\end{pmatrix}
	=
	{\bf T}_{\omega,n-1}
	\begin{pmatrix}
		0 \\ 1
	\end{pmatrix},
\end{eqnarray*}
and let $R^{(i)}_n,\, \theta^{(i)}_n,\, n\geq 1$, $i=1,\, 2$ be the corresponding Pr\"ufer radii and phases.
 %for initial conditions $R^{(1)}_1=R^{(2)}_1=1$, $\theta^{(1)}_1=0$, $\theta^{(2)}_1=\pi/2$, and let $x^{(i)},\, i=1,2$ be their expression in the original coordinates. 
 We let $r_n=\frac{R^{(1)}_n}{R^{(2)}_n}$ and $\vartheta_n$ be as in Lemma \ref{thm:eigendirection}. Recall the relation $\Phi_n=\ppru_n \Psi_n$. In particular,
 \begin{eqnarray*}
 	\Phi_n^{(i)}
 	=
 	\begin{pmatrix}
 		\vp^{(i)}_{+,n}
 		\\
 		\vp^{(i)}_{-,n}
 	\end{pmatrix}
 	=
 	(-1)^{n-1}
 	R_n
 	\begin{pmatrix}
 		-\sqrt{p_2}\cos(\bar{\theta}_n^{(i)})
 		\\
 		\sqrt{-p_1}\cos(\bar{\theta}_n^{(i)}+k)
 	\end{pmatrix}.
 \end{eqnarray*}
Thus it follows from some elementary trigonometry that
\begin{eqnarray*}
	\vp^{(1)}_{+,n}\vp^{(2)}_{-,n}-\vp^{(1)}_{-,n}\vp^{(2)}_{+,n}
	=
	R^{(1)}_n R^{(2)}_n  \sin(2k) \sin(\theta^{(1)}_n-\theta^{(2)}_n),
%	\\
%	x^{(1)}_{n+1}x^{(2)}_n-x^{(1)}_nx^{(2)}_{n+1}=R^{(1)}_n R^{(2)}_n  \sin k \sin(\theta^{(2)}_n-\theta^{(1)}_n).
\end{eqnarray*}
where $\bar{\theta}^{(i)}_n = \theta^{(i)}_n-(2n-1)k$.
On the other hand,
\begin{eqnarray*}
	\vp^{(1)}_{+,n}\vp^{(2)}_{-,n}-\vp^{(1)}_{-,n}\vp^{(2)}_{+,n}
	&=&
%	\det \left( 
%	{\bf T}_n
%	\begin{pmatrix}
%		1 \\ 0
%\end{pmatrix}		
%	\,
%	{\bf T}_n
%	\begin{pmatrix}
%		0 \\ 1
%	\end{pmatrix}		
%	\right)\\
%	&=&
	\det\left(
	{\bf T}_{\omega,n}
	\begin{pmatrix}
		1 & 0 \\ 0 & 1
	\end{pmatrix}
	\right)	
	=1.
\end{eqnarray*}
This, together with the convergence
\begin{eqnarray*}
	\lim_{n\to\infty} \frac{R^{(i)}_n}{\log n} = \beta,\quad i=1,2,
\end{eqnarray*}
gives
\begin{eqnarray}\label{eq:limit-sines}
	\lim_{n\to\infty}\frac{\log |\sin(\theta^{(2)}_n-\theta^{(1)}_n)|}{\log n} 
	=\lim_{n\to\infty}\frac{\log |\sin(\bar\theta^{(2)}_n-\bar\theta^{(1)}_n)|}{\log n} 
	= -2\beta.
\end{eqnarray}
%As $r_n=R^{(1)}_n/R^{(2)}_n$,
%\begin{eqnarray}\label{eq:decomposition-radius-0}
%	r_n 
%	&=& 
%	\prod^n_{j=1} \left( 1 - \frac{V_j}{\sin k} \left(\sin(2\bar{\theta}^{(1)}_j) - \sin(2\bar{\theta}^{(2)}_j) \right)
%	+
%	\frac{V_j^2}{\sin^2 k} \left(\cos^2(\bar{\theta}^{(1)}_j) - \cos^2(\bar{\theta}^{(2)}_j) \right) \right),
%\end{eqnarray}
Remember the decomposition \eqref{eq:prufer-transform-s1-radii} that we summarize as
\begin{eqnarray*}
	(R^{(i)}_{n+1})^2
	&=&
	\Big{(}
	1
	-\frac{p_2}{\sin(2k)}\sin(2\bar\theta_n^{(i)})\vo(n)
	\\
	&&
	\phantom{blablablablabla}
	+\frac{p_1}{\sin(2k)}\sin(2(\bar\theta^{(i)}_n-k))\vt(n+1)
	+E_j^{(i)}
	\Big{)}
	(R^{(i)}_{n})^2.
\end{eqnarray*}
We have to estimate the difference of the expansions for
$\log R^{(1)}_n$ and $\log R^{(2)}_n$.
By \eqref{eq:limit-sines}, one has $|\sin(\bar\theta^{(2)}_n-\bar\theta^{(1)}_n)| \lesssim n^{-\beta+\epsilon}$, for any $\epsilon>0$.
Hence, there exist random sequences $(m_n)_n\subset\n^*$ and $(\Delta_n)_n\subset \R$ such that $\bar\theta^{(1)}_n-\bar\theta^{(2)}_n=m_n\pi + \Delta_n$ and $|\Delta_n|\lesssim n^{-\beta+\epsilon}$. Therefore,
\begin{equation*}
	\sin(2\bar\theta^{(2)}_n)
	=
	\sin(2\bar\theta^{(1)}_n+2\Delta_n)
	\simeq
	\sin(2\bar\theta^{(1)}_n)+2\cos(2\bar\theta^{(1)}_n)\Delta_n.
\end{equation*}
This shows that 
\begin{eqnarray*}
	\left|\vo(j) \left(\sin(2\bar{\theta}^{(1)}_j) - \sin(2\bar{\theta}^{(2)}_j)\right)\right| 
	&\lesssim& j^{-\frac{1}{2}-2\beta + \epsilon}.
\end{eqnarray*}
By means of similar arguments, one can show that
\begin{eqnarray*}
	\left|\vt(j+1) \left(\sin(2(\bar{\theta}^{(1)}_j-k)) - \sin(2(\bar{\theta}^{(2)}_j-k))\right)\right| 
	&\lesssim& j^{-\frac{1}{2}-2\beta + \epsilon}.
\end{eqnarray*}
and
\begin{eqnarray*}
	|E_j^{(1)}-E_j^{(2)}|
	&\lesssim&
	 j^{-1-2\beta + \epsilon}.
\end{eqnarray*}
Hence, 
\begin{eqnarray}
	\log r_n 
	&=& 
	-\frac{p_2}{\sin(2k)}\sum^n_{j=1} \frac{\vo(j)}{\sin k} \left(\sin(2\bar{\theta}^{(1)}_j) - \sin(2\bar{\theta}^{(2)}_j) \right)
	\\
	\label{eq:decomposition-radius}
	&&
	+
	\frac{p_1}{\sin(2k)}\sum^n_{j=1} \frac{\vt(j+1)}{\sin k} \left(\sin(2(\bar{\theta}^{(1)}_j-k)) - \sin(2(\bar{\theta}^{(2)}_j-k)) \right)
	+
	\sum^n_{j=1} A_j
\end{eqnarray}
 where the first two sums are convergent martingales by Lemma \ref{thm:martingales} with $\gamma=\frac12 + 2\beta -\epsilon$ and the last one is absolutely convergent as $A_j=O(j^{-1-2\beta+\epsilon})$. This shows that $r_n\to r_{\infty}\in(0,\infty)$ almost surely which implies that $\vartheta_n$ has a limit $\vartheta_{\infty}\neq 0,\, \frac{\pi}{2}$ by the first part of Lemma \ref{thm:eigendirection}. 

%\theta^{()} x^{()} R^{()} \theta^{()}
The equivalence \eqref{eq:asymptotic-conditions} in our context corresponds to
\begin{eqnarray*}
	\lim_{n\to\infty}\frac{\log R_n( \vartheta_{\infty})}{\log n}=-\beta
	\quad 
	\text{if and only if}
	\quad
	\limsup_n \frac{\log |r_n-r_{\infty}|}{\log n} \leq -2\beta.
\end{eqnarray*}
Let us denote by $\log r_n = M_n + S_n$ the decomposition \eqref{eq:decomposition-radius} and $\log r_{\infty}=M_{\infty}+S_{\infty}$, where $M_n$ is the martingale part and $M_{\infty}$ and $S_{\infty}$ are the almost sure limits of $M_n$ and $S_n$ respectively. Then,
\begin{eqnarray*}
	|r_{\infty}-r_n|
%	&=&
%	e^{M_{\infty}+A_{\infty}}-e^{M_{n}+A_{n}}
	&=&
	\e^{M_{\infty}+S_{\infty}}\left| 1-\e^{M_n-M_{\infty}+S_n-S_{\infty}}\right|
	\simeq
	\e^{M_{\infty}+S_{\infty}}\left| M_n-M_{\infty}+S_n-S_{\infty}\right|\\
	&\lesssim&
	\e^{M_{\infty}+S_{\infty}} n^{-2\beta+2\epsilon},
\end{eqnarray*}
by the last statement of Lemma \ref{thm:martingales} with $\gamma = \frac12 + 2\beta - \epsilon$ and $A_j=O(j^{-1-2\beta+\epsilon})$.
This finishes the proof of Proposition \ref{thm:decay-eigenfunctions}.
\end{proof}
%%%%%%%%%%%%%%%%%%%%%%%%%%%%%%%%%%%%%%%%%%%%%%%%%%%%%%%%%%%%
%%%%%%%%%%%%%%%%%%%%%%%%%%%%%%%%%%%%%%%%%%%%%%%%%%%%%%%%%%%%
We state \cite[Theorem 8.3]{LS} which allowed us to prove pure point spectrum in the sub-critical regime:
%%%%%%%%%%%%%%%%%%%%%%%%%%%%%%%%%%%%%%%%%%%%%%%%%%%%%%%%%%%%
%%%%%%%%%%%%%%%%%%%%%%%%%%%%%%%%%%%%%%%%%%%%%%%%%%%%%%%%%%%%
\begin{theorem}\label{thm:LS-osc}
	Let $(A_n)_{n\geq 1}$ be $2\times2$ real unimodular matrices and let ${\bf A}_n=A_n \cdots A_1$ such that
	\begin{eqnarray*}
		\sum_{n\geq 1} \frac{\| A_{n+1}\|}{\| {\bf A}_n \|}<\infty.
	\end{eqnarray*}
	Suppose there exists a monotone increasing function $g:\n^* \to (0,\infty)$ such that
	\begin{eqnarray*}
		\lim_{n\to\infty} \frac{\log \| A_n \|}{g(n)}=0
		\quad
		\text{and}
		\quad
		\lim_{n\to\infty} \frac{\log \| {\bf A}_n \|}{f(n)}=1,
	\end{eqnarray*}
	and such that
	\begin{eqnarray*}
		\sum_{n\geq 1} \e^{-\epsilon g(n)}<\infty,
	\end{eqnarray*}
	for all $\epsilon>0$. Then, there exists an angle $\vartheta_0$ such that
	\begin{eqnarray*}
		\lim_{n\to\infty} \frac{\log \| {\bf A}_n \widehat{\vartheta}_0 \|}{g(n)}=-1.
	\end{eqnarray*}
\end{theorem}
%%%%%%%%%%%%%%%%%%%%%%%%%%%%%%%%%%%%%%%%%%%%%%%%%%%%%%%%%%%%
%%%%%%%%%%%%%%%%%%%%%%%%%%%%%%%%%%%%%%%%%%%%%%%%%%%%%%%%%%%%

%%%%%%%%%%%%%%%%%%%%%%%%%%%%%%%%%%%%%%%%%%%%%%%%%%%%%%%%%%%%
%%%%%%%%%%%%%%%%%%%%%%%%%%%%%%%%%%%%%%%%%%%%%%%%%%%%%%%%%%%%
%%%%%%%%%%%%%%%%%%%%%%%%%%%%%%%%%%%%%%%%%%%%%%%%%%%%%%%%%%%%
%%%%%%%%%%%%%%%%%%%%%%%%%%%%%%%%%%%%%%%%%%%%%%%%%%%%%%%%%%%%

\subsection{A martingale inequality}\label{sec:martingales}

%%%%%%%%%%%%%%%%%%%%%%%%%%%%%%%%%%%%%%%%%%%%%%%%%%%%%%%%%%%%
%%%%%%%%%%%%%%%%%%%%%%%%%%%%%%%%%%%%%%%%%%%%%%%%%%%%%%%%%%%%
The following corresponds to \cite[Lemma 8.4]{KLS}. We formulate it in full generality but provide a short proof under the assumption that $V_{\omega,i}(n)=\lambda n^{-\alpha}\omega_{n,i}$ for uniformly bounded random variables $\omega_{n,i}$. 
%%%%%%%%%%%%%%%%%%%%%%%%%%%%%%%%%%%%%%%%%%%%%%%%%%%%%%%%%%%%
%%%%%%%%%%%%%%%%%%%%%%%%%%%%%%%%%%%%%%%%%%%%%%%%%%%%%%%%%%%%
\begin{lemma}\label{thm:martingales}
	Let $(Z_j)_j$ be i.i.d. random variables with $\esp[Z_n]=0$ and $\esp[|Z_n|^2]\leq n^{-2\gamma}$ for some $\gamma>0$. Let $\mathcal{G}_n=\sigma(Z_1,\cdots, \, Z_n)$ and let $Y_n \in \mathcal{G}_{n-1}$ for $n\geq 1$ such that $|Y_n|\leq 1$.
	Define
	\begin{eqnarray*}
		M_n = \sum^n_{j=1} Y_jZ_j \quad \text{and} \quad s_n = \sum^n_{j=1} \frac{1}{j^{2\gamma}}.
	\end{eqnarray*}
	Then, $(M_n)_n$ is a $\mathcal{G}_n$-martingale and 
	\begin{enumerate}[label=\alph*.-]
		\item[(i)] For $\gamma \leq \frac12$ and all $\varepsilon>0$,
			\begin{eqnarray*}
				\lim_{n\to\infty} s_n^{-\frac{1+\varepsilon}{2}}M_n= 0,\quad \quad \p-a.s.
			\end{eqnarray*}
	
		\vspace{1ex}
		
		\item[(ii)] For $\gamma>\frac12$, $(M_n)_n$ converges $\p$-almost surely to a finite (random) limit $M_{\infty}$ and, for all $\kappa<\gamma - \frac12$, we have
		\begin{eqnarray*}
			\lim_{n\to\infty} n^{\kappa} \left( M_{\infty}-M_n\right) = 0,\quad \p-a.s.
		\end{eqnarray*}
	\end{enumerate}		
\end{lemma}
%%%%%%%%%%%%%%%%%%%%%%%%%%%%%%%%%%%%%%%%%%%%%%%%%%%%%%%%%%%%
%%%%%%%%%%%%%%%%%%%%%%%%%%%%%%%%%%%%%%%%%%%%%%%%%%%%%%%%%%%%
\begin{proof}
	The reader can consult the book \cite{Durrett} for the general properties of martingales used below. 
	The sequence $(M_n)_n$ is a martingale thanks to our hypothesis on $(Y_n)_n$ and $(Z_n)_n$: indeed, since $M_n,\, Y_{n+1} \in \mathcal{G}_n$, $Y_{n+1}$ is bounded and $Z_{n+1}$ is independent of $\mathcal{G}_n$ and centered, we have, $\p$-almost surely,
	\begin{eqnarray*}
		\esp[M_{n+1} | \mathcal{G}_n]
		&=&
		\esp\left[Y_{n+1}Z_{n+1}  + M_{n} \Big{|}\mathcal{G}_n\right]\\
		&=&
		Y_{n+1}\esp[Z_{n+1}] + M_n = M_n.
\end{eqnarray*}		
	As stated above, we assume $Z_n = n^{-\gamma} X_n$ with $|X_n|\leq 1$ and $\esp[X_n]=0$ to simplify the argument.
	Let $\gamma\leq \frac12$. We use Azuma's inequality \cite{Azuma}: let $(M_n)_n$ be a martingale such that $|M_n-M_{n-1}|\leq c_n$ for all $n\geq 1$. Then,
	\begin{eqnarray*}
		\p\left[ |M_n-M_0| \geq t\right] \leq 2 \exp\left\{ - \frac{t^2}{2 \sum^n_{j=1} c_j^2}\right\}.
	\end{eqnarray*}
	In our case, $M_0=0$, $c_j = 2 j^{-\gamma}$, and taking $t=s_n^{\frac{1+\varepsilon'}{2}}$ for $0<\varepsilon'<\varepsilon$, we obtain
	\begin{eqnarray*}
		\p\left[ |M_n| \geq s_n^{\frac{1+\varepsilon'}{2}}\right] \leq 2\ \e^{ - C n^{\varepsilon'}},
	\end{eqnarray*}
	for some $C>0$. The claim $(i)$ then follows from Borel-Cantelli's lemma.
	
	Now, let $\gamma>\frac12$. Noticing that, for $i<l$,
%	\begin{eqnarray}
%		\esp[X_kY_k X_l Y_l] 
%		= \esp[\esp[X_kY_k X_l Y_l| \mathcal{F}_{l-1}]]
%		= \esp[X_kY_k X_l \esp[Y_l| \mathcal{F}_{l-1}]]
%		=0,
%	\end{eqnarray}
	\begin{eqnarray*}
		\esp[X_iY_i X_l Y_l] 
		=
		\esp[X_l]\esp[X_iY_i Y_l] 
		=0,
	\end{eqnarray*}
	we have
	\begin{eqnarray*}
		\sup_n\esp[M_n^2] = \sup_n \sum^n_{j=1} \frac{\esp[X_j^2 Y_j^2]}{j^{2\gamma}} \leq \sum_{j\geq 1} \frac{1}{j^{2\gamma}}<\infty.
	\end{eqnarray*}
	Hence, $(M_n)_n$ is bounded in $L^2$ and, as a consequence, converges almost surely, i.e., there exists a random variable $M_{\infty}$ such that $\lim_{n\to\infty}M_n=M_{\infty}$, $\p$-a.s.. Finally, applying Azuma's inequality to the martingale $(M_{n+i}-M_n)_{i\geq0}$, we obtain
	\begin{eqnarray*}
		\p\left[ n^{\kappa}\left| M_{n+i}-M_n\right| \geq 1\right]\leq 2 \exp\left\{ - C n^{2(\gamma-\frac12 -\kappa)}\right\},
	\end{eqnarray*}
	for all $i\geq0$.
	Choosing $\kappa<\gamma-\frac12$, the last claim follows from Fatou's lemma, the convergence of $(M_n)_n$ and Borel-Cantelli.
\end{proof}
%%%%%%%%%%%%%%%%%%%%%%%%%%%%%%%%%%%%%%%%%%%%%%%%%%%%%%%%%%%%
%%%%%%%%%%%%%%%%%%%%%%%%%%%%%%%%%%%%%%%%%%%%%%%%%%%%%%%%%%%%
\begin{remark}
	The result above is proved in \cite[Lemma 8.4]{KLS} under the second moment assumption replacing our use of Azuma's inequality by Doob's inequality.
	For a short proof assuming bounded exponential moments, see \cite[Lemma A.1]{CY}.
\end{remark}
%%%%%%%%%%%%%%%%%%%%%%%%%%%%%%%%%%%%%%%%%%%%%%%%%%%%%%%%%%%%
%%%%%%%%%%%%%%%%%%%%%%%%%%%%%%%%%%%%%%%%%%%%%%%%%%%%%%%%%%%%
%%%%%%%%%%%%%%%%%%%%%%%%%%%%%%%%%%%%%%%%%%%%%%%%%%%%%%%%%%%%
%%%%%%%%%%%%%%%%%%%%%%%%%%%%%%%%%%%%%%%%%%%%%%%%%%%%%%%%%%%%

\subsection{Control of the phases}\label{sec:control-phases}

%%%%%%%%%%%%%%%%%%%%%%%%%%%%%%%%%%%%%%%%%%%%%%%%%%%%%%%%%%%%
%%%%%%%%%%%%%%%%%%%%%%%%%%%%%%%%%%%%%%%%%%%%%%%%%%%%%%%%%%%%

%%%%%%%%%%%%%%%%%%%%%%%%%%%%%%%%%%%%%%%%%%%%%%%%%%%%%%%%%%%%
%%%%%%%%%%%%%%%%%%%%%%%%%%%%%%%%%%%%%%%%%%%%%%%%%%%%%%%%%%%%

The next lemma provides the control of the Pr\"ufer phases needed to complete the proof of Proposition \ref{thm:lyapunov-exponents}. The strategy is taken from \cite{KLS}. Recalling the definitions of $Q_{n,1}$ and $Q_{n,2}$ from \eqref{eq:sum-of-phases},
%%%%%%%%%%%%%%%%%%%%%%%%%%%%%%%%%%%%%%%%%%%%%%%%%%%%%%%%%%%%
%%%%%%%%%%%%%%%%%%%%%%%%%%%%%%%%%%%%%%%%%%%%%%%%%%%%%%%%%%%%
\begin{lemma}\label{thm:control-phases}
Assume \textbf{(A3a)} and \textbf{(A4)}.
Let $0<\alpha\leq \frac12$. For each fixed energy corresponding to a value of $k \in (-\pi, -\tfrac{\pi}{2})$ different from $-\frac{5\pi}{8},-\frac{3\pi}{4}$ and $-\frac{7\pi}{8}$,
	\begin{equation*}
		\lim_{n\to\infty} \frac{Q_{n,i}}{\sum^n_{j=1} j^{-2\alpha}}
		=
		0,
	\end{equation*}
		for $i=1,2$. Moreover, for each compact energy interval $I\subset \mathring\Sigma$, the convergence is uniform over all initial values $\theta_0 \in [0,2\pi)$ and $E\in I$ corresponding to values of $k$ different from $-\frac{5\pi}{8},-\frac{3\pi}{4}$ and $-\frac{7\pi}{8}$.
\end{lemma}
%%%%%%%%%%%%%%%%%%%%%%%%%%%%%%%%%%%%%%%%%%%%%%%%%%%%%%%%%%%%
%%%%%%%%%%%%%%%%%%%%%%%%%%%%%%%%%%%%%%%%%%%%%%%%%%%%%%%%%%%%
\begin{proof}
	We will show that
	\begin{eqnarray*}
	\limsup_{n\to\infty}
	\frac{\sum^n_{j=1}\esp[ V_{\omega,j}^2]\cos 4 \bar \theta_j}{\sum^n_{j=1} j^{-2\alpha}}
	=
	0,
	\qquad \p-\text{a.s.,}
\end{eqnarray*}
the other terms being handled similarly. Note that the prefactor accompanying this term in the definition of $Q_{n,1}$ is uniformly bounded over compact energy intervals. The computations below are uniform in the initial condition $\theta_0$ and only assume $k\notin \tfrac{\pi}{8}\Z$.
%%%%%%%%%%%%%%%%%%%%%%%%%%%%%%%%%%%%%%%%%%%%%%%%%%%%%%%%%%%%
%%%%%%%%%%%%%%%%%%%%%%%%%%%%%%%%%%%%%%%%%%%%%%%%%%%%%%%%%%%%

We begin with a simple observation: from \eqref{eq:prufer-transform-s1}, for any compact interval $I\subset \mathring\Sigma$, there exists a constant $C=C(I)\in(0,\infty)$ such that
\begin{eqnarray*}
	&&
	\left| \e^{i(\theta_{n+1}-\theta_n)}-1\right| 
	=
	\left| \frac{\zeta_{n+1}}{\zeta_n}-1\right|
	\\
	&&
	\phantom{blabla}
	\leq 
	C \left( 
		|\vo(n)|+|\vt(n+1)|+|\vo(n)\vt(n+1)|
	\right)
	\leq 1,
\end{eqnarray*}
for $n\geq n^*(\omega)$ for some $n^*(\omega)=n^*(\omega,I)<\infty$ thanks to \textbf{(A4)}. Hence, for $n\geq n^*(\omega)$, $|\theta_{n+1}-\theta_n|< \frac{\pi}{2}$ and, recalling \textbf{(A4)} once more, we have
\begin{eqnarray*}
	|\theta_{n+1}-\theta_n| \leq \frac{\pi}{2} |\sin(\theta_{n+1}-\theta_n)| 
	\leq 
	\frac{\pi}{2} \left|\e^{i(\theta_{n+1}-\theta_n)}-1\right| 
	\leq 
	c_0(\omega)\ n^{-\frac{2\alpha}{3}},
\end{eqnarray*}
for some $c_0(\omega)=c_0(\omega,I)\in(0,\infty)$.
This can be written in the equivalent form 
\begin{eqnarray}\label{eq:small-phases}
	|\bar{\theta}_{n+1}-\bar{\theta}_n+2k|\leq c_0(\omega) \ n^{-\frac{2\alpha}{3}},
\end{eqnarray}
which will be more suitable for our purposes. By possibly increasing the value of $c_0(\omega)$, we can assume that \eqref{eq:small-phases} holds for all $n\geq n^*(\omega)$ with $c_0(\omega)<\infty$ $\p$-amost surely. For $p\geq 1$, define
\begin{eqnarray*}
	\mathcal{E}_p = \left\{ \omega:\, c_0(\omega) \leq p\right\},
\end{eqnarray*}
and observe that $\Omega=\displaystyle\bigcup_{p\geq 1} \mathcal{E}_p$.
%%%%%%%%%%%%%%%%%%%%%%%%%%%%%%%%%%%%%%%%%%%%%%%%%%%%%%%%%%%%
%%%%%%%%%%%%%%%%%%%%%%%%%%%%%%%%%%%%%%%%%%%%%%%%%%%%%%%%%%%%

The key proof is \cite[Lemma 8.5]{KLS} which states the following: suppose that $y\in\R$ is not in $\pi \Z$. Then, there exists a sequence of integers $q_l \to \infty$ such that
	\begin{eqnarray}\label{eq:trigonometrics}
		\left|\sum^{q_l}_{j=1} \cos \theta_j \right| \leq 1 + \sum^{q_l}_{j=1} \left| \theta_j-\theta_0 - jy \right|,
	\end{eqnarray}
	for all $(\theta_j)_{j\geq 0}\subset \R$.
%%%%%%%%%%%%%%%%%%%%%%%%%%%%%%%%%%%%%%%%%%%%%%%%%%%%%%%%%%%%
%%%%%%%%%%%%%%%%%%%%%%%%%%%%%%%%%%%%%%%%%%%%%%%%%%%%%%%%%%%%
	We take $y=-8k$. Let $p\geq 1$ and $\omega\in\mathcal{E}_p$. 
	Let $n$ be large enough so that it can be written as 	$n=n_0+Kq_l$ with $n_0\ge q_l^2$ and $4c_0(\omega) n_0^{-\alpha}\le q_l^{-2}$. Then,
\begin{eqnarray*}
	 \left|\sum_{j=n_0+1}^n j^{-2\alpha}\cos(4\bar\theta_j)\right|
	 &=&
	 \left|\sum_{m=0}^K \sum_{r=1}^{q_l}(n_0+mq_l+r)^{-2\alpha}\cos(4\bar\theta(n_0+mq_l+r))\right|
	 \\
	&\leq&	 
	\sum_{m=0}^K(n_0+mq_l)^{-2\alpha} \left|\sum_{r=1}^{q_l}\cos(4\bar\theta(n_0+mq_l+r))\right|
	\\
	&&
	+
	\sum_{m=0}^K \sum_{r=1}^{q_l} \left|(n_0+mq_l+r)^{-2\alpha}-(n_0+mq_l)^{-2\alpha}\right|
	\\
	&=:&
	A + B.
\end{eqnarray*}
We first estimae the term $A$ using \eqref{eq:trigonometrics} to get
\begin{eqnarray*}
	A
	&\leq&
	\sum_{m=0}^{K}(n_0+mq_l)^{-2\alpha}\left(1+4\sum_{r=1}^{q_l}|\bar\theta(n_0+mq_l+r)-\bar\theta(n_0+mq_l)+2kr|\right).
\end{eqnarray*}
Now, by \eqref{eq:small-phases} it follows that
\begin{eqnarray*}
	&&
	4\sum_{r=1}^{q_l}|\bar\theta(n_0+mq_l+r)-\bar\theta(n_0+mq_l)+2kr|
	\leq
	c_0 \sum_{r=1}^{q_l}\sum_{s=1}^r(n_0+mq_l+r)^{-\frac{2\alpha}{3}}
	\\
	&&
	\phantom{blablablabla}
	\leq
	c_0(\omega)(n_0+mq_l)^{-\frac{2\alpha}{3}}\sum_{r=1}^{q_l} r
	\leq
	c_0(\omega)q_l^2 n_0^{-\frac{2\alpha}{3}}
	\leq
	1.
\end{eqnarray*}
Thus,
\begin{eqnarray*}
	A
	\leq
	2
	\sum_{m=0}^{K}(n_0+mq_l)^{-2\alpha}
	\leq
	2 q_l^{-2\alpha} \sum_{m=0}^{K}(n_0q_l^{-1}+m)^{-2\alpha}
	\leq 
	c_1 q_l^{-2\alpha} \sum_{j=1}^{K}j^{-2\alpha},
%	c_1 q_l^{-2\alpha} K^{1-2\alpha}
%	\leq
%	c_1 q_l^{-1} n^{1-2\alpha},
\end{eqnarray*}
for some finite $c_1>0$. To estimate $B$, we use that
\begin{eqnarray*}
	\left| (n_0+mq_l+r)^{-2\alpha}-(n_0+mq_l)^{-2\alpha}\right|
	\leq
	c_2 (n_0+mq_l)^{-2\alpha-1} r,
\end{eqnarray*}
for some finite $c_2>0$ which allows us to write
\begin{eqnarray*}
	B
	&\leq&
	c_2 \sum_{m=0}^K \sum_{r=1}^{q_l} (n_0+mq_l)^{-2\alpha-1} r
	\leq 
	c_2 q_l^2 n_0^{-1}
	\sum_{m=0}^K
	(1+n_0^{-1}mq_l)^{-1}(n_0+mq_l)^{-2\alpha}
	\\
	&\leq&
	c_2
	\sum_{m=0}^K
	(n_0+mq_l)^{-2\alpha},
\end{eqnarray*}
where we used $q_l^2 n_0^{-1}\leq 1$. This last sum can be estimated as above. 
Combining, we obtain
\begin{eqnarray*}
	\left|\sum^{n}_{j=1} j^{-2\alpha} \cos 4 \bar{\theta}_j \right|
	\leq
	\sum^{n_0}_{j=1}j^{-2\alpha} + c_3 q_l^{-2\alpha} \sum_{j=1}^{K}j^{-2\alpha},
\end{eqnarray*}
for some finite $c_3>0$ and all $\omega\in\mathcal{E}_p$. Hence,
\begin{eqnarray*}
	\limsup_{n\to\infty}\frac{\left|\sum^{n}_{j=1} j^{-2\alpha} \cos 4 \bar{\theta}_j \right|}{\sum^n_{j=1}j^{-2\alpha}}
	\leq c_3 q_l^{-2\alpha},
\end{eqnarray*}
for all $\omega\in\mathcal{E}_p$.
We can then let $l\to\infty$.
As the events $\mathcal{E}_p$ exhaust $\Omega$, this finishes the proof.
\end{proof}
%%%%%%%%%%%%%%%%%%%%%%%%%%%%%%%%%%%%%%%%%%%%%%%%%%%%%%%%%%%%
%%%%%%%%%%%%%%%%%%%%%%%%%%%%%%%%%%%%%%%%%%%%%%%%%%%%%%%%%%%%
%%%%%%%%%%%%%%%%%%%%%%%%%%%%%%%%%%%%%%%%%%%%%%%%%%%%%%%%%%%%
%%%%%%%%%%%%%%%%%%%%%%%%%%%%%%%%%%%%%%%%%%%%%%%%%%%%%%%%%%%%
The next lemma provides the control of the phases needed to complete the proof of Proposition \ref{thm:integrated-lyapunov}.
%%%%%%%%%%%%%%%%%%%%%%%%%%%%%%%%%%%%%%%%%%%%%%%%%%%%%%%%%%%%
%%%%%%%%%%%%%%%%%%%%%%%%%%%%%%%%%%%%%%%%%%%%%%%%%%%%%%%%%%%%
\begin{lemma}\label{thm:control-phases-integrated}
Let $0<\alpha\leq \frac12$.
Assume \textbf{(A1)}-\textbf{(A3a)} and \textbf{(A5)}.
 For each fixed energy corresponding to a value of $k \in (-\pi, -\tfrac{\pi}{2})$ different from $-\frac{5\pi}{8},-\frac{3\pi}{4}$ and $-\frac{7\pi}{8}$,
	\begin{eqnarray*}
		\lim_{n\to\infty} \frac{\esp[Q_{n,i}]}{\sum^n_{j=1} j^{-2\alpha}}
		=
		0,
	\end{eqnarray*}
	for $i=1,2$. Moreover, for each compact energy interval $I\subset \mathring\Sigma$, the convergence is uniform over all initial values $\theta_0 \in [0,2\pi)$ and $E\in I$ corresponding to values of $k$ different from $-\frac{5\pi}{8},-\frac{3\pi}{4}$ and $-\frac{7\pi}{8}$.
\end{lemma}
%%%%%%%%%%%%%%%%%%%%%%%%%%%%%%%%%%%%%%%%%%%%%%%%%%%%%%%%%%%%
%%%%%%%%%%%%%%%%%%%%%%%%%%%%%%%%%%%%%%%%%%%%%%%%%%%%%%%%%%%%
\begin{proof}
	By Borel-Cantelli ,
		\begin{eqnarray*}
			|V_{\omega,i}(n)| \leq n^{-\frac{2\alpha}{3}-\frac{\varepsilon}{2}},
		\end{eqnarray*}
	for all $n\geq \tau$, $i=1,2$ for some $\p$-almost surely finite $\tau=\tau(\omega)$. Thanks to the uniform control of the previous lemma, we have
	\begin{eqnarray*}
		\lim_{n\to\infty}
		\frac{\esp\left[\left|\sum^{n}_{j=\tau(\omega)} j^{-2\alpha} \cos 4 \bar{\theta}_j \right|\right]}{\sum^n_{j=1}j^{-2\alpha}}
		=
		0.
	\end{eqnarray*} 
	It is then enough to show that
	\begin{eqnarray*}
		\esp\left[ \sum^{\tau(\omega)}_{j=1} j^{-2\alpha}\right]
		<
		\infty.
	\end{eqnarray*}
	If $\alpha \neq \frac12$, we have
	\begin{eqnarray*}
		\esp\left[ \sum^{\tau(\omega)}_{j=1} j^{-2\alpha}\right]
		&=&
		\sum_{k\geq 1} 
		\left(		
			\sum^k_{j=1} j^{-2\alpha}
		\right)
		\p[\tau(\omega)=k]
		\\
		&\leq&
		C_1
		\sum_{k\geq 1} k^{1-2\alpha}
		\p[\tau(\omega)=k],
	\end{eqnarray*}
	for some finite $C_1>0$. We can estimate the probability inside the sum:
	\begin{eqnarray*}
		\p[\tau(\omega)=k]
		&\leq&
		\p\left[|V_j| > j^{-\frac{2\alpha}{3}-\frac{\varepsilon}{2}},\, \forall \, j<k\right]
		\\
		&=&
		\prod^{k-1}_{j=1} \p\left[|V_j| > j^{-\frac{2\alpha}{3}-\frac{\varepsilon}{2}}\right]
		\leq
		\prod^{k-1}_{j=1} 
		j^{(\frac{2\alpha}{3}+\frac{\varepsilon}{2})p}
		\esp\left[|V_j|^p\right]
		\\
		&\leq&
		C_2^k \prod^{k-1}_{j=1} j^{-\varepsilon p}
		\leq
		C_2^k \e^{-c \varepsilon p k \log k},
	\end{eqnarray*}
	for some finite $C_1>0$ and $c>0$. Hence,
	\begin{eqnarray*}
		\esp\left[ \sum^{\tau(\omega)}_{j=1} j^{-2\alpha}\right]
		\leq
		C_1
		\sum_{k\geq 1} k^{1-2\alpha} C_2^k \e^{-c \varepsilon p k \log k}
		<
		\infty.
	\end{eqnarray*}
	The case $\alpha=\frac12$ is similar. All the above estimates hold uniformly in $E\in I$ corresponding to values of $k$ different from $-\frac{5\pi}{8},-\frac{3\pi}{4}$ and $-\frac{7\pi}{8}$.
\end{proof}

\end{document}